\documentclass[sigconf]{acmart}

\usepackage{hyperref}           
\hypersetup{
    colorlinks=true,
    linkcolor=blue,
    filecolor=red,      
    urlcolor=magenta,
    breaklinks=true,            
}
\usepackage{breakurl}           

\usepackage{todonotes}
\usepackage{shuffle}
\usepackage[capitalise]{cleveref}
\usetikzlibrary{decorations.pathreplacing,calligraphy}

\newenvironment{sketch}{%
  \renewcommand{\proofname}{Proof Sketch}\proof}{\endproof}
  
\newenvironment{claimproof}{%
  \renewcommand{\proofname}{Proof of Claim}\proof}{\hfill $\lhd$}

\newtheorem{claim}{Claim}

\newcommand{\proofsubparagraph}[1]{\emph{#1}}

\newcommand{\qrank}{\mathsf{qr}}

\mathchardef\mhyphen="2D

\newcommand{\concrel}{R_\circ}

\newcommand{\foeq}{\mathsf{FO}[\mathsf{EQ}]}

\newcommand{\langinf}{L(w^\omega, v^\omega)}

\newcommand{\cvector}[1]{\vec{c}^{\, #1}}

\newcommand{\expo}{\mathsf{exp}}
\newcommand{\qr}{\mathsf{qr}}
\newcommand{\fraisse}{Fra\"{\i}ss\'e }

\newcommand{\facts}{\mathsf{Facs}}

\newcommand{\subs}{\sigma}

\newcommand{\dynfo}{\ifmmode{\mathsf{DynFO}}\else{\textsf{DynFO}}\xspace\fi}
\newcommand{\dynprop}{\ifmmode{\mathsf{DynPROP}}\else{\textsf{DynPROP}}\xspace\fi}
\newcommand{\dyncq}{\ifmmode{\mathsf{DynCQ}}\else{\textsf{DynCQ}}\xspace\fi}
\newcommand{\dynucq}{\ifmmode{\mathsf{DynUCQ}}\else{\textsf{DynUCQ}}\xspace\fi}

\newcommand{\dynqf}{\ifmmode{\mathsf{DynQF}}\else{\textsf{DynQF}}\xspace\fi}

\newcommand{\absins}[1]{\ifmmode{\mathsf{ins_{#1}}}\else{\textsf{ins_{#1}}}\xspace\fi}
\newcommand{\absreset}{\ifmmode{\mathsf{reset}}\else{\textsf{reset}}\xspace\fi}

\newcommand{\select}{\zeta}

\newcommand{\fun}[1]{\llbracket #1 \rrbracket}

\newcommand{\bind}[2]{#1\{#2\}}

\newcommand{\strucvar}{\mathfrak{u}}

\newcommand{\var}{\mathsf{vars}}

\newcommand{\splog}{\ifmmode{\mathsf{SpLog}}\else{\textsf{SpLog}}\xspace\fi}
\newcommand{\splogneg}{\ifmmode{\mathsf{SpLog}^{\neg}}\else{\textsf{SpLog}$^{\neg}$}\xspace\fi}

\newcommand{\dpcsplog}{\ifmmode{\mathsf{DPC}}\else{\textsf{DPC}}\xspace\fi}
\newcommand{\pcsplog}{\ifmmode{\mathsf{PC}}\else{\textsf{PC}}\xspace\fi}

\newcommand{\ECrtext}{\textsf{EC\textsuperscript{reg}}}
\newcommand{\EC}{\ifmmode{\mathsf{EC}}\else{\textsf{EC}}\xspace\fi}
\newcommand{\ECr}{\ifmmode{\mathsf{EC^{reg}}}\else{\ECrtext}\xspace\fi}
\newcommand{\fc}{\mathsf{FC}}
\newcommand{\fcreg}{\mathsf{FC}[\mathsf{REG}]}

\newcommand{\reg}{\mathsf{REG}}

\newcommand{\regconst}{\mathbin{\dot{\in}}}

\newcommand{\logeq}{\mathbin{\dot{=}}}

\newcommand{\fo}{\mathsf{FO}}

\newcommand{\cq}{\ifmmode{\mathsf{CQ}}\else{\textsf{CQ}}\xspace\fi}
\newcommand{\ucq}{\ifmmode{\mathsf{UCQ}}\else{\textsf{UCQ}}\xspace\fi}

\newcommand{\lang}{\mathcal{L}}

\newcommand{\signature}{\tau}

\newcommand{\emptyword}{\varepsilon}

\newcommand{\true}{\ifmmode{\mathsf{True}}\else{\textsf{True}}\xspace\fi}
\newcommand{\false}{\ifmmode{\mathsf{False}}\else{\textsf{False}}\xspace\fi}
\newcommand{\union}{\mathrel{\cup}}
\newcommand{\intersect}{\mathrel{\cap}}

\newcommand{\df}{:=}

\begin{document}
\sloppy

\title[Generalized Core Spanner Inexpressibility via Ehrenfeucht-\fraisse Games for FC]{Generalized Core Spanner Inexpressibility \\ via Ehrenfeucht-\fraisse Games for FC}

\settopmatter{printacmref=false} %
\renewcommand\footnotetextcopyrightpermission[1]{} %
\pagestyle{plain}

\author{Sam M. Thompson}
\affiliation{%
  \institution{Loughborough University}
  \city{Loughborough}
  \country{UK}}

\author{Dominik D. Freydenberger}
\affiliation{%
  \institution{Loughborough University}
  \city{Loughborough}
  \country{UK}}

\renewcommand{\shortauthors}{D.\ D.\ Freydenberger and S.\ M.\ Thompson}

\begin{abstract}
Despite considerable research on document spanners, little is known about the expressive power of generalized core spanners.
In this paper, we use Ehrenfeucht-\fraisse games to obtain general inexpressibility lemmas for the logic $\fc$ (a finite-model variant of the theory of concatenation). 
Applying these lemmas give inexpressibility results for $\fc$ that we lift to generalized core spanners. 
In particular, we give several relations that cannot be selected by generalized core spanners, thus demonstrating the effectiveness of the inexpressibility lemmas.
As an immediate consequence, we also gain new insights into the expressive power of core spanners.
\end{abstract}



\maketitle

\section{Introduction}\label{sec:intro}
Ehrenfeucht-\fraisse games are a fundamental tool for establishing inexpressibility results in finite-model theory.
In this paper, we develop techniques to use Ehrenfeucht-\fraisse games for the logic~$\fc$ and use these to gain insights into the expressive power of generalized core spanners.

\paragraph*{Document Spanners.}
SystemT is a rule-based information extraction system, developed by IBM, which includes a declarative query language called Annotation Query Language (AQL) -- see~\cite{chiticariu2010systemt, krishnamurthy2009systemt, li2011systemt}.
Fagin, Kimelfeld, Reiss, and Vansummeren~\cite{fag:spa} developed a formal framework for information extraction called \emph{document spanners} (or just \emph{spanners}), which captures the core functionality of AQL.

One can define the process of querying a word with a spanner in a two-step process.
First, so-called \emph{extractors} obtain relations of intervals (or \emph{spans}) from the text. 
For the purposes of this introduction, we assume these extractors to be \emph{regex formulas}, which are regular expressions with capture variables.
For example, if someone wanted to find all occurrences of common misspellings in a text document, they could consider the regex formula $\gamma(x) \df \Sigma^* \cdot \bind{x}{ \mathtt{acheive} \lor \mathtt{begining} \lor \cdots \lor \mathtt{wether} } \cdot \Sigma^*$. 
Then, $\gamma$ extracts a unary relation of spans within the given text document where one of the misspellings occur.

Secondly, the extracted relations are combined using a relational algebra.
This paper is concerned with the class of \emph{generalized core spanners}, which allow for $\cup$ (union), $\pi$ (projection), $\bowtie$ (natural join), $\setminus$ (difference), and $\zeta^=$ (equality selection). 
While union, projection, natural join, and difference are defined as one might expect (assuming prior knowledge of relational algebra), equality selection is more text specific.
The equality selection $\select^=_{x,y}$ only chooses those tuples for which $x$ and $y$ are mapped to spans that represent the same factor of the input word (at potential different locations).

A word relation $R \subseteq (\Sigma^*)^r$, for some $r \in \mathbb{N}$, is \emph{selectable} by generalized core spanners if we add an $R$ selection operator (e.g., $\select^R$) to our relational algebra without increasing the expressive power.
One fundamental question regarding generalized core spanners is what relations are selectable?
In other words, what is the expressive power of generalized core spanners?
For example, from Theorem 5.14 of~\cite{frey2019finite}, we know that length equality $\select^\mathsf{len}_{x,y}$ (i.e., only keep those tuples for which the span that $x$ is mapped to has the same length as the span that $y$ is mapped to) cannot be expressed by generalized core spanners.

\paragraph*{Ehrenfeucht-\fraisse Games and FC}
For relational first-order logic (and therefore relational algebra), one of the fundamental tools for inexpressibility are \emph{Ehrenfeucht-\fraisse Games}; for example, see~\cite{ebbinghaus1999finite,immerman1998descriptive,libkin2004elements}.
To use Ehrenfeucht-\fraisse games in order to yield inexpressibility results for generalized core spanners, we need a suitable logic.
For this task, we use $\fc$; a logic that was introduced by Freydenberger and Peterfreund~\cite{frey2019finite} as a finite-model variant on the theory of concatenation. 
For this paper, we define this logic slightly differently to~\cite{frey2019finite}.
First, we treat every word $w \in \Sigma^*$ as a relational structure consisting of a universe that contains all factors of $w$, a ternary concatenation relation $\concrel$ where $(x,y,z) \in \concrel$ if and only if $x$, $y$, and $z$ are all factors of $w$ where $x = y \cdot z$, and constant symbols for every $\mathtt{a} \in \Sigma$ and for the emptyword $\emptyword$.
Then, $\fc$ is simply a first-order logic over these structures.
As syntactic sugar, we use $(x \logeq y \cdot z)$ rather than $\concrel(x,y,z)$ for atomic $\fc$-formulas.
As an example, consider the formula
\[ \varphi \df \forall z \colon \bigl( \neg (z \logeq \emptyword) \rightarrow \neg \exists x, y \colon ( x \logeq z \cdot y) \land (y \logeq z \cdot z) \bigr). \]
This formula states, that for all factors $z$ where $z \neq \emptyword$ does not hold (all non-empty factors), we have that there does not exist a factor $x$ where $x = zzz$.
In other words, $\varphi$ defines those words $w \in \Sigma^*$ that do not contain $uuu$ where $u$ is a non-empty word.

If we consider Ehrenfeucht-\fraisse games over two structures that represent a word $w \in \Sigma^*$, we can gain inexpressibility results for $\fc$.
Then, using techniques from combinatorics on words (namely, \emph{commutation}), we are able to lift these inexpressibility results to generalized core spanners.

\paragraph*{Related Work.}
Fagin et al.~\cite{fag:spa} introduces three main classes of document spanners.
The \emph{regular spanners} (regex formulas extended with projection, union, and natural join) are perhaps the most intensively studied class with a sizeable amount of research on enumeration algorithms~\cite{florenzano2020efficient,amarilli2020constant,schmid2021spanner}.
When it comes to expressive power, Fagin et al.~\cite{fag:spa} showed that regular spanners cannot go beyond \emph{recognizable relations}.
\emph{Core spanners} extend regular spanners with equality selection, allowing for more expressive power at the cost of tractability, see e.g.~\cite{fre:splog,fre:doc,frey2019finite}.
With regards to their expressive power, \cite{fag:spa} gave the \emph{core-simplification lemma}, allowing one to simplify core spanners into a normal form in order to study their expressive power.
Freydenberger~\cite{fre:splog} uses inexpressibility techniques from word equations (namely, results in Karhum\"{a}ki, Mignosi, and Plandowski~\cite{karhumaki2000expressibility}) to obtain an inexpressibility lemma for core spanners, which yields relations that cannot be selected by core spanners.

When it comes to generalized core spanners, the techniques from core spanners do not work.
As far as the authors are aware, only two inexpressibility results for generalized core spanners exist in literature: 
First, Peterfreund, ten Cate, Fagin, and Kimelfeld~\cite{peterfreund2019recursive} showed that over a unary alphabet, only \emph{semi-linear languages} can be expressed by generalized core spanners.
Then, Freydenberger and Peterfreund~\cite{frey2019finite} used the \emph{Feferman-Vaught theorem} to show that the language containing those words $w \df \mathtt{a}^n \mathtt{b}^n$ for $n \in \mathbb{N}$ is not expressible by generalized core spanners.
However, the use of the Feferman-Vaught theorem in~\cite{frey2019finite} relies on the limited structure of the language $\mathtt{a}^n \mathtt{b}^n$, and it is unclear how to generalize this proof technique.
Note that from the data complexity of model checking $\fc$-formulas, we know that it cannot express languages outside of $\mathsf{LOGSPACE}$, see~\cite{frey2019finite}.
While this observation might not be considered particularly enlightening from an inexpressibility point of view, we state this for completeness sake.

The expressive power of other classes of spanners (that is, apart from regular spanners, core spanners, and generalized core spanners) has also been considered.
For example, see~\cite{schmid2020purely,peterfreund2019recursive,peterfreund2021grammars,thompson2023languages}.

\paragraph*{Issues with Standard Techniques.}
Beyond toy examples, Eh\-ren\-feucht-\fraisse games are difficult to play, and usually require rather involved combinatorial arguments. 
Therefore, many techniques have been developed in order to achieve sufficient criteria that give Duplicator a winning strategy (see~\cite{fagin1996easier} for a survey).
For example, \emph{locality} is often used to show that Duplicator has a winning strategy.
Unfortunately, these locality results often fail for \emph{non-sparse structures}.
The structures we look at in the present paper are very non-sparse, and therefore we cannot simply apply locality results in order to gain inexpressibility tools.

Another useful tool is the \emph{Feferman-Vaught theorem}. 
Freydenberger and Peterfreund~\cite{frey2019finite} used this result  to prove that $\mathtt{a}^n \mathtt{b}^n$ is not an $\fc$-language, which then implies that length selection is unattainable for generalized core spanners.
To this end, \cite{frey2019finite} defines $\foeq$, which extends first-order logic over a linear order with symbol predicates with a built-in equality relation.
Then, we can decompose these structures that represent words of the form $\mathtt{a}^n \mathtt{b}^m$ into two substructures (one for $\mathtt{a}^n$ and one for $\mathtt{b}^m$).
Finally, the Feferman-Vaught theorem can be invoked to prove that $\mathtt{a}^n \mathtt{b}^n$ is not an $\fc$ language.
However, this proof idea does not generalize beyond words which can be decomposed into disjoint sections (at least not without a sizeable amount of extra machinery).

\paragraph*{Structure of the Paper.}
\cref{sec:prelims} gives notational conventions and definitions that are used throughout this article.
In~\cref{sec:EFgames}, we define Ehrenfeucht-\fraisse games for $\fc$, and give some basic results.
The main technical contributions are given in~\cref{sec:pumping}, where we work towards the \emph{Fooling Lemma} for $\fc$ -- a general tool for showing a language cannot be expressed in $\fc$.
On the way to the Fooling Lemma, we give two important lemmas called the \emph{Pseudo-Congruence Lemma} and the \emph{Primitive Power Lemma}.
We then use the Fooling Lemma in~\cref{sec:spanners} to obtain various relations that are not selectable by generalized core spanners.
Due to space constraints, most of the proofs can be found in the appendix.

\section{Preliminaries}\label{sec:prelims}
Let $\mathbb{N} \df \{0,1,2,\dots\}$ and let $\mathbb{N}_+ \df \mathbb{N} \setminus \{ 0 \}$ where $\setminus$ denotes set difference.
For $n \geq 1$, we use $[n]$ for $\{ 1,2,\dots,n\}$. 
The cardinality of a set $S$ is denoted by $|S|$.
For a vector $\vec a \in A^k$ for some set $A$ and $k \in \mathbb{N}$, we write $x \in \vec a$ to denote that $x$ is a component of $\vec a$.

We use $\Sigma$ for a fixed and finite alphabet of terminal symbols and we use $\Xi$ for a countably infinite set of variables. 
If $w = w_1 \cdot w_2 \cdot w_2$ where $w, w_1, w_2, w_3 \in \Sigma^*$, then we call $w_1$ a prefix of $w$, $w_2$ a factor~$w$, and $w_3$ a suffix of $w$.
If $w_1 \neq w$, then $w_1$ is a \emph{strict prefix} of $w$; and likewise, if $w_3 \neq w$, then $w_3$ is a \emph{strict suffix} of $w$.
We denote that $w_2$ is a \emph{factor} of $w$ by $w_2 \sqsubseteq w$ and if $w_2 \neq w$ also holds, then $w_2 \sqsubset w$.
For $w \in \Sigma^*$, we write the set of all factors of $w$ as $\facts(w) \df \{ u \in \Sigma^* \mid u \sqsubseteq w \}$.
We use $|w|$ to denote the length of a word $w \in \Sigma^*$, and for some $\mathtt{a} \in \Sigma$, we use $|w|_\mathtt{a}$ to denote the number of occurrences of $\mathtt{a}$ within $w$.
For $w \in \Sigma^*$, let $w^k \in \Sigma^*$ denote the word that consists of $k$ repetitions of $w$.
We always assume $w^0 = \emptyword$.

A word $w \in \Sigma^+$ is called \emph{imprimitive} if $w = z^k$ for some $z \in \Sigma^*$ and $k > 1$.
We always assume $\emptyword$ to be imprimitive.
If $w \in \Sigma^+$ is not imprimitive, then $w$ is \emph{primitive}.
That is, $w \in \Sigma^+$ is primitive if for all $z \in \Sigma^+$, we have that $w = z^m$ implies $w = z$.

\paragraph*{The logic FC}
We introduce a definition of $\fc$ that is slightly more technical than the original one from~\cite{frey2019finite}, and more suitable for defining games. The differences, and why they do no affect the expressive power, are discussed at the end of this section.

For the purposes of this paper, we consider one fixed signature $\signature_\Sigma \df \{ \concrel, \mathtt{a}_1, \dots, \mathtt{a}_m, \emptyword \}$ for every terminal alphabet $\Sigma \df \{ \mathtt{a}_1, \dots, \mathtt{a}_m\}$, where $\concrel$ is a ternary relation symbol and where $\emptyword$ and $\mathtt{a}_i$ for every $i \in [m]$ are constant symbols.
Given a word $w \in \Sigma^*$, let $\mathfrak{A}_w \df (A, \concrel^{\mathfrak{A}_w}, \mathtt{a}_1^{\mathfrak{A}_w}, \dots, \mathtt{a}_m^{\mathfrak{A}_w}, \emptyword^{\mathfrak{A}_w})$ be the $\signature_\Sigma$-structure that represents $w \in \Sigma^*$ as follows:
\begin{itemize}
\item $A \df \facts(w) \union \{ \perp \}$ is the universe,
\item $\concrel^{\mathfrak{A}_w} \df \{ (a,b,c) \in \facts(w)^3 \mid a = b \cdot c \}$ is concatenation, restricted to factors of $w$, 
\item $\mathtt{a}^{\mathfrak{A}_w} \df \mathtt{a}$ if $|w|_\mathtt{a} \geq 1$, and $\mathtt{a}^{\mathfrak{A}_w} = \perp$ otherwise, and
\item $\emptyword^{\mathfrak{A}_w} \df \emptyword$.
\end{itemize}
Note that if $|w|_\mathtt{a} = 0$, then $\mathtt{a}^{\mathfrak{A}_w} = \perp$.
One can think of $\perp$ as a null value; however, we usually deal with those words where $|w|_\mathtt{a} \geq 1$ for all $\mathtt{a} \in \Sigma$.
Therefore, due to the nature of the constant symbols used in $\signature_\Sigma$, we often do not distinguish between a constant symbol $\mathtt{a} \in \signature_\Sigma$, and the terminal symbol $\mathtt{a} \in \Sigma$.
In other words, when the structure is clear from context, we use $\mathtt{a} \in \Sigma$ rather than $\mathtt{a}^{\mathfrak{A}_w} \in A$.

An $\fc$-formula is a first-order formula, where the atomic formulas are of the form $\concrel(x,y,z)$ for variables or constants $x$, $y$, and $z$.
As syntactic sugar, we use $(x \logeq y \cdot z)$ as $\fc$ atomic formulas, as we always interpret $\concrel$ as concatenation.
More formally:
\begin{definition}
We define $\fc$, the set of all $\fc$-formulas, recursively as:
\begin{itemize}
\item If $x, y, z \in \Xi \union \Sigma \union \{ \emptyword \}$, then $(x \logeq y \cdot z) \in \fc$,
\item if $\varphi, \psi \in \fc$, then $(\varphi \land \psi) \in \fc$, $(\varphi \lor \psi) \in \fc$, and $\neg \varphi \in \fc$,
\item if $\varphi \in \fc$ and $x \in \Xi$, then $\forall x \colon \varphi \in \fc$ and $\exists x \colon \varphi \in \fc$.
\end{itemize}
\end{definition}

An interpretation $\mathcal{I} \df (\mathfrak{A}_w, \subs)$ consists of a $\signature_\Sigma$-structure $\mathfrak{A}_w$ that represents some $w \in \Sigma^*$, and a mapping $\subs \colon (\Sigma \union \Xi \union \{\emptyword\}) \rightarrow \facts(w) \union \{ \perp \}$, where $\subs(\mathtt{a}) = \mathtt{a}^{\mathfrak{A}_w}$ for each constant symbol $\mathtt{a} \in \signature_\Sigma$.
Furthermore, we always assume $\subs(x) \neq \perp$ for all $x \in \Xi$.

\begin{definition}
For any $\varphi\in \fc$, and any interpretation $\mathcal{I} \df (\mathfrak{A}_w, \subs)$, where $\mathfrak{A}_w$ represents $w \in \Sigma^*$, we define $\mathcal{I} \models \varphi$ as follows:
Let $\mathcal{I} \models (x \logeq y \cdot z)$ if $(\subs(x), \subs(y), \subs(z)) \in \concrel^{\mathfrak{A}_w}$.
For the quantifiers, we have that $\mathcal{I} \models \exists x \colon \varphi$ (or $\mathcal{I} \models \forall x \colon \varphi$) if $\subs_{x \rightarrow u} \models \varphi$ holds for an (or all) $u \in \facts(w)$, where $\subs_{x \rightarrow u}(x) = u$ and $\subs_{x \rightarrow u}(y) = \subs(y)$ for all $y \in \Xi \setminus \{ x \}$.
Negation, disjunction, and conjunction are defined canonically.
\end{definition}

We use the usual first-order logic definition of bound variables and free variables.
If $\varphi$ is a sentence (that is, there are no free variables in $\varphi$), then we simply write $\mathfrak{A}_w \models \varphi$.
Furthermore, for any $w \in \Sigma^*$, there is a unique $\signature_\Sigma$-structure $\mathfrak{A}_w$ that represents $w$. 
We can therefore define the \emph{language} of an $\fc$-sentence.

\begin{definition}
The language defined by a sentence $\varphi \in \fc$ is $\lang(\varphi) \df \{ w \in \Sigma^* \mid \mathfrak{A}_w \models \varphi \}$.
Let $\lang(\fc)$ be the class of languages definable by an $\fc$-sentence.
\end{definition}

We write $\fun{\varphi}(w)$ to denote the set of all mappings $\subs$ such that $(\mathfrak{A}_w, \subs) \models \varphi$, where $\mathfrak{A}_w$ is the $\signature_\Sigma$-structure that represents $w \in \Sigma^*$.
Note that for simplicity, we assume that for any $\subs \in \fun{\varphi}(w)$, the domain of $\subs$ is exactly the free variables of $\varphi$.

For a word relation $R \subseteq (\Sigma^*)^k$, we say that $R$ is \emph{definable in $\fc$} if there exists $\varphi_R \in \fc$ with free variables $x_1, x_2, \dots, x_k$ such that for any $w \in \Sigma^*$, we have $\subs \in \fun{\varphi_R}(w)$ if and only if 
\[ (\subs(x_1), \subs(x_2), \dots, \subs(x_k) ) \in R \intersect \facts(w)^k. \]
We say that such a formula $\varphi_R$ \emph{defines} $R$.

\begin{example}\label{example:FCrelations}
Consider the following $\fc$-formula:
\[ \varphi_w (x) \df \neg \exists z_1, z_2 \colon \bigl( ( (z_1 \logeq z_2 \cdot x) \lor (z_1 \logeq x \cdot z_2) ) \land \neg (z_2 \logeq \emptyword) \bigr). \]
This formula states that there is no factor that is a concatenation of $\subs(x)$ and a non-empty word $\subs(z_2)$. 
Thus, for $w \in \Sigma^*$, it must be that $\fun{\varphi}(w) = \{ \subs \}$ where $\subs(x) = w$.
Note that we use $(z_2 \logeq \emptyword)$ as a shorthand for $(z_2 \logeq \emptyword \cdot \emptyword)$.

We can use $\varphi_w(x)$ to define the following sentence:
\[ \varphi_{ww} \df \exists x \colon \exists y \colon \bigl( \varphi_w(x) \land (x \logeq yy) \bigr). \]
It follows that $\mathfrak{A} \models \varphi_{ww}$ if and only if $\mathfrak{A}$ is a $\signature_\Sigma$-structure that represents $ww$ for some $w \in \Sigma^*$.

Now consider $\varphi(x,y) \df (x \logeq yy)$.
If follows that $\varphi$ defines the relation $R_{\mathsf{copy}} \df \{ (u,v) \subseteq {(\Sigma^*)}^2 \mid u = vv \}$.
Since we know that $R_\mathsf{copy}$ can be defined in $\fc$, we can use it as an additional atomic formula without changing the expressive power of $\fc$.
Furthermore, it is rather straightforward to generalize this relation to the $\fc$ definable relation $R_{k \mathsf{\mhyphen copies}} \df \{ (u,v) \in (\Sigma^*)^2 \mid u = v^k \}$.
\end{example}

Before moving on, we first make a brief note on the difference between the definition of $\fc$ in this paper, and the definition given in~\cite{frey2019finite}.
First, the definition of $\fc$ given in~\cite{frey2019finite} allows for an arbitrarily large right-hand side. 
That is, atomic formulas of the form $x \logeq \alpha$ where $x \in \Xi$ and $\alpha \in (\Sigma \union \Xi)^*$.
Whereas we use $(x \logeq y \cdot z)$ as atomic formulas, where $x,y,z \in \Xi \union \Sigma \union \{ \emptyword \}$.
It is clear that an arbitrarily large right-hand side is shorthand for a binary concatenation term (for example, see Freydenberger and Thompson~\cite{freydenberger2021splitting}).
The reason behind our choice, is that we have a finite signature $\signature_\Sigma$ which more closely aligns with ``traditional'' finite-model theory.
The second difference is that we do not use a \emph{universe variable} $\strucvar$ to denote the ``input word'' (see~\cite{frey2019finite} for more details).
However, referring back to~\cref{example:FCrelations}, we can simply use a subformula to simulate the behaviour of $\strucvar$.
While there is also a small difference with the semantic definitions (such as the use of $\perp$), these differences are negligible and do not change the expressive power.

\section{Ehrenfeucht-\fraisse Games for FC.}\label{sec:EFgames}
In this section, we discuss the use of Ehrenfeucht-\fraisse games for $\fc$.
Although broader definitions of the concepts defined here exist (see e.g.~\cite{ebbinghaus1999finite,immerman1998descriptive,libkin2004elements}), we tailor our definitions to the logic~$\fc$.

\begin{definition}[Partial Isomorphism]\label{defn:partialiso}
Let $\mathfrak{A}_w$ and $\mathfrak{B}_v$ be two $\signature_\Sigma$-structures that represent $w \in \Sigma^*$ and $v \in \Sigma^*$, respectively.
Let $A \df \facts(w) \union \{ \perp \}$ and $B \df \facts(v) \union \{ \perp \}$ be the universes of these structures.
For $\vec a = (a_1,\dots,a_n)$ and $\vec b = (b_1, \dots, b_n)$ where $a_i \in A$ and $b_i \in B$ for all $i \in [n]$, we say that $(\vec a, \vec b)$ defines a \emph{partial isomorphism} between $\mathfrak{A}_w$ and $\mathfrak{B}_v$ if: 
\begin{itemize}
\item For every $i \in [n]$, and constant symbol $c \in \signature_\Sigma$, we have $a_i = c^{\mathfrak{A}_w}$ if and only if $b_i = c^{\mathfrak{B}_v}$,
\item for every $i,j \in [n]$, we have that $a_i = a_j$ if and only if $b_i = b_j$, and
\item for every $i,j,k \in [n]$, we have that $a_i = a_j \cdot a_k$ if and only if $b_i = b_j \cdot b_k$.
\end{itemize}
\end{definition}

We now discuss Ehrenfeucht-\fraisse games.
This game is played between two players, called \emph{Spoiler} and \emph{Duplicator}.
For our purposes, Ehrenfeucht-\fraisse games are played on two $\signature_\Sigma$-structures $\mathfrak{A}_w$ and $\mathfrak{B}_v$ that represent some $w \in \Sigma^*$ and $v \in \Sigma^*$ respectively.
Let $A \df \facts(w) \union \{ \perp \}$ and $B \df \facts(v) \union \{ \perp \}$.
The players play a fixed number (which we often denote with $k \in \mathbb{N}$) of rounds.
For any $i \leq k$, the $i$-th round is played as follows:
\begin{description}
\item[Spoiler's Move:] Spoiler picks a structure $\mathfrak{A}_w$ or $\mathfrak{B}_v$, and then picks some element $a_i \in A$ if they chose $\mathfrak{A}_w$, or some $b_i \in B$ if they chose $\mathfrak{B}_v$.
\item[Duplicator's Move:] Duplicator then responds by picking some element in the other structure.
\end{description}

For any $\signature_\Sigma$-structure $\mathfrak{A}$, where $\Sigma \df \{ \mathtt{a}_1, \dots, \mathtt{a}_n \}$, we use $\cvector{\mathfrak{A}}$ to denote $(\mathtt{a}_1^\mathfrak{A}, \dots, \mathtt{a}_n^\mathfrak{A}, \emptyword^{\mathfrak{A}})$.
Let $\mathfrak{A}_w$ and $\mathfrak{B}_v$ be two $\signature_\Sigma$ structures that represent $w \in \Sigma^*$ and $v \in \Sigma^*$ respectively.
For a $k$-round game $\mathcal{G}$ over $\mathfrak{A}_w$ and $\mathfrak{B}_v$, let $a_i \in A$ and $b_i \in B$ be the elements chosen in round $i$ for each $i \in [k]$.
Furthermore, let $\vec a = (a_1,a_2,\dots,a_{k+|\Sigma|+1})$ and $\vec b = (b_1,b_2,\dots,b_{k+|\Sigma|+1})$ be tuples of the chosen elements, combined with the tuples $\cvector{\mathfrak{A}_w}$ and $\cvector{\mathfrak{B}_v}$ respectively.
That is, the tuples consisting of the last $|\Sigma|+1$ components of $\vec a$ and $\vec b$ are $\cvector{\mathfrak{A}_w}$ and $\cvector{\mathfrak{B}_v}$ respectively.
Then, Duplicator has won the $k$-round game if and only if $(\vec a, \vec b)$ is a partial isomorphism between $\mathfrak{A}_w$ and $\mathfrak{B}_v$.

For a $k$-round game $\mathcal{G}$ over $\mathfrak{A}_w$ and $\mathfrak{B}_v$, we say that Duplicator has a \emph{winning strategy} for $\mathcal{G}$ if and only if there is a way for Duplicator to play such that no matter what Spoiler chooses, Duplicator can win $\mathcal{G}$.
Otherwise, Spoiler has a winning strategy for $\mathcal{G}$.
If Duplicator has a winning strategy for $\mathcal{G}$, then we write $\mathfrak{A}_w \equiv_k \mathfrak{B}_v$. 
Since $\mathfrak{A}_w$ and $\mathfrak{B}_v$ are uniquely determined by $w,v \in \Sigma^*$, we can simply write $w \equiv_k v$ when convenient.

\begin{example}
Let $\Sigma = \{ \mathtt{a} \}$.
For any $i \in \mathbb{N}$, Spoiler has a 2-round winning strategy on $\mathfrak{A}_w$ and $\mathfrak{B}_v$ where $w \df \mathtt{a}^{2i}$ and $v \df \mathtt{a}^{2i-1}$:
\begin{description}
\item[Round One.] Spoiler chooses the structure $\mathfrak{A}_w$ and the factor $a_1 = \mathtt{a}^{2i}$. Duplicator responds with some $b_1 \sqsubseteq v$.
\item[Round Two.] The choice for Spoiler in this round is broken into two cases, depending on Duplicator's choice in round one:
\begin{enumerate}
\item If Duplicator responds in round one with $b_1 = \mathtt{a}^{2i-1}$, then Spoiler chooses the structure $\mathfrak{A}_w$ and $a_2 = \mathtt{a}^{i}$, as there does not exists a factor $b_2 \sqsubseteq v$ such that $v = b_2 \cdot b_2$, Spoiler has won after two rounds. 
\item If Duplicator responds in round one with $b_1 \neq \mathtt{a}^{2i-1}$, then Spoiler chooses the structure $\mathfrak{B}_v$ and $b_2 = b_1 \cdot \mathtt{a}$. 
Duplicator must respond with some factor $a_2$, however a factor of the form $a_1 \cdot \mathtt{a}$ does not exist.
\end{enumerate}
\end{description}
For either of the cases given in round two, we have that $(a_1, a_2, \mathtt{a}, \emptyword)$ and $(b_1, b_2, \mathtt{a}, \emptyword)$ does not form a partial isomorphism and thus, Spoiler has a winning strategy for the two round game.
Consequently, $\mathtt{a}^{2i} \not \equiv_2 \mathtt{a}^{2i-1}$ for any $i \in \mathbb{N}_+$.
\end{example}

Next, we connect Ehrenfeucht-\fraisse games to $\fc$.
Each $\varphi \in \fc$ has a so-called \emph{quantifier rank}, denoted by the function $\qrank \colon \fc \rightarrow \mathbb{N}$.
Let $\varphi, \psi \in \fc$.
We define $\qrank$ recursively as follows:
\begin{itemize}
\item If $\varphi \in \fc$ is an atomic formula, then $\qrank(\varphi) \df 0$,
\item $\qrank(\neg \varphi) \df \qrank(\varphi)$, 
\item $\qrank(\varphi \land \psi) = \qrank(\varphi \lor \psi) \df \mathsf{max} \left( \qrank(\varphi), \qrank(\psi) \right)$, and 
\item $\qrank(Q x \colon \varphi) \df \qrank(\varphi) + 1$ for any $x \in \Xi$ and $Q \in \{\forall, \exists\}$.
\end{itemize}
Let $\fc(k)$ denote the set of sentences $\varphi \in \fc$ such that $\qrank(\varphi) \leq k$.

Note that the set of all $\signature_\Sigma$-structures $\mathfrak{A}_w$ that represent $w \in \Sigma^*$ are simply a subset of all structures over a signature that contains a ternary relational symbol, and $|\Sigma|+1$ constant symbols.
Therefore, we can state an $\fc$-version of the the Ehrenfeucht-\fraisse theorem  (see e.g.~\cite{ebbinghaus1999finite,immerman1998descriptive,libkin2004elements} for the general version).

\begin{theorem}\label{thm:ehrenfeucht}
Let $\mathfrak{A}_w, \mathfrak{B}_v$ be $\signature_\Sigma$-structures that represent $w \in \Sigma^*$ and $v \in \Sigma^*$ respectively. Then, for any $k \in \mathbb{N}$, the following are equivalent:
\begin{enumerate}
\item $\mathfrak{A}_w \equiv_k \mathfrak{B}_w$,
\item $\mathfrak{A}_w \models \varphi$ if and only if $\mathfrak{B}_v \models \varphi$ for all $\varphi \in \fc(k)$.
\end{enumerate}
\end{theorem}

Thus,~\cref{thm:ehrenfeucht} is useful for proving inexpressibility for first-order logic over finite models.
A slight reformulation of~\cref{thm:ehrenfeucht} yields the following observation:

\begin{lemma}\label{obs:equivToLang}
Let $L \subseteq \Sigma^*$. If for every $k \in \mathbb{N}$, there exist  $w \in L$ and $u \notin L$ where $w \equiv_k u$, then $L$ is not an $\fc$-language.
\end{lemma}
For our next step, we rely on an inexpressibility technique for languages over unary alphabets. 
Over a unary alphabet $\{\mathtt{a}\}$, we can identify every unary word $\mathtt{a}^n$ with the natural number $n$, which allows us to treat languages $L\subseteq \{\mathtt{a}\}^*$ as sets $S_L$ of natural numbers.

A set $S\subseteq\mathbb{N}$ is \emph{linear} if there exist an $r\geq 0$ and $m_0,\dots,m_r\in \mathbb{N}$ such that $S=\{m_0+\sum_{i=1}^r m_in_i \mid n_i\geq 0\}$. 
A set is \emph{semi-linear} if it is a finite union of linear sets, and we call a unary language $L$ semi-linear if the corresponding set $S_L$ is semi-linear.

Over a unary alphabet,  semi-linear languages are exactly the languages defined by \emph{Presburger arithmetic} (Ginsburg and Spanier~\cite{ginsburg1966semigroups}), core spanners (Freydenberger and Holldack~\cite{fre:doc}) or generalized core spanners (Peterfreund, ten Cate, Fagin, and Kimelfeld~\cite{peterfreund2019recursive}). By Freydenberger and Peterfreund~\cite{frey2019finite}, this means that $\fc$ (over unary alphabets) defines exactly the semi-linear languages.

As $2^n$ grows faster than any linear function, no finite union of linear sets can define $\{2^n\}$.
Hence, $L_{\mathsf{pow}}\df\{ \mathtt{a}^{2^n} \mid n \in \mathbb{N} \}$ is not semi-linear and, therefore, is not expressible in $\fc$.
From this, we  almost directly infer the following:

\begin{lemma}\label{lemma:pow2}
Let $\mathtt{a} \in \Sigma$. 
For every $k \in \mathbb{N}$, there exists $p,q \in \mathbb{N}$ such that $\mathtt{a}^{p} \equiv_k \mathtt{a}^{q}$ where $p \neq q$.
\end{lemma}

\cref{lemma:pow2} is one of our main building blocks for $\fc$ inexpressibility.
That is, we shall utilize Duplicator's winning strategy for the $k$-round game over structures that represent $\mathtt{a}^p$ and $\mathtt{a}^q$ in order to construct winning strategies for more general structures.

A binary relation $R \subseteq (\Sigma^*)^2$ is a \emph{congruence relation} if $a \mathrel{R} b$ and $c \mathrel{R} d$ implies $(a \cdot c) \mathrel{R} (b \cdot d)$.
When considering the more common $\fo[<]$ over words (where words are encoded as a linear order with symbol predicates), we have that $\equiv_k$ is a congruence relation (see Section 3.2 of~\cite{thomas1993ehrenfeucht}).
Unfortunately, this does not hold for $\fc$.

\begin{proposition}\label{lemma:notCongruent}
There exists $u, u', v, v' \in \Sigma^*$ where, for any $k \geq 5$, we have that $u \equiv_k v$ and $u' \equiv_k v'$, but $u \cdot u' \not\equiv_k v \cdot v'$.
\end{proposition}

\cref{lemma:notCongruent} rules out what would be a convenient technique for strategy composition.
Despite this, in~\cref{sec:congruence} we  consider a special case where $\equiv$ can be used as if it were a congruence relation.

\section{FC Inexpressibility}\label{sec:pumping}
Our next goal is using Ehrenfeucht-\fraisse games in order to yield inexpressibility results for $\fc$.
To do so, we heavily utilize \emph{strategy composition}.
That is, we ``bootstrap'' known winning strategies for Duplicator -- gathered from inexpressibility results such as~\cref{lemma:pow2} -- and construct more general winning strategies for Duplicator.
The main technical contribution of this section is the Fooling Lemma for $\fc$ that is derived from the \emph{Pseudo-Congruence Lemma} and the \emph{Primitive Power Lemma}.
\cref{sec:congruence} contains the Pseudo-Congruence Lemma (\cref{lemma:congruence}), and~\cref{sec:primitivepower} contains the Primitive Power Lemma (\cref{lemma:primitivePower}).
Both these sections also contain many useful lemmas on the way to proving the respective result.
Before looking at these lemmas, we first demonstrate the capabilities of $\fc$ by giving a language that (somewhat surprisingly) is expressible in $\fc$.

For every $n \in \mathbb{N}$, we define $F_n \in \{ \mathtt{a}, \mathtt{b} \}^*$ recursively as follows: $F_0 \df \mathtt{a}$, $F_1 \df \mathtt{ab}$, and $F_i \df F_{i-1} \cdot F_{i-2}$ for all $i \geq 2$.
The \emph{Fibonacci word} is the limit $F_\omega$ (see~\cite{lothaire2002algebraic} for details on $F_\omega$ and its properties).
We show that the language $L_\mathsf{fib}$ containing those words $\mathtt{c} F_0 \mathtt{c} F_1 \mathtt{c} \cdots \mathtt{c} F_n  \mathtt{c}$ for $\mathtt{c} \in \Sigma$ and for all $n \in \mathbb{N}$ is expressible in $\fc$.

\begin{proposition}\label{prop:fib}
$L_\mathsf{fib} \in \lang(\fc)$.
\end{proposition}
While the definition suggests that expressing $L_\mathsf{fib}$ requires a logic that allows recursion, we can use the universal quantifier to simulate recursion in certain cases.
As a curious aside, this also shows that $\fc$ does not have a pumping lemma in the sense that for sufficiently long words, some non-empty factor can be repeated an arbitrary number of times without falling out of the language.
This is due to the fact that $F_{\omega}$ does not contain any factor of the form $u^4$ with $u\neq\emptyword$, see Karhumäki~\cite{karhumaki1983cube}.

\subsection{The Pseudo-Congruence Lemma}\label{sec:congruence}
In this section, we first give some necessary lemmas, we then give the Pseudo-Congruence Lemma along with a proof idea. 
To conclude this section, we consider some consequences.

First, we consider the case where Spoiler has chosen a factor that is so short (with respect to the number of remaining rounds) that Duplicator must respond with the identical factor or lose.

\begin{lemma}\label{lemma:consistentStrats}
Let $\mathfrak{A}_w$ and $\mathfrak{B}_v$ be $\signature_\Sigma$-structures that represent $w \in \Sigma^*$ and $v \in \Sigma^*$, where $\mathfrak{A}_w \equiv_k \mathfrak{B}_v$.
Let $\vec a = (a_1, a_2, \dots, a_{k+|\Sigma|+1})$ and $\vec b= (b_1, b_2, \dots, b_{k+|\Sigma|+1})$ be the tuple resulting from a $k$-round game over $\mathfrak{A}_w$ and $\mathfrak{B}_v$ where Duplicator plays their winning strategy.
If $r + |a_r| -1 < k$ or $r + |b_r| - 1 < k$ for some $r \in [k]$, then $b_r = a_r$. 
\end{lemma}

Next, we show that for rounds $r$ with $r \leq k-2$, if Spoiler picks a prefix (or a suffix), then Duplicator must respond with a prefix (or suffix, respectively); assuming they play a winning strategy.

\begin{lemma}\label{lemma:prefixSuffix}
Let $\mathfrak{A}_w$ and $\mathfrak{B}_v$ be two $\signature_\Sigma$-structures that represent $w \in \Sigma^*$ and $v \in \Sigma^*$, where $\mathfrak{A}_w \equiv_k \mathfrak{B}_v$.
Let $\vec a = (a_1, a_2, \dots, a_{k+|\Sigma|+1})$ and $\vec b= (b_1, b_2, \dots, b_{k+|\Sigma|+1})$ be the tuple resulting from a $k$-round game over $\mathfrak{A}_w$ and $\mathfrak{B}_v$ where Duplicator plays their winning strategy.
For any $r \leq k-2$, we have that $a_r$ is a suffix (or prefix) of $w$ if and only if $b_r$ is a suffix (or prefix respectively) of~$v$. 
\end{lemma}

While we mainly use~\cref{lemma:consistentStrats} and~\cref{lemma:prefixSuffix} to prove the Pseudo-Congruence Lemma, they also provide some insights  into necessary conditions of Duplicator's strategy.

In some of the subsequent proofs, we use strategy compositions that require quite technical proofs of correctness. 
To avoid handwaving, we shall sometimes define Duplicator's strategy using what we call \emph{look-up games}.
If Spoiler and Duplicator are playing a $k$-round game $\mathcal{G}$ over $\mathfrak{A}_w$ and $\mathfrak{B}_v$, a \emph{look-up} game is a $k'$-round auxiliary game $\mathcal{G}'$ over two (potentially different) structures $\mathfrak{C}$ and $\mathfrak{D}$, where $\mathfrak{C} \equiv_{k'} \mathfrak{D}$.
The idea is that in the $i$-th round in $\mathcal{G}$, each move by Spoiler corresponds to a move by Spoiler in $\mathcal{G}'$.
Then, Duplicator ``looks up'' what their response would be in $\mathcal{G}'$ to form their response to Spoiler in $\mathcal{G}$.
Generalizing this idea, Duplicator can use multiple look-up games $\mathcal{G}_1, \mathcal{G}_2, \dots, \mathcal{G}_n$ to form their response in $\mathcal{G}$.
In other words, if Duplicator has winning strategies in $\mathcal{G}_1, \mathcal{G}_2,  \dots, \mathcal{G}_n$, then Duplicator has a winning strategy in $\mathcal{G}$.

While $\equiv_k$ is not a congruence relation in general (see~\cref{lemma:notCongruent}), we identify a special case where we can use $\equiv$ as if it were a congruence relation.

\begin{lemma}[Pseudo-Congruence Lemma]\label{lemma:congruence}
Let $w_1, w_2, v_1, v_2 \in \Sigma^*$ where $\facts(w_1) \intersect \facts(w_2) = \facts(v_1) \intersect \facts(v_2)$, and let $r \df \mathsf{max} \{|u| \in \mathbb{N} \mid u \in \facts(w_1) \intersect \facts(w_2) \}$.
If $w_1 \equiv_{k+ r + 2} v_1$ and $w_2 \equiv_{k+ r + 2} v_2$ for some $k \in \mathbb{N}_+$, then $w_1 \cdot w_2 \equiv_k v_1 \cdot v_2$.
\end{lemma}
\begin{sketch}
We consider two look-up games $\mathcal{G}_1$ and $\mathcal{G}_2$.
For $\iota \in \{1, 2 \}$, the game $\mathcal{G}_\iota$ is a $k+r+2$ round game over $w_\iota$ and $v_\iota$.
We assume Duplicator plays $\mathcal{G}_1$ and $\mathcal{G}_2$ using their winning strategy.
We use these look-up games to give Duplicator's strategy for the $k$ round game $\mathcal{G}$ over $w \df w_1 \cdot w_2$ and $v \df v_1 \cdot v_2$.

For each round $p \in [k]$ of $\mathcal{G}$, we determine Spoiler's choice in $\mathcal{G}_1$ and $\mathcal{G}_2$ from Spoiler's choice in $\mathcal{G}$.
If Spoiler chooses some $u \in \facts(w_\iota) \union \facts(v_\iota)$, where $\iota \in \{ 1, 2 \}$, then we let Spoiler's choice in round $p$ of $\mathcal{G}_\iota$ be $u$.
If Spoiler chooses some $u \in \facts(w) \setminus (\facts(w_1) \union \facts(w_2))$, then we consider $u_1 \in \facts(w_1)$ and $u_2 \in \facts(w_2)$ such that $u = u_1 \cdot u_2$. Then, we let Spoiler choose $u_1$ in round $p$ of $\mathcal{G}_1$ and $u_2$ in round $p$ of $\mathcal{G}_2$; see~\cref{fig:middle}.
We do the analogous action if Spoiler chooses $u \in \facts(v) \setminus (\facts(v_1) \union \facts(v_2))$.
If, for example, Spoiler chooses $u \in \facts(w_1) \setminus \facts(w_2)$, then we assume that Spoiler ``skips'' round $p$ of $\mathcal{G}_2$.
Analogously, Spoiler skips round $p$ of $\mathcal{G}_1$ if Spoiler chose $u \in \facts(w_2) \setminus \facts(w_1)$ in round $p$ of $\mathcal{G}$.

To sum up Duplicator's strategy for $\mathcal{G}$ informally: If Spoiler chooses an element from either $\facts(w_1)$ or $\facts(v_1)$, then Duplicator responds with their winning strategy for $\mathcal{G}_1$.
Likewise, if Spoiler chooses from either $\facts(w_2)$ or $\facts(v_2)$, then Duplicator responds with their winning strategy for $\mathcal{G}_2$.
Finally, if Spoiler chooses $u \in \facts(w) \setminus \bigl( \facts(w_1) \union \facts(w_2) \bigr)$ or $u \in \facts(v) \setminus \bigl( \facts(w_2) \union \facts(v_2) \bigr)$, then we split $u = u_1 \cdot u_2$ where $(u_1, u_2) \in \facts(w_1) \times \facts(w_2)$, or $(u_1, u_2) \in \facts(v_1) \times \facts(v_2)$, and Duplicator responds with the concatenation of their winning strategy when Spoiler chooses $u_1$ in $\mathcal{G}_1$ and their winning strategy when Spoiler chooses $u_2$ in $\mathcal{G}_2$.

The proof that this composition of strategies is a winning strategy for Duplicator is somewhat tedious.
However, in this sketch, we informally explain how this strategy is indeed a strategy for Duplicator over $w$ and $v$.
That is, Duplicator's strategy is well-defined and Duplicator always chooses a factor of $w$ or $v$.

The extra $r+2$ rounds for $\mathcal{G}_1$ and $\mathcal{G}_2$ are to ensure that Duplicator makes certain choices.
If in $\mathcal{G}$ Spoiler chooses some $u \in \facts(w_1) \intersect \facts(w_2)$ for some round $p \leq k$, then $|u| \leq r$ holds (from the lemma statement).
Therefore, in both $\mathcal{G}_1$ and $\mathcal{G}_2$, Duplicator responds with $u$, see~\cref{lemma:consistentStrats}.
The same holds when Spoiler chooses some $u \in \facts(v_1) \intersect \facts(v_2)$.
Therefore, Duplicator's choice is well defined for the case where Spoiler chooses some $u \in \facts(w_1) \intersect \facts(w_2)$ due to the fact that Duplicator's responses in $\mathcal{G}_1$ and $\mathcal{G}_2$ coincide.
The same holds when Spoiler chooses some factor in $\facts(v_1) \intersect \facts(v_2)$.

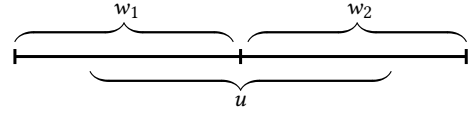
\begin{figure}
\centering
\begin{tikzpicture}[fill=white]
\tikzset{
    position label/.style={
       below = 3pt, 
       text height = 1.5ex,
       text depth = 1ex
    },
   brace/.style={
     decoration={brace, mirror},
     decorate
   }
}

\draw [-, line width=0.4mm] (-3,0) -- (3,0);
\draw [-, line width=0.4mm] (0, -.1) -- (0,.1);
\draw [-, line width=0.4mm] (3, -.1) -- (3,.1);
\draw [-, line width=0.4mm] (-3, -.1) -- (-3,.1);
\draw [decorate,decoration = {calligraphic brace, amplitude=7}, line width=0.3mm] (-3, 0.2) -- (-0.1, 0.2) node [above, xshift= -40pt, yshift=5.5pt] {$w_1$};
\draw [decorate,decoration = {calligraphic brace, amplitude=7}, line width=0.3mm] (0.1, 0.2) -- (3, 0.2) node [above, xshift= -40pt, yshift=5.5pt] {$w_2$};
\draw [decorate,decoration = {calligraphic brace, amplitude=7}, line width=0.3mm] (2, -0.2) -- (-2, -0.2) node [below, xshift= 57pt, yshift= -5.5pt] {$u$};

\end{tikzpicture}
\caption{Illustration of factors in $\facts(w) \setminus ( \facts(w_1) \union \facts(w_2))$.}
\label{fig:middle}
\end{figure}

Let us consider the case where for round $p \leq k$ in $\mathcal{G}$, Spoiler chooses some $u \in \facts(w) \setminus (\facts(w_1) \union \facts(w_2))$.
Observing~\cref{fig:middle}, we can split $u$ into the suffix $u_1$ of $w_1$ and the prefix $u_2$ of $w_2$.
Thus, due to~\cref{lemma:prefixSuffix}, in $\mathcal{G}_1$, Duplicator responds to~$u_1$ with a suffix $u_1'$ of $v_1$.
Likewise, in $\mathcal{G}_2$, Duplicator responds to $u_2$ with a prefix $u_2'$ of $v_2$.
Consequently, $u_1' \cdot u_2'$ is indeed a factor of $v$.
The analogous reasoning holds if Spoiler chooses some $u \in \facts(v) \setminus (\facts(v_1) \union \facts(v_2))$.

We have given an description of Duplicator's strategy, and informally sketched a proof that this strategy is well-defined.
\end{sketch}

Freydenberger and Peterfreund~\cite{frey2019finite} showed that the language $\{ \mathtt{a}^n \cdot \mathtt{b}^n \mid n \in \mathbb{N} \}$ is not expressible in $\fc$.
Their proof uses an alternative logic $\foeq$ with equivalent expressive power, for which the Feferman-Vaught theorem can be invoked.
As this approach heavily relies on the  restricted nature of words $\mathtt{a}^n \cdot \mathtt{b}^m$, 
it seems to be very difficult to generalize.
Furthermore, even for $\foeq$, the Feferman-Vaught theorem needs some additional reasoning. 
Compare this to the following alternative proof.

\begin{example}\label{example:anbn}
From~\cref{lemma:pow2}, we know that for every $k \in \mathbb{N}$, there exists $p, q \in \mathbb{N}$ where $\mathtt{a}^p \equiv_{k+2} \mathtt{a}^q$ and, without loss of generality, $p < q$.
Trivially, we have that $\mathtt{b}^p \equiv_{k+2} \mathtt{b}^p$ for all $k \in \mathbb{N}$.
Furthermore, $\facts(\mathtt{a}^n) \intersect \facts(\mathtt{b}^m) = \{ \emptyword \}$  for any $m,n \in \mathbb{N}$.
We can now invoke~\cref{lemma:congruence}, with $r=0$, to state that for any $k \in \mathbb{N}$, there exists $p, q \in \mathbb{N}$ such that $\mathtt{a}^q \cdot \mathtt{b}^p \equiv_k \mathtt{a}^p \cdot \mathtt{b}^p$ and $p < q$.
Consequently, $\{ \mathtt{a}^i \cdot \mathtt{b}^j \mid 0 \leq i \leq j \} \notin \lang(\fc)$ and $\{ \mathtt{a}^n \cdot \mathtt{b}^n \mid n \in \mathbb{N} \} \notin \lang(\fc)$.
\end{example}

In addition to giving an alternative proof of an existing result, we can also use~\cref{lemma:congruence} to prove new inexpressibility results.

\begin{proposition}\label{prop:subword}
$ \{ \mathtt{a}^n \cdot (\mathtt{ba})^n \mid n \in \mathbb{N} \} \notin \lang(\fc)$.
\end{proposition}
\begin{proof}
To prove this result, we show that for every $k \in \mathbb{N}$, there exists $p,q \in \mathbb{N}$ where $p < q$ and $\mathtt{a}^q (\mathtt{ba})^q \equiv_k \mathtt{a}^p (\mathtt{ba})^q$.

From~\cref{lemma:pow2}, for every $k \in \mathbb{N}$, there exists $p,q \in \mathbb{N}$ such that $\mathtt{a}^q \equiv_{k+3} \mathtt{a}^p$ where $q \neq p$.
Without loss of generality, assume that $p < q$. 
Trivially, $(\mathtt{ba})^q \equiv_{k + 3}(\mathtt{ba})^q$ for every $k \in \mathbb{N}$.
Since, $\facts(\mathtt{a}^m) \intersect \facts((\mathtt{ba})^n) = \{ \emptyword, \mathtt{a} \}$ for any $m,n \in \mathbb{N}$, we can invoke~\cref{lemma:congruence} with $r=1$.
It therefore follows that for every $k \in \mathbb{N}$, there exists $p,q \in \mathbb{N}$ where $p \neq q$ and $\mathtt{a}^q (\mathtt{ba})^q \equiv_k \mathtt{a}^p (\mathtt{ba})^q$.
From~\cref{obs:equivToLang} we infer that $ \{ \mathtt{a}^n \cdot (\mathtt{ba})^n \mid n \in \mathbb{N} \} \notin \lang(\fc)$.
\end{proof}

Observing~\cref{prop:subword}, we  see that the Pseudo-Congruence Lemma by itself is a useful tool for inexpressibility.

\subsection{The Primitive Power Lemma}\label{sec:primitivepower}
In this section, we look at the Primitive Power Lemma. 
Before giving this result, we first consider a lemma for words that are primitive (that is, not a power of a shorter word).

For any $w \in \Sigma^+$, let $\expo_w \colon \Sigma^+ \rightarrow \mathbb{N}$ be a function that maps $u \in \Sigma^+$ to $m \in \mathbb{N}$ if $w^m \sqsubseteq u$ and there does not exist $m' > m$ such that $w^{m'} \sqsubseteq u$. 
Less formally, $\expo_w (u)$ is the maximum value $m \in \mathbb{N}$ such that $w^m$ is a factor of $u$.

\begin{example}
Consider $u \df \mathtt{aaaabaabaab}$.
Then, $\expo_{\mathtt{a}}(u) = 4$, and $\expo_{\mathtt{aab}}(u) = 3$.
\end{example}

The following  is due to some basic results on primitive words.

\begin{lemma}\label{obs:factorOfRep}
Let $w \in \Sigma^+$ be a primitive word and let $m \in \mathbb{N}$.
If $u \sqsubseteq w^m$ where $\expo_w(u) > 0$, then there is a unique proper suffix $u_1 \in \Sigma^*$ of $w$ and a unique proper prefix $u_2 \in \Sigma^*$ of $w$, such that $u = u_1 \cdot w^{\expo_w(u)} \cdot u_2$.
\end{lemma}

While perhaps not particularly interesting by itself, \cref{obs:factorOfRep} shall be an important part of the proof of the Primitive Power Lemma -- a result that (in a sense) generalizes~\cref{lemma:pow2}.

\begin{lemma}[Primitive Power Lemma]\label{lemma:primitivePower}
If $\mathtt{a}^p \equiv_{k+3} \mathtt{a}^q$ for $\mathtt{a} \in \Sigma$ and $p,q,k \in \mathbb{N}_+$, then $w^p \equiv_k w^q$ for any primitive word $w \in \Sigma^+$.
\end{lemma}
\begin{sketch}
To prove this lemma, we give a winning strategy for Duplicator for the $k$-round game $\mathcal{G}$ over $w^p$ and $w^q$.
This strategy is defined based on the $k+3$-round look-up game $\mathcal{G}_l$ over $\mathtt{a}^p$ and $\mathtt{a}^q$, where $\mathtt{a}^p \equiv_{k+3} \mathtt{a}^q$.

For each round $p \leq k$ in $\mathcal{G}$, if Spoiler chooses $u$ where $\expo_w(u) = n$, then we let Spoiler choose $\mathtt{a}^n$ in round $p$ of $\mathcal{G}_l$ in the corresponding word.
Duplicator's response in $\mathcal{G}_l$ is some $\mathtt{a}^m$. Then, we form Duplicator's response to Spoiler in $\mathcal{G}$ as follows:
\begin{itemize}
\item If in $\mathcal{G}_l$, Duplicator responds with $\emptyword$, then let Duplicator's response in $\mathcal{G}$ be the same factor that Spoiler chose. This is because it can be shown that Spoiler chose some $u$ where $\expo_w(u) = 0$. 
\item If in $\mathcal{G}_l$, Duplicator responds with $\mathtt{a}^m$ where $m \in \mathbb{N}_+$, then we can show that in $\mathcal{G}$, Spoiler chose $u = u_1 \cdot w^n \cdot u_2$ where $u_1$ and $u_2$ are a unique suffix and prefix of $w$ (see~\cref{obs:factorOfRep}). We define Duplicator's response in $\mathcal{G}$ as $u_1 \cdot w^m \cdot u_2$. See~\cref{fig:smallpowerStrat} for an informal illustration of Duplicator's strategy.
\end{itemize}

\begin{figure}
    \centering
    \begin{tikzpicture}
    \tikzstyle{vertex}=[rectangle,fill=white!35,minimum size=12pt,inner sep=3pt,outer sep=1.5pt]
   \tikzstyle{vertex2}=[rectangle,fill=white!35,draw=gray,rounded corners=.1cm]

    \draw[dotted] (-0.2,1) ellipse (2.3cm and 0.6cm);
    \draw[dotted] (-0.2,-0.9) ellipse (2.3cm and 0.6cm);
    
    \node[vertex] (1) at (-1, 1) {$u_1 \cdot w^n \cdot u_2$};
    \node[vertex] (2) at (1, 1) {$\mathtt{a}^n$};
    
    \node[vertex] (3) at (1,-1) {$\mathtt{a}^m$};
    \node[vertex] (4) at (-1,-1) {$u_1 \cdot w^m \cdot u_2$};
    
    \path[-latex,dashed] (1) edge node [above] {} (2);
    \path[-latex] (2) edge node [right] {$\mathcal{G}_l$} (3);
    \path[-latex,dashed] (3) edge node [below] {} (4);
    \path[-latex] (1) edge node [left] {$\mathcal{G}$} (4);
    \end{tikzpicture}
    \caption{Duplicator's strategy when Spoiler chooses some $u$ where $\mathsf{exp}_w(u) \geq 1$. The top oval illustrates Spoiler's choices in $\mathcal{G}$ and $\mathcal{G}_l$, and the bottom oval illustrates Duplicator's responses in $\mathcal{G}$ and $\mathcal{G}_l$.}
    \label{fig:smallpowerStrat}
\end{figure}
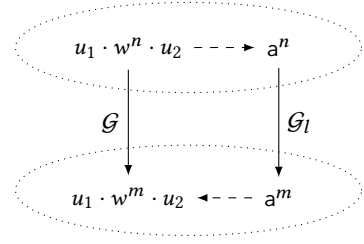

In the full proof, we show that this strategy is indeed a winning strategy for Duplicator by utilizing several lemmas regarding primitive words (such as~\cref{obs:factorOfRep}). 
\end{sketch}

For any $w \in \Sigma^+$, the \emph{primitive root} of $w$ is the unique primitive word $z \in \Sigma^*$ such that $w = z^k$ for some $k \in \mathbb{N}$.
Therefore, almost immediately from~\cref{lemma:primitivePower}, we are able to show:

\begin{proposition}\label{lemma:allWords}
For any $w \in \Sigma^+$ and any $k \in \mathbb{N}$, there exists $p \in \mathbb{N}_+$ and $v \in \Sigma^*$ such that $w^p \equiv_k v$ where $w^p \neq v$.
\end{proposition}

As an immediate consequence of~\cref{lemma:primitivePower}, it follows that for any primitive word $w \in \Sigma^+$, the language $w^{2^n}$ is not expressible in $\fc$.
While this may seem somewhat restrictive, due to the fact that this result holds for primitive words, any non-primitive word is a repetition of primitive words, hence~\cref{lemma:allWords}.

\subsection{The Fooling Lemma for FC}\label{subsec:pumping}

The Fooling Lemma for $\fc$ combines the Pseudo-Congruence Lemma with the Primitive Power Lemma. 
In order to give a convenient formulation of this result, we first look at \emph{co-primitive} words:
For two words $w, v \in \Sigma^+$, we say that $w$ and $v$ are \emph{conjugate} if there exists $x,y \in \Sigma^*$ such that $w = x \cdot y$ and $v = y \cdot x$.
If $w$ and $v$ are not conjugate, and both $w$ and $v$ are primitive words, then we say that $w$ and $v$ are \emph{co-primitive}.

\begin{example}
It is clear that $u \df \mathtt{aabba}$ and $v \df \mathtt{aaabb}$ are both primitive words.
However, $u$ and $v$ are not co-primitive.
This can easily be observed by considering $x \df \mathtt{aabb}$ and $y \df \mathtt{a}$.
Then, $u = xy$ and $v = yx$.
Now consider $u' \df \mathtt{aba}$ and $v' \df \mathtt{bba}$.
Again, both $u'$ and $v'$ are primitive.
Furthermore, from the fact that $|u'|_\mathtt{a} \neq |v'|_{\mathtt{a}}$, we can also see that $u'$ and $v'$ are co-primitive.
\end{example}

For any two words $w, v \in \Sigma^*$, let $w^\omega = w \cdot w \cdots$ and $v^\omega \df v \cdot v \cdots$ be the one-sided infinite word consisting of continuous repetitions of $w$ and $v$ respectively. 

We now state the \emph{periodicity lemma}.
For our purposes, we use the formulation of this result given in~\cite{hadravova2015equation}.
Further information regarding the periodicity lemma can be found in~\cite{handbookOfFL}.

\begin{lemma}[Periodicity lemma]
Let $w$ and $v$ be primitive words.
If $w^\omega$ and $v^\omega$ have a common factor of length at least $|w|+|v|-1$, then $w$ and $v$ are conjugate.
\end{lemma}

The periodicity lemma and some further reasoning gives us:

\begin{lemma}\label{lemma:coprim}
For any primitive words $w,v \in \Sigma^*$, the following are equivalent:
\begin{enumerate}
\item $w$ and $v$ are co-primitive,
\item there exists $n_0, m_0 \in \mathbb{N}$ such that $ \facts(w^{n_0}) \intersect \facts(v^{m_0}) = \facts(w^{n}) \intersect \facts(v^{m}) $ for all $n > n_0$ and all $m > m_0$, and
\item there exists $r \in \mathbb{N}$ such that $r \geq \mathsf{max} \{ |w| \in \mathbb{N} \mid w \in \facts(u^n) \intersect \facts(v^m) \}$ for all $n,m \in \mathbb{N}$.
\end{enumerate}
\end{lemma}

We are almost ready to give the Fooling Lemma for $\fc$.
Before doing so, the Fooling Lemma is simply a possible way of combining the results we have given so far, and there could be a more general inexpressibility lemma for $\fc$ just using the results from this paper.
However, for our purposes, \cref{lemma:FCpumping} is enough.

\begin{lemma}[Fooling Lemma]\label{lemma:FCpumping}
Let $w_1, w_2, w_3 \in \Sigma^*$, let $u,v \in \Sigma^+$ be two co-primitive words, let $\varphi \in \fc$, and let $f \colon \mathbb{N} \rightarrow \mathbb{N}$ be an injective function.
If $w_1 \cdot u^p \cdot w_2 \cdot v^{f(p)} \cdot w_3 \in \lang(\varphi)$ for all $p \in \mathbb{N}$, then $w_1 \cdot u^s \cdot w_2 \cdot v^t \cdot w_3 \in \lang(\varphi)$ for some $s,t \in \mathbb{N}$, where $f(s) \neq t$.
\end{lemma}
\begin{sketch}
We use the Primitive Power Lemma along with the Pseudo-Congruence Lemma to determine that for every $k \in \mathbb{N}$ there exists $p, q \in \mathbb{N}$ such that $w_1 \cdot u^p \cdot w_2 \equiv_k w_1 \cdot u^q \cdot w_2$ where $p \neq q$.
Then, we again invoke the Pseudo-Congruence Lemma (using~\cref{lemma:coprim} to ensure that the required prerequisites hold) to conclude that for every $k \in \mathbb{N}$, there exists $p, q \in \mathbb{N}$ where $p \neq q$ and $ w_1 \cdot u^p \cdot w_2 \cdot v^{f(p)} \cdot w_3 \equiv_k w_1 \cdot u^q \cdot w_2 \cdot v^{f(p)} \cdot w_3$.
Since $f$ is injective, $f(p) \neq f(q)$ if $p \neq q$.
\end{sketch}

Immediately from~\cref{lemma:FCpumping}, we infer the following:

\begin{proposition}\label{lemma:pumpingConsequence}
Let $w_1, w_2, w_3 \in \Sigma^*$, and let $u,v \in \Sigma^+$ be two co-primitive words.
Then, for any injective function $f \colon \mathbb{N} \rightarrow \mathbb{N}$, the language $\{ w_1 \cdot u^p \cdot w_2 \cdot v^{f(p)} \cdot w_3 \mid p \in \mathbb{N}  \}$
is not an $\fc$-language.
\end{proposition}

So far, we have developed several inexpressibility techniques.
Next, we look at using these tools to show that the following languages are not $\fc$ languages.

\begin{lemma}\label{lemma:nonFClangs}
The following languages are not expressible in $\fc$:
\begin{itemize}
\item $L_1 \df \{ \mathtt{a}^n \cdot (\mathtt{ba})^n  \mid n \in \mathbb{N} \}$,
\item $L_2 \df \{ \mathtt{a}^i \cdot (\mathtt{ba})^j \mid 1 \leq i \leq j  \}$,
\item $L_3 \df \{ \mathtt{b}^n \cdot \mathtt{a}^m \cdot \mathtt{b}^{n+m} \mid m,n \in \mathbb{N} \}$,
\item $L_4 \df \{ \mathtt{b}^n \cdot \mathtt{a}^m \cdot \mathtt{b}^{n \cdot m} \mid m,n \in \mathbb{N} \}$,  
\item $L_5 \df \{ (\mathtt{abaabb})^m \cdot (\mathtt{bbaaba})^m \mid m \in \mathbb{N} \}$, and
\item $L_6 \df \{ \mathtt{a}^n \cdot \mathtt{b}^n \cdot (\mathtt{ab})^n \mid n \in \mathbb{N} \}$.
\end{itemize}
\end{lemma}

The authors note that it seems at least possible to show that $L_3$ and $L_4$ in~\cref{lemma:nonFClangs} are not $\fc$ languages by using the Feferman-Vaught Theorem and the logic $\foeq$ (see~\cite{frey2019finite} for more details).
However, in order to do so, a sizeable amount of work would be needed.
In comparison, using the techniques developed in this work, it is rather straightforward to show $L_3, L_4 \notin \lang(\fc)$.

While~\cref{lemma:nonFClangs} seems somewhat ``language theoretic'', in~\cref{thm:relations} we shall use~\cref{lemma:nonFClangs} to show that several relations are not definable in $\fcreg$.

\section{Generalized Core Spanner Inexpressibility}\label{sec:spanners}
Freydenberger and Peterfreund~\cite{frey2019finite} proved that $\fc$, extended with so-called \emph{regular constraints} captures the expressive power of generalized core spanners.
That is, a word relation $R \subseteq (\Sigma^*)^k$ is definable in $\fcreg$ if and only if $R$ is \emph{selectable} by generalized core spanners.
We do not give the definition of generalized core spanners here, since we are able to simply use $\fcreg$ instead.
See~\cite{fag:spa} for the formal definitions of generalized core spanners\footnote{Note that in~\cite{fag:spa} what we call generalized core spanners is called core spanners with difference. The term ``generalized core spanner'' was coined in~\cite{peterfreund2019recursive}.}.

\begin{definition}
A \emph{regular constraint} is an atomic formula $(x \regconst \gamma)$ where $x \in \Xi$ and $\gamma$ is a regular expression.
Let $\fc[\reg]$ be the logic that extends $\fc$ with regular constraints.
The semantics are as follows: For a regular constraint $(x \regconst \gamma)$, and an $\fc$ interpretation $\mathcal{I} \df (\mathfrak{A}_w, \subs)$, where $\mathfrak{A}_w$ is a $\signature_\Sigma$-structure that represents $w \in \Sigma^*$, we write $\mathcal{I} \models (x \regconst \gamma)$ if $\subs(x) \sqsubseteq w$, and $\subs(x) \in \lang(\gamma)$.
\end{definition}

Unfortunately, the introduction of regular constraints also introduces some difficulties when dealing with the quantifier rank of a formula.
Whilst one could easily adapt the definition of quantifier rank to $\fc[\reg]$, we run into the issue that there are infinite $\fc[\reg]$ formulas of quantifier rank one.
Consider $\exists x \colon (x \regconst \gamma)$, along with the infinite number of regular languages.
This would mean that we cannot use~\cref{thm:ehrenfeucht} for $\fcreg$, which requires (up to logical equivalence) a finite number of formulas of quantifier rank $k$ for any $k \in \mathbb{N}$ (see, e.g.,~\cite{libkin2004elements} for more details).

We now consider a way to handle the regular constraints in order to derive inexpressibility results for generalized core spanners.

\begin{definition}
A language $L \subseteq \Sigma^*$ is \emph{bounded} if it is a subset of a language $w_1^* \cdot w_2^* \cdots w_n^*$ where $n \in \mathbb{N}$ and $w_i \in \Sigma^*$ for all $i \in [n]$.
\end{definition}

A \emph{bounded regular language} is simply a language that is both bounded, and regular.
From previous literature (see~\cite{fre:splog,frey2019finite} for example), we know that every bounded regular language can be expressed in $\fc$.
Thus, using a very similar proof idea to Theorem 6.2 from~\cite{fre:splog}, we are able to show the following:

\begin{lemma}\label{lemma:bounded}
If $L \subseteq \Sigma^*$ is a Boolean combination of bounded languages, then $L \in \lang(\fc)$ if and only if $L \in \lang(\fc[\reg])$.
\end{lemma}

From Lemma 5.5 in~\cite{frey2019finite}, we know that if every regular constraint in an $\fcreg$-formula is a so-called \emph{simple regular expression}, then we can rewrite the $\fcreg$-formula into an $\fc$-formula.
We can therefore replace every constraint that is a Boolean combination of simple regular expressions with an $\fc$-formula.
We do not go further into this, as~\cref{lemma:bounded} is all that we require.

Before looking at relations that are not definable by $\fcreg$, we first define some concepts.

For $x,y \in \Sigma^*$, we say that $x$ is a \emph{scattered subword} of $y$ (denoted by $x \sqsubseteq_\mathsf{scatt} y$) if for some $n \geq 1$, there exists $x_1, x_2, \dots x_n \in \Sigma^*$ and $y_0, y_1, \dots, y_n \in \Sigma^*$ such that $x = x_1 x_2 \cdots x_n$ and $y = y_0 x_1 y_1 x_2 y_2 \cdots x_n y_n$.
 
A similar concept to scattered subwords is the \emph{shuffle product}.
For $x, y \in \Sigma^*$, we define the shuffle product of $x$ and $y$, denoted $x \shuffle y$, as the set of all words $x_1 y_1 \cdot x_2 y_2 \cdots x_n y_n$ where $x = x_1 \cdots x_n$ and $y = y_1 \cdots y_n$ for $x_i, y_i \in \Sigma^*$ for all $i,j \in [n]$.

\begin{example}
Consider the words $u \df \mathtt{abba}$ and $v \df \mathtt{aa}$.
It is clear that $v \sqsubseteq_\mathsf{scatt} u$, and $\mathtt{ababaa} \in u \shuffle v$.
\end{example}

We are now ready to use~\cref{lemma:nonFClangs} to show that several relations are not selectable by generalized core spanners.
 
\begin{theorem}\label{thm:relations}
The following are not definable in $\fcreg$:
\begin{itemize}
\item $\mathsf{Num}_\mathtt{a} \df \{ (x,y) \in (\Sigma^*)^2 \mid |x|_\mathtt{a} = |y|_\mathtt{a} \}$ for any $\mathtt{a} \in \Sigma$,
\item $\mathsf{Add} \df \{ (x,y,z) \in (\Sigma^*)^3 \mid |z| = |x|+|y| \}$,
\item $\mathsf{Mult} \df \{ (x,y,z) \in (\Sigma^*)^3 \mid |z| = |x| \cdot |y| \}$,
\item $\mathsf{Scatt} \df \{ (x,y) \in (\Sigma^*)^2 \mid x \sqsubseteq_{\mathsf{scatt}} y \}$, 
\item $\mathsf{Perm} \df \{ (x,y) \in (\Sigma^*)^2 \mid x \text{ is a permutation of } y \}$, 
\item $\mathsf{Rev} \df \{ (x,y) \in (\Sigma^*)^2 \mid x \text{ is the reverse of } y \}$, 
\item $\mathsf{Shuff} \df \{ (x,y,z) \in (\Sigma^*)^3 \mid z \in x \shuffle y \}$,
\item $\mathsf{Morph}_h \df \{ (x, y) \in (\Sigma^*)^2 \mid y=h(x) \}$ for a given morphism $h \colon \Sigma^* \rightarrow \Sigma^* $.
\end{itemize}
A \emph{morphism} is a function $h \colon \Sigma^* \rightarrow \Sigma^*$ with $h(xy) = h(x) \cdot h(y)$ for all $x, y \in \Sigma^*$.
\end{theorem}
\begin{sketch}
Assuming some relation $R$ from theorem statement is definable in $\fcreg$, we can easily construct $\varphi \in \fcreg$ such that $\lang(\varphi)$ is a language given in~\cref{lemma:nonFClangs}.
This results in a contradiction due to~\cref{lemma:bounded}.
A similar proof idea was used in Proposition 6.7 of~\cite{fre:splog} to show that several relations are not definable by core spanners.
\end{sketch}

Immediately from~\cite{frey2019finite}, all the relations given in~\cref{thm:relations} are not selectable by generalized core spanners.
Consequently, if one needed to express a query from~\cref{thm:relations}, they would have to add a specialized operator to obtain this functionality.

For any relation $R \subseteq {(\Sigma^*)}^k$, where $k \in \mathbb{N}$, we have that $R$ is definable in $\fcreg$ if and only if ${(\Sigma^*)}^k \setminus R$ is definable by $\fcreg$.
This is simply due to the fact that $\fc$ is closed under complementation.
Thus, the complement of all relations given in~\cref{thm:relations} are also not selectable by generalized core spanners.

\paragraph{Core Spanners} Core spanners are a subclass of generalized core spanners (see~\cite{fag:spa}).
While there are known inexpressibility results for core spanners (for example, $\mathsf{Num}_\mathtt{a}$, $\mathsf{Scatt}$, $\mathsf{Perm}$, $\mathsf{Rev}$, and $R_< \df \{ (u,v) \mid |u| < |v| \}$ were shown to be not selectable by core spanners in~\cite{fre:splog}, and~\cite{fag:spa} proved that length equality is not selectable by core spanners), our results extend what is known about the expressive power of core spanners.
For example, since the complement of the relations given in~\cref{thm:relations} are not selectable by generalized core spanners, they are also not selectable by core spanners.
Furthermore, we also show that $\mathsf{Add}$, $\mathsf{Mult}$, $\mathsf{Shuff}$ and $\mathsf{Morph}_h$ are not selectable by generalized core spanners, and thus are not selectable by core spanners.

\section{Conclusions}
In~\cref{sec:EFgames}, we established Ehrenfreucht-\fraisse games for $\fc$.
Then, in~\cref{sec:pumping}, we gave two technical lemmas -- the Pseudo-Congruence Lemma, and the Primitive Power Lemma -- which serve the basis of the Fooling Lemma.
These lemmas are the main contributions of this work; however, in~\cref{sec:spanners}, we lifted our inexpressibility results to $\fcreg$ which immediately from~\cite{frey2019finite} yields inexpressibility results for generalized core spanners.

The main limitation of our results is that our generalized core spanner inexpressibility results require Boolean combinations of bounded languages.
It is possible to use closure properties in order to partially overcome this issue. 
For example, consider the language $L \df \{ w \in \Sigma^* \mid |w|_{\mathtt{a}} = |w|_{\mathtt{b}} \}$.
Since $\fcreg$ is closed under intersection with regular languages, $L \in \lang(\fcreg)$ if and only if $L \intersect \mathtt{a}^* \mathtt{b}^* \in \lang(\fcreg)$.
It is clear that $L \intersect \mathtt{a}^* \mathtt{b}^*$ is the language $\mathtt{a}^n \mathtt{b}^n$ and thus $L \notin \lang(\fcreg)$.

Such an approach can only get us so far, and therefore going beyond Boolean combinations of bounded languages is an important area for future research.
We note that it is still open as to whether $\fc$ and $\fcreg$ have the same expressive power. 
The authors believe that this is unlikely, although proving that $\fc$ and $\fcreg$ do indeed have equivalent expressive power would remove the restriction of Boolean combinations of bounded languages.

As an intermediate step, one could look at more general regular languages that can be expressed by $\fc$.
This would close the gap between what we know can be expressed by $\fc$ and what can be expressed by $\fcreg$.
Then, a similar result to~\cref{lemma:bounded} could gain further inexpressibility results for generalized core spanners.
A starting point for this direction would be to see whether regular constraints beyond the Boolean closure of \emph{simple regular expressions} are achievable in $\fc$ (see Lemma 5.5 in~\cite{frey2019finite}).

Another promising area of future work, is looking at related two player games.
For example, pebble games could be used to derive insights on $\fc$-formulas with a finite number of variables (for example, see Chapter 11 of~\cite{libkin2004elements}).
Alternatively, the restriction to existential Ehrenfreucht-\fraisse games could yield further results regarding the inexpressibility of core spanners.

We note that there are many further open problems regarding $\fc$.
We refer to~\cite{frey2019finite} for a more comprehensive list of future research directions.

\bibliographystyle{ACM-Reference-Format}
\bibliography{main}

\appendix
\clearpage
\onecolumn

\section{Appendix for Section~\ref{sec:EFgames}}

\subsection{Proof of Lemma~\ref{obs:equivToLang}}
\begin{proof}
We prove this observation by a contradiction.
To that end, assume there exists $\varphi \in \fc$ such that $\lang(\varphi) = L$, and let $\qrank(\varphi) = n$ for some $n \in \mathbb{N}$.
Further assume that we have $w \equiv_n u$, where $w \in L$ and $u \notin L$.
Since $w \in L$ and $\lang(\varphi) = L$, we know that $\mathfrak{A}_w \models \varphi$ where $\mathfrak{A}_w$ is the $\signature_\Sigma$-structure that represents $w$.
Furthermore, as $w \equiv_n u$, we know from~\cref{thm:ehrenfeucht} that $\mathfrak{B}_u \models \varphi$ where $\mathfrak{B}_u$ is the $\signature_\Sigma$-structure that represents $u$.
Therefore, $u \in \lang(\varphi)$ which is a contradiction.
\end{proof}

\subsection{Proof of Lemma~\ref{lemma:pow2}}
\begin{proof}
Working towards a contradiction, assume that there exists some $k \in \mathbb{N}$ such that for all $p, q \in \mathbb{N}$ where $q$ is not a power of two we have that $\mathtt{a}^{2^p} \not\equiv_k \mathtt{a}^{q}$.
Observing~\cref{obs:equivToLang}, there exists $\varphi \in \fc(k)$ such that $\mathfrak{A}_w \models \varphi$ if and only if $w = \mathtt{a}^{2^n}$ for some $n \in \mathbb{N}$.
As $L_\mathsf{pow} \df \{ \mathtt{a}^{2^n} \mid n \in \mathbb{N} \}$ is not semi-linear, it is not definable in $\fc$. Thus, we have reached a contradiction.
\end{proof}

\subsection{Proof of Lemma~\ref{lemma:notCongruent}}
\begin{proof}
For every $k$, there exists $p, q \in \mathbb{N}$ where $p \neq q$ and $\mathtt{a}^p  \equiv_k \mathtt{a}^q $, see~\cref{lemma:pow2}.
Trivially, $\mathtt{b} \cdot \mathtt{a}^p \equiv_k \mathtt{b} \cdot \mathtt{a}^p$.
But $\mathtt{a}^p \cdot \mathtt{b} \cdot \mathtt{a}^p \not\equiv_k \mathtt{a}^q \cdot \mathtt{b} \cdot \mathtt{a}^p$.
To show that this does indeed hold, we give an $\fc$-formula that accepts words of the form $v \cdot \mathtt{b} \cdot v$ for $v \in \Sigma^*$, and thus can distinguish between $\mathtt{a}^p \cdot \mathtt{b} \cdot \mathtt{a}^p $ and $\mathtt{a}^q \cdot \mathtt{b} \cdot \mathtt{a}^p$ where $p \neq q$.
\[
 \varphi \df \exists x, y, z  \colon \Bigl( (y \logeq x \cdot z) \land (z \logeq \mathtt{b} \cdot x) \land 
\neg \exists z_1,z_2 \colon \bigl( ( (z_1 \logeq z_2 \cdot y) \lor (z_1 \logeq y \cdot z_2)) \land \neg (z_2 \logeq \emptyword) \bigr) \Bigr). 
\]
For any $\signature_\Sigma$ structure $\mathfrak{A}_w$ that represents some word $w \in \Sigma^*$, the subformula 
\[ \neg \exists z_1,z_2 \colon \bigl( ( (z_1 \logeq z_2 \cdot y) \lor (z_1 \logeq y \cdot z_2)) \land \neg (z_2 \logeq \emptyword) \bigr) \]
 states that $\subs(y) = w$ for any substitution $\subs$ where $(\mathfrak{A}_w, \subs) \models \varphi$.
Thus, $\varphi$ states that for $w \in \lang(\varphi)$, we have that $w = x \mathtt{b} x$.
Consequently, $\lang(\varphi) = \{ v \cdot \mathtt{b} \cdot v \mid v \in \Sigma^* \}$ and therefore $\mathtt{a}^p \cdot \mathtt{b} \cdot \mathtt{a}^p \not\equiv_k \mathtt{a}^q \cdot \mathtt{b} \cdot \mathtt{a}^p$ for any $k \geq \qr(\varphi)$.
Since $\qr(\varphi) = 5$, the stated lemma holds.
\end{proof}

\section{Proof of Proposition~\ref{prop:fib}}
\begin{proof}
Recall that for every $n \in \mathbb{N}$, we define $F_n \in \{ \mathtt{a}, \mathtt{b} \}^*$ recursively as follows: $F_0 \df \mathtt{a}$, $F_1 \df \mathtt{ab}$, and $F_i \df F_{i-1} \cdot F_{i-2}$ for all $i \geq 2$.

In this proof, we show that $L_\mathsf{fib} \df \{ \mathtt{c}  F_0  \mathtt{c}  F_1  \mathtt{c} \cdots \mathtt{c} F_n \mathtt{c} \mid n \in \mathbb{N} \}$ over $\Sigma \df \{ \mathtt{a}, \mathtt{b}, \mathtt{c} \}$ is an $\fc$ language.
For this proof, we use arbitrary concatenation as a shorthand for binary concatenation.

First, consider
\[ \varphi_w(x) \df \neg \exists z_1, z_2 \colon \Bigl( \neg (z_1 \logeq \emptyword) \land \bigl( (z_2 \logeq z_1 \cdot x) \lor (z_2 \logeq x \cdot z_1) \bigr)  \Bigr). \]

This states that for any interpretation $(\mathfrak{A}_w, \subs)$ where $\mathfrak{A}_w$ represents $w \in \Sigma^*$, and $(\mathfrak{A}_w, \subs) \models \varphi_w(x)$, it holds that $\subs(x) = w$.

Now consider:
\[
\varphi_{\mathsf{struc}} \df \exists x_1, \strucvar \colon \bigl( \varphi_w(\strucvar) \land (\strucvar \logeq \mathtt{c a c ab c} x_1 \mathtt{c}) \land \neg \exists x_2 \colon (x_2 \logeq \mathtt{cc}).  \bigr)
\]

The formula $\varphi_\mathsf{struc}$ states that any $w \in \lang(\varphi_\mathsf{struc})$ is of the form $\mathtt{c} \cdot \mathtt{a} \cdot \mathtt{c} \cdot \mathtt{ab} \cdot \mathtt{c} \cdot \bigl( \{ \mathtt{a}, \mathtt{b} \}^+ \cdot \mathtt{c} \bigr)^+$.

Next, consider:
\[ \varphi_\mathsf{fib} \df \exists \strucvar \colon \Bigl( \varphi_\mathsf{struc} \land \forall x, y_1, y_2, y_3 \colon \bigl( (x \logeq \mathtt{c} y_1 \mathtt{c} y_2 \mathtt{c} y_3 \mathtt{c}) \rightarrow \bigl( \varphi_\mathtt{c}(y_1) \lor \varphi_\mathtt{c}(y_2) \lor \varphi_\mathtt{c}(y_3) \lor (y_3 \logeq y_2 \cdot y_1)  \bigr)  \bigr) \Bigr), \]
where $\varphi_\mathtt{c}(x) \df \exists y, z \colon ( x \logeq y \cdot \mathtt{c} \cdot z)$, and $\rightarrow$ is logical implication.

The formula $\varphi_\mathsf{fib}$ accepts those words that are in $\lang(\varphi_\mathsf{struc})$, and for any factor of the form $\mathtt{c} y_1 \mathtt{c} y_2 \mathtt{c} y_3 \mathtt{c}$ for $y_1, y_2, y_3 \in \Sigma^*$, either 
\begin{itemize}
\item for some $i \in [3]$, the word $y_i$ contains a $\mathtt{c}$, or 
\item $y_3 = y_2 \cdot y_1$.
\end{itemize}
Thus, for all factors of the form $\mathtt{c} y_1 \mathtt{c} y_2 \mathtt{c} y_3 \mathtt{c}$ for $y_1, y_2, y_3 \in \{ \mathtt{a}, \mathtt{b} \}^*$, we have that $y_3 = y_2 \cdot y_1$.

Consequently, $\lang(\varphi_\mathsf{fib}) = L_\mathsf{fib}$.
\end{proof}

\section{Appendix for Section~\ref{sec:congruence}}

Recall that~\cref{sec:congruence} is dedicated to the Pseudo-Congruence Lemma.
Before giving the proof of this result, we first consider some lemmas and definitions that are required for the proof.

\subsection{Proof of Lemma~\ref{lemma:consistentStrats}}
\begin{proof}
Let $\mathfrak{A}_w$ and $\mathfrak{B}_v$ be two $\signature_\Sigma$-structures that represent $w \in \Sigma^*$ and $v \in \Sigma^*$ respectively, where $\mathfrak{A}_w \equiv_k \mathfrak{B}_v$.
We denote the universe of $\mathfrak{A}_w$ and $\mathfrak{B}_v$ by $A$ and $B$ respectively.
Let $(\vec a, \vec b)$ be the $k+|\Sigma|+1$-tuples resulting from a $k$-round game over $\mathfrak{A}_w$ and $\mathfrak{B}_v$ where Duplicator plays their winning strategy.
Thus, we know that $(\vec a, \vec b)$ forms a partial isomorphism.
To distinguish individual components of $\vec a$ and $\vec b$, let $\vec a = (a_1, a_2, \dots, a_{k+|\Sigma|+1})$ and let $\vec b = (b_1, b_2, \dots, b_{k+|\Sigma|+1})$, where for each $i \in [k+|\Sigma|+1]$, we have that $a_i \in A$ and $b_i \in B$.

Working towards a contradiction, assume that $r + |a_r| - 1 < k$ and $a_r \neq b_r$ for some $r \in [k]$.
Then, we can define a winning strategy for Spoiler as follows:
Let $a_r = \mathtt{a}_1 \cdot \mathtt{a}_2 \cdots \mathtt{a}_l$, where $\mathtt{a}_i \in \Sigma$ for $i \in [l]$.
\begin{itemize}
\item For round $r+1$, Spoiler chooses $\mathfrak{A}_w$, and $\mathtt{a}_1 \cdot \mathtt{a}_2$,
\item for round $r+2$, Spoiler chooses $\mathfrak{A}_w$, and $\mathtt{a}_1 \cdot \mathtt{a}_2 \cdot \mathtt{a}_3$,
\item $\dots$
\item for round $r+l-1$, Spoiler chooses $\mathfrak{A}_w$, and $\mathtt{a}_1 \cdot \mathtt{a}_2 \cdots \mathtt{a}_l$.
\end{itemize}

For $r+1$, Duplicator must respond with $b_{r+1} \df \mathtt{a}_1 \cdot \mathtt{a}_2$ since $(\vec a, \vec b)$ forms a partial isomorphism and $a_{r+1}$ is the concatenation of two constant symbols.
Therefore, $a_{r+1} = b_{r+1}$.
Then, for each $j \in \{r+2, \dots, r+l-1 \}$, we have that $a_j = a_{j-1} \cdot \mathtt{a}$ for some $\mathtt{a} \in \Sigma$. 
Again, since $(\vec a, \vec b)$ forms a partial isomorphism, $b_j = b_{j-1} \cdot \mathtt{a}$.

Therefore, $a_{r+l-1} = a_r$ and $b_{r+l-1} = a_r$.
However, $b_{r+l-1} \neq b_r$ since $a_r \neq b_r$.
Hence, $a_{r+l-1} = a_r$ and $b_{r+l-1} \neq b_r$.
Thus, Spoiler has won $\mathcal{G}$.
This is a contradiction, since $(\vec a,\vec b)$ is a partial isomorphism.
Consequently, our assumption that $a_r \neq b_r$ cannot hold.

Analogously, if $r + |b_r| - 1 < k$, then it must hold that $a_r = b_r$.
\end{proof}

\subsection{Proof of Lemma~\ref{lemma:prefixSuffix}}
\begin{proof}
Let $\mathfrak{A}_w$ and $\mathfrak{B}_v$ be two $\signature_\Sigma$-structures that represent $w \in \Sigma^*$ and $v \in \Sigma^*$ respectively, where $\mathfrak{A}_w \equiv_k \mathfrak{B}_v$.
We denote the universe of $\mathfrak{A}_w$ and $\mathfrak{B}_v$ by $A$ and $B$ respectively.
Let $(\vec a, \vec b)$ be the $k+|\Sigma|+1$-tuples resulting from a $k$-round game over $\mathfrak{A}_w$ and $\mathfrak{B}_v$ where Duplicator plays their winning strategy.
Thus, we know that $(\vec a, \vec b)$ forms a partial isomorphism.
To distinguish individual components of $\vec a$ and $\vec b$, let $\vec a = (a_1, a_2, \dots, a_{k+|\Sigma|+1})$ and let $\vec b = (b_1, b_2, \dots, b_{k+|\Sigma|+1})$, where for each $i \in [k+|\Sigma|+1]$, we have that $a_i \in A$ and $b_i \in B$.

Working towards a contradiction, assume that $a_r$ is a suffix of $w$ and $b_r$ is not a suffix of $v$ for some $r \leq k-2$.
We now show that Spoiler can use rounds $k-1$ and $k$ to win the $k$-round game over $\mathfrak{A}_w$ and $\mathfrak{B}_v$.

For round $k-1$, let Spoiler choose $\mathfrak{A}_w$ and $w$. 
Duplicator responds with some $u \sqsubseteq v$. We look at two separate cases based on whether $u = v$, or not.

\proofsubparagraph{Case 1, $u \neq v$:}
For round $k$, let Spoiler choose $\mathfrak{B}_v$ and $b_{k} \df u \cdot \mathtt{a} $ such that $\mathtt{a} \in \Sigma$ and $u \mathtt{a} \sqsubseteq v$.
It follows that $b_{k} = b_{k-1} \cdot \mathtt{a}$ yet there does not exist $u' \sqsubseteq w$ such that $u' = a_r \cdot \mathtt{a}$.
Consequently, Spoiler has won and we have thus reached a contradiction.

\proofsubparagraph{Case 2, $u=v$:}
For round $k$, let Spoiler choose $\mathfrak{A}_w$ and $a_k$ such that $a_{k-1} = a_k \cdot a_r$.
Duplicator must choose $b_k$ such that $b_{k-1} = b_k \cdot b_r$.
Thus it follows that $b_r$ is a suffix of $v$ and consequently, we have reached a contradiction.

\proofsubparagraph{Concluding the proof.}
Note that we have not considered the case where $a_r$ is a prefix of $w$, however this follows using the analogous reasoning.
Thus, so far we have shown that if $r \leq k-2$ and $a_r$ is a prefix (or suffix) of $w$, then $b_r$ is a prefix (or suffix respectively) of $v$.
However, $\mathfrak{A}_w$ and $\mathfrak{B}_v$ are arbitrarily named, and hence it immediately follows that if $r \leq k-2$ and $b_r$ is a prefix (or suffix) of $v$, then $a_r$ is a prefix (or suffix respectively) of $w$.
\end{proof}

\subsection{Proof of Lemma~\ref{lemma:congruence}}
We first give a necessary definition.

\begin{definition}
Let $\mathfrak{A}_{w}$ be a $\signature_\Sigma$-structures with the universe $A$ .
For any $A' \subset A$, let $\mathfrak{A}_w|_{A'}$ be the $\signature_\Sigma$-structure 
\[ (A' \union \{ \perp \}, \concrel', \mathtt{a}_1^{\mathfrak{A}_w|_{A'}}, \dots, \mathtt{a}_n^{\mathfrak{A}_w|_{A'}}, \emptyword^{\mathfrak{A}_w|_{A'}}) \]
where $\concrel' \subseteq (A')^3$ denotes concatenation restricted to elements of $A'$, and $\Sigma = \{ \mathtt{a}_1, \dots, \mathtt{a}_n \}$, and where the constants are interpreted as usual.
\end{definition}

In other words, $\mathfrak{A}_w |_{A'}$ is the restriction of the structure $\mathfrak{A}_w$ to the sub-universe $A' \union \{ \perp \}$.

We are now ready for the proof of~\cref{lemma:congruence}.

\begin{proof}
Let $w_1, w_2, v_1, v_2 \in \Sigma^*$ where $\facts(w_1) \intersect \facts(w_2) = \facts(v_1) \intersect \facts(v_1)$, and let $r \df \mathsf{max} \{|u| \in \mathbb{N} \mid u \in \facts(w_1) \intersect \facts(w_2) \}$.
Furthermore, assume that $w_1 \equiv_{k+ r + 2} v_1$ and $w_2 \equiv_{k+r + 2} v_2$ for some $k \in \mathbb{N}_+$.

Let $\mathfrak{A}_w$ be a $\signature_\Sigma$-structure for $w \df w_1 \cdot w_2$ and let $\mathfrak{B}_v$ be a $\signature_\Sigma$-structure for $v \df v_1 \cdot v_2$.
Let $A$ and $B$ be the universes for $\mathfrak{A}_w$ and $\mathfrak{B}_v$ respectively.

Consider the following subsets of $A$ and $B$:
\begin{align*}
A_\mathsf{other} \df & A \setminus \bigl( \facts(w_1) \union \facts(w_2) \union \{ \perp \} \bigr) \text{ and} \\
B_\mathsf{other} \df & B \setminus \bigl( \facts(v_1) \union \facts(v_2) \union \{ \perp \} \bigr).
\end{align*}
That is, $A_\mathsf{other}$ contains those factors of $w_1 \cdot w_2$ that are not factors of either $w_1$ or $w_2$.
Likewise, $B_\mathsf{other}$ contains those factors of $v_1 \cdot v_2$ that are not factors of either $v_1$ or $v_2$. 

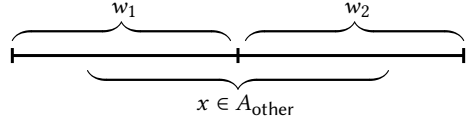
\begin{figure}
\centering
\begin{tikzpicture}[fill=white]
\tikzset{
    position label/.style={
       below = 3pt, 
       text height = 1.5ex,
       text depth = 1ex
    },
   brace/.style={
     decoration={brace, mirror},
     decorate
   }
}

\draw [-, line width=0.4mm] (-3,0) -- (3,0);
\draw [-, line width=0.4mm] (0, -.1) -- (0,.1);
\draw [-, line width=0.4mm] (3, -.1) -- (3,.1);
\draw [-, line width=0.4mm] (-3, -.1) -- (-3,.1);
\draw [decorate,decoration = {calligraphic brace, amplitude=7}, line width=0.3mm] (-3, 0.2) -- (-0.1, 0.2) node [above, xshift= -40pt, yshift=5.5pt] {$w_1$};
\draw [decorate,decoration = {calligraphic brace, amplitude=7}, line width=0.3mm] (0.1, 0.2) -- (3, 0.2) node [above, xshift= -40pt, yshift=5.5pt] {$w_2$};
\draw [decorate,decoration = {calligraphic brace, amplitude=7}, line width=0.3mm] (2, -0.2) -- (-2, -0.2) node [below, xshift= 60pt, yshift= -5.5pt] {$x \in A_\mathsf{other}$};

\end{tikzpicture}
\caption{Illustration of factors in $A_\mathsf{other}$. The factors of $B_\mathsf{other}$ are analogously illustrated by replacing $w_1$ and $w_2$ with $v_1$ and $v_2$ respectively.}
\label{fig:other}
\end{figure}

Notice that for any $x \in A_\mathsf{other}$, we have that $x = x_1 \cdot x_2$, where $x_1$ is a suffix of $w_1$ and $x_2$ is a prefix of $w_2$.
Analogously, for any $y \in B_\mathsf{other}$, we have that $y = y_1 \cdot y_2$, where $y_1$ is a suffix of $v_1$ and $y_2$ is a prefix of $v_2$.
For an illustration, see~\cref{fig:other}.

We now define two functions 
\begin{align*}
f_\mathsf{split} \colon & A_\mathsf{other} \rightarrow \facts(w_1) \times \facts(w_2) \text{ and} \\
g_\mathsf{split} \colon & B_\mathsf{other} \rightarrow \facts(v_1) \times \facts(v_2).
\end{align*}
Let $f_\mathsf{split}(x) = (x_1, x_2)$ where $x = x_1 \cdot x_2$ and where $x_1 \in \facts(w_1)$ and $x_2 \in \facts(w_2)$.
Let $g_\mathsf{split}(y) = (y_1, y_2)$ where $y = y_1 \cdot y_2$ and where $y_1 \in \facts(v_1)$ and $y_2 \in \facts(v_2)$.
The precise definition of these functions is not important, and these functions are only used to get a unique pair from $\facts(w_1)$ and $\facts(w_2)$ (or $\facts(v_1)$ and $\facts(v_2)$) from a given element of $A_\mathsf{other}$ (or $B_\mathsf{other}$).

Note that since $w_1 \equiv_{k+r+2} v_1$ and $w_2 \equiv_{k+r+2} v_2$ for some $k \in \mathbb{N}_+$, we have that $\mathfrak{A}_w |_{\facts(w_1)} \equiv_{k+r+2} \mathfrak{B}_v |_{\facts(v_1)}$, and $\mathfrak{A}_w |_{\facts(w_2)} \equiv_{k+r+2} \mathfrak{B}_v |_{\facts(v_2)}$.
Let $\mathcal{G}$ be a $k$-round game over $\mathfrak{A}_w$ and $\mathfrak{B}_v$.
We assume that for each round of $\mathcal{G}$, Spoiler does not choose $\perp$ since Duplicator can always respond with $\perp$.
Therefore, we only look at Duplicator's strategy for the case where Spoiler chooses some factor of $w$ or $v$.
Duplicator's strategy to win $\mathcal{G}$ is based upon two look-up games.

\proofsubparagraph{The look-up games.}
Let $\mathcal{G}_1$ be a $k+r+2$ round look-up game over $\mathfrak{A}_w |_{\facts(w_1)}$ and $\mathfrak{B}_v |_{\facts(v_1)}$.
Let $\mathcal{G}_2$ be a $k+r+2$ round look-up game over $\mathfrak{A}_w |_{\facts(w_2)}$ and $\mathfrak{B}_v |_{\facts(v_2)}$.
To help streamline the proof of correctness, we shall assume that Spoiler can skip certain rounds in $\mathcal{G}_1$ and $\mathcal{G}_2$.
That is, Spoiler can decide not to make a move for some round, and Duplicator therefore does not need to respond.
In reality, this can be easily done by Spoiler choosing $\perp$ for example; however, we assume Spoiler can skip as it is simpler for us to deal with.

\emph{Spoiler's Choice in $\mathcal{G}_1$ and $\mathcal{G}_2$:}
We now give Spoiler's choice in $\mathcal{G}_\iota$ for each $\iota \in \{ 1, 2 \}$.
For every $p \in [k]$, Spoiler's choice for round $p$ in $\mathcal{G}_\iota$ is uniquely determined from Spoiler's choice in $\mathcal{G}$.
Recall that $\mathcal{G}_\iota$ is a $k+r+2$ round game over $\mathfrak{A}_{w} |_{\facts(w_\iota)}$ and $\mathfrak{B}_v |_{\facts(v_\iota)}$, for $\iota \in \{ 1, 2 \}$.
Spoiler's choice in round $p \in [k]$ of $\mathcal{G}_\iota$ is defined as follows:
First, we consider the case where in round $p$ of $\mathcal{G}$, Spoiler chooses $\mathfrak{A}_w$.
Then, in round $p$ of $\mathcal{G}_\iota$, Spoiler chooses the structure $\mathfrak{A}_w |_{\facts(w_\iota)}$ and the factor Spoiler chooses is defined as follows:
\begin{itemize}
\item If in $\mathcal{G}$ Spoiler chooses some $u \in \facts(w_\iota)$, then Spoiler chooses $u$ in round $p$ of $\mathcal{G}_\iota$,
\item if in $\mathcal{G}$ Spoiler chooses some $u \in A_\mathsf{other}$ where $f_\mathsf{split}(u) = (u_1,u_2)$, then 
\begin{itemize}
\item Spoiler chooses $u_1$ in $\mathcal{G}_\iota$ if $\iota = 1$,
\item Spoiler chooses $u_2$ in $\mathcal{G}_\iota$ if $\iota = 2$,
\end{itemize}
\item if in $\mathcal{G}$ Spoiler chooses some $u \in \facts(w_\varrho) \setminus \facts(w_\iota)$ where $\varrho \in [2] \setminus \{ \iota \}$, then Spoiler chooses to skip round $p$ in $\mathcal{G}_\iota$.
\end{itemize}

Next, we consider the case where Spoiler chooses $\mathfrak{B}_v$ in round $p$ of $\mathcal{G}$, where $r \leq k$.
For this case, we have that Spoiler chooses $\mathfrak{B}_v |_{\facts(v_\iota)}$ in round $p$ of $\mathcal{G}$.
\begin{itemize}
\item If in $\mathcal{G}$ Spoiler chooses some $u \in \facts(v_\iota)$, then Spoiler chooses $u$ in round $p$ of $\mathcal{G}_\iota$,
\item if in $\mathcal{G}$ Spoiler chooses some $u \in B_\mathsf{other}$ where $g_\mathsf{split}(u) = (u_1,u_2)$, then 
\begin{itemize}
\item Spoiler chooses $u_1$ in $\mathcal{G}_\iota$ if $\iota = 1$,
\item Spoiler chooses $u_2$ in $\mathcal{G}_\iota$ if $\iota = 2$,
\end{itemize}
\item if in $\mathcal{G}$ Spoiler chooses some $u \in \facts(v_\varrho) \setminus \facts(v_\iota)$ where $\varrho \in [2] \setminus \{ \iota \}$, then Spoiler chooses to skip round $p$ in $\mathcal{G}_\iota$.
\end{itemize}

\emph{A Short Discussion on $\mathcal{G}_1$ and $\mathcal{G}_2$:}
For rounds $k+1$ to $k+r+2$, we assume Spoiler chooses any structure and any factor in both $\mathcal{G}_1$ and $\mathcal{G}_2$.
This ensures that in any round $p \in [k]$ of $\mathcal{G}_1$ and $\mathcal{G}_2$, if Spoiler chooses some $u \in \facts(w_1) \intersect \facts(w_2)$ (or $u \in \facts(v_1) \intersect \facts(v_2)$), we have that Duplicator responds with $u$.
This is because $|u| \leq r$, and Duplicator has a winning strategy for the $k+r+2$-round games $\mathcal{G}_1$ and $\mathcal{G}_2$, see~\cref{lemma:consistentStrats}.
After Spoiler has made their choice in $\mathcal{G}_1$ and $\mathcal{G}_2$, Duplicator responds in both games using their winning strategy.

Let $\vec a_1 = (a_{1,1}, a_{1,2}, \dots, a_{1,k+|\Sigma|+1})$ and let $\vec b_1 = (b_{1,1}, b_{1,2}, \dots, b_{1,k+|\Sigma|+1})$ be the tuples generated from the first $k$-rounds of $\mathcal{G}_1$, along with the interpreted constants.
Let $\vec a_2 = (a_{2,1}, a_{2,2}, \dots, a_{2,k+|\Sigma|+1})$ and let $\vec b_2 = (b_{2,1}, b_{2,2}, \dots, b_{2,k+|\Sigma|+1})$ be the tuples generated from the first $k$-rounds of $\mathcal{G}_2$, along with the interpreted constants.
If Spoiler chose to skip round $p \in [k]$ in $\mathcal{G}_1$ (or $\mathcal{G}_2$), then we say that $a_{1,p} = b_{1,p} = \perp$ (or $a_{2,p} = b_{2,p} = \perp$).
Note that skipping a round has no bearing on whether the generated tuples are a partial isomorphism.
Therefore, since $\mathfrak{A}_w |_{\facts(w_1)} \equiv_{k+r+2} \mathfrak{B}_v |_{\facts(v_1)}$, and $\mathfrak{A}_w |_{\facts(w_2)} \equiv_{k+r+2} \mathfrak{B}_v |_{\facts(w_2)}$, it follows that $( \vec a_1, \vec b_1)$ and $( \vec a_2, \vec b_2)$ both form a partial isomorphism. 

Before moving on to define Duplicator's strategy in $\mathcal{G}$, we make some remarks about the tuples generated from $\mathcal{G}_1$ and $\mathcal{G}_2$.
For each $p \in [k]$, let $d_{1,p}$ be the factor that Duplicator chose in round $p$ of $\mathcal{G}_1$, and let $d_{2,p}$ be the factor that Duplicator chose in round $p$ of $\mathcal{G}_2$.
From the definition of Spoiler's choice in $\mathcal{G}_1$ and $\mathcal{G}_2$ we have that:
\begin{itemize}
\item If Spoiler skipped round $p$ in $\mathcal{G}_1$, then for round $p$ in $\mathcal{G}$, Spoiler chose some $u \in \facts(w_2) \setminus \facts(w_1)$, or some $u \in \facts(v_2) \setminus \facts(v_1)$. Analogously, if Spoiler skipped round $p$ in $\mathcal{G}_2$, then for round $p$ in $\mathcal{G}$, Spoiler chose some $u \in \facts(w_1) \setminus \facts(w_2)$, or some $u \in \facts(v_1) \setminus \facts(v_2)$.
\item If in round $p$ of $\mathcal{G}$, Spoiler chose some $u \in \facts(w_1) \intersect \facts(w_2)$, or some $u \in \facts(v_1) \intersect \facts(v_2)$, then it follows that $d_{1,p} = d_{2,p} = u$. This is because we know that $|u| \leq r$, and Duplicator must be able to win after $k+r+2$-rounds. Therefore, we can use~\cref{lemma:consistentStrats} to determine that $d_{1,p} = d_{2,p} = u$ must hold.
\item The final case is when Spoiler chose some $u \in A_\mathsf{other}$, or $u \in B_\mathsf{other}$. Note that this is the only case where Spoiler does not skip round $p$ in both $\mathcal{G}_1$ and $\mathcal{G}_2$, and where $d_{1,p} \neq d_{2,p}$.
\end{itemize}

\proofsubparagraph{Duplicator's strategy.}
Duplicator derives their response in $\mathcal{G}$ using their responses in $\mathcal{G}_1$ and $\mathcal{G}_2$.
For each round $p \in [k]$, the structure Duplicator chooses is always the opposite structure to what Spoiler has chosen, therefore, we only look at which factor Duplicator chooses.
We use $d_{\iota,p}$ for $\iota \in \{ 1,2 \}$ to denote Duplicator's response to Spoiler in round $p$ of $\mathcal{G}_\iota$.
Duplicator response to round $p \in [k]$ in $\mathcal{G}$ is as follows:
\begin{itemize}
\item If Spoiler skipped round $p$ in $\mathcal{G}_1$, then Duplicator's choice in round $p$ of $\mathcal{G}$ is $d_{2,p}$;
\item if Spoiler skipped round $p$ in $\mathcal{G}_2$, then Duplicator's choice in round $p$ of $\mathcal{G}$ is $d_{1,p}$;
\item if Spoiler chose some $u \in \facts(w_1) \intersect \facts(w_2)$, or some $u \in \facts(v_1) \intersect \facts(v_2)$ in $\mathcal{G}$, then Duplicator responds with $d_{1,p}$. Recall that $d_{1,p} = d_{2,p} = u$ for this case; and
\item if Spoiler chose some $u \in A_\mathsf{other}$, or $u \in B_\mathsf{other}$, then Duplicator responds with $d_{1,p} \cdot d_{2,p}$. Note that due to~\cref{lemma:prefixSuffix}, it must hold that $d_{1,p}$ is a suffix of $w_1$ (or $v_1$), and $d_{2,p}$ is a prefix of $w_2$ (or $v_2$ respectively). Refer back to~\cref{fig:other} for an illustration of elements in $A_\mathsf{other}$ and $B_\mathsf{other}$. Therefore $d_{1,p} \cdot d_{2,p}$ is always an element of the structure in which Duplicator plays round $p$.
\end{itemize}
This gives a complete strategy for Duplicator in $\mathcal{G}$.

To sum up Duplicator's strategy for $\mathcal{G}$ informally: If Spoiler chooses an element from either $\facts(w_1)$ or $\facts(v_1)$, then Duplicator responds with their winning strategy for $\mathcal{G}_1$.
If Spoiler chooses from either $\facts(w_2)$ or $\facts(v_2)$, then Duplicator responds with their winning strategy for $\mathcal{G}_2$.
Finally, if Spoiler chooses from $A_\mathsf{other}$ or $B_\mathsf{other}$, then we split the factor Spoiler chose into two factors from $\facts(w_1) \times \facts(w_2)$, or $\facts(v_1) \times \facts(v_2)$, and Duplicator responds with the concatenation of their winning strategy from $\mathcal{G}_1$ and $\mathcal{G}_2$.

\proofsubparagraph{Correctness.}
Our next focus is showing that the strategy just defined for Duplicator is indeed a winning strategy for the game $\mathcal{G}$.
To that end, let $\vec a \df (a_1, a_2, \dots, a_{k+|\Sigma|+1})$ and $\vec b \df (b_1, b_2, \dots, b_{k+|\Sigma|+1})$, where $a_i \in A$ and $b_i \in B$ for all $i \in [k+|\Sigma|+1]$, be the tuples that result from $\mathcal{G}$ along with the interpreted constants, where Duplicator plays the strategy we have just defined.

Recall that $\vec a_1 = (a_{1,1}, a_{1,2}, \dots, a_{1,k+|\Sigma|+1})$ and $\vec b_1 = (b_{1,1}, b_{1,2}, \dots, b_{1,k+|\Sigma|+1})$ are the tuples generated from the first $k$ rounds of $\mathcal{G}_1$, and $\vec a_2 = (a_{2,1}, a_{2,2}, \dots, a_{2,k+|\Sigma|+1})$ and $\vec b_2 = (b_{2,1}, b_{2,2}, \dots, b_{2,k+|\Sigma|+1})$ are the tuples generated from the first $k$ rounds of $\mathcal{G}_2$.
Also recall, that the last $|\Sigma|+1$ components of these tuples are the interpretations of the constant symbols.
We know that $(\vec a_1, \vec b_1)$ and $(\vec a_2, \vec b_2)$ both form a partial isomorphism.

While Duplicator's strategy for $\mathcal{G}$ is based upon their strategy for the first $k$-rounds of $\mathcal{G}_1$ and $\mathcal{G}_2$, Duplicator can survive $k+r+2$-rounds of $\mathcal{G}_1$ and $\mathcal{G}_2$.
This ensures that Duplicator makes certain decisions in $\mathcal{G}_1$ and $\mathcal{G}_2$ which ensures that Duplicator's strategy for $\mathcal{G}$ is indeed a winning strategy.
Thus, to prove that Duplicator's strategy for $\mathcal{G}$ is a winning strategy, we consider Spoiler's choices in $\mathcal{G}_1$ and $\mathcal{G}_2$ for rounds $k+1$ up to $k+r+2$.
In fact, we only look at rounds $k+1$, $k+2$, and $k+3$, since we leave the last $r-1$ rounds to ensure that if Spoiler chooses some $u$ such that $|u| \leq r$, then Duplicator responds with $u$.

Let $\kappa \df k+|\Sigma|+1$.
Then, we use $a_{1,\kappa +i}$ for $i \in [r+2]$ to denote Spoiler's choice in round $k+i$ (assuming Spoiler chooses the structure $\mathfrak{A}_w |_{\facts(w_1)}$). 
Thus, after all $k+r+2$ rounds, we have that the first $k$-rounds of $\mathcal{G}_1$ and $\mathcal{G}_2$ are in the first $k$-components of the corresponding tuples, the next $|\Sigma|+1$ components are then the interpreted constants, and finally, the last $r+2$ components are the chosen elements for the last $r+2$-rounds.
This is simply a permutation of the resulting tuples, and therefore $(\vec a_1, \vec b_1)$ and $(\vec a_2, \vec b_2)$ still form a partial isomorphism.

To help with some subsequent reasoning, we now prove the following claim:
\begin{claim}\label{claim:Aother}
For every $p \in [k]$, we have that $a_p \in A_\mathsf{other}$ if and only if $b_p \in B_\mathsf{other}$.
\end{claim}
\begin{claimproof}
Without loss of generality, assume that in round $p$ of $\mathcal{G}$, Spoiler chose $a_p \in A_\mathsf{other}$.
Therefore, we have that $a_p = a_{1,p} \cdot a_{2,p}$ where $a_{1,p} \in \facts(w_1)$ and $a_{2,p} \in \facts(w_2)$.
It follows that Duplicator responds with $b_p = b_{1,p} \cdot b_{2,p}$, where $b_{1,p} \in \facts(v_1)$ and $b_{2,p} \in \facts(v_2)$.
Working towards a contradiction, assume (without loss of generality) that $b_p \in \facts(v_1)$.
It follows that $b_{2,p} \in \facts(v_1)$ holds, and therefore $b_{2,p} \in \facts(w_1) \intersect \facts(w_2)$.
Since it therefore holds that $|b_{2,p}| \leq r$, we have that $a_{2,p} = b_{2,p}$, see~\cref{lemma:consistentStrats}.
Hence, $a_{2,p} \in \facts(w_1) \intersect \facts(w_2)$ holds.

We now look at the game $\mathcal{G}_1$ for rounds $k+1$ to $k+r+2$.
Recall that Duplicator plays $\mathcal{G}_1$ with their winning strategy.
\begin{itemize}
\item In round $k+1$, Spoiler can choose $\mathfrak{B}_w |_{\facts(v_1)}$ and $b_{1,\kappa+1} \df a_{2,p}$. 
Duplicator must respond with $a_{1,\kappa +1} = a_{2,p}$.
This is because $| a_{2,p} | \leq r$, since $a_{2,p} \in \facts(w_1) \intersect \facts(w_2)$, and Duplicator must be able to survive the next $r$ rounds, see~\cref{lemma:consistentStrats}..
\item In round $k+2$, Spoiler can choose $\mathfrak{B}_w |_{\facts(v_1)}$ and $b_{1,\kappa+2} \df b_{1,p} \cdot b_{1,\kappa+1}$.
We note that Spoiler can choose $b_{1,\kappa+2}$ since $b_{1,\kappa+2} = b_p$ and we have assumed that $b_p \in \facts(v_1)$.
Duplicator must respond with $a_{1,\kappa+2} = a_{1,p} \cdot a_{1,\kappa+1}$.
\end{itemize}
However, $a_{1,\kappa+2} = a_{1,p} \cdot a_{2,p}$ and therefore $a_{1,\kappa+2} = a_p$, which implies $a_{1,\kappa+2} \in \facts(w_1)$.
This is a contradiction since we have assumed $a_p \in A_\mathsf{other}$, and therefore $a_p \notin \facts(w_1)$.
\end{claimproof}

To show that $(\vec a, \vec b)$ is indeed a partial isomorphism, we show that the following conditions hold:
\begin{itemize}
\item For every $i \in [k+|\Sigma|+1]$, and constant symbol $c \in \signature_\Sigma$, we have $a_i = c^{\mathfrak{A}_w}$ if and only if $b_i = c^{\mathfrak{B}_v}$,
\item for every $i,j \in [k+|\Sigma|+1]$, we have that $a_i = a_j$ if and only if $b_i = b_j$, and
\item for every $i,j,l \in [k+|\Sigma|+1]$, we have that $a_l = a_i \cdot a_j$ if and only if $b_l = b_i \cdot b_j$.
\end{itemize}

First note that directly from the definition of Duplicator's strategy, we have that $a_i \in \Sigma$ if and only if $b_i \in \Sigma$ and $a_i = b_i$.
Furthermore, for every $i,j \in [k+|\Sigma|+1]$, we have that $a_i = a_j$ if and only if $a_i = a_j \cdot \emptyword$.
Symmetrically, for every $i,j \in [k+|\Sigma|+1]$, we have that $b_i = b_j$ if and only if $b_i = b_j \cdot \emptyword$.

Thus, if $a_i \cdot a_j = a_l$ if and only if  $b_i \cdot b_j = b_l$ for all $a_i, a_j, a_l \in \vec a$ and for all $b_i,b_j,b_l \in \vec b$, then $(\vec a, \vec b)$ forms a partial isomorphism.
Hence, the following claim completes the proof of correctness for Duplicator's strategy:

\begin{claim}
For every $i,j,l \in [k+|\Sigma|+1]$, we have that $a_i \cdot a_j = a_l$ if and only if $b_i \cdot b_j = b_l$.
\end{claim}
\begin{claimproof}
Without loss of generality, assume that $a_i \cdot a_j = a_l$.
We consider three cases based on $a_l$, and we shall prove that each case yields $b_i \cdot b_j = b_l$.
Before doing so, we state some facts.

From the definition of Spoiler's and Duplicator's choices in $\mathcal{G}$, $\mathcal{G}_1$ and $\mathcal{G}_2$, we have that the following conditions hold for any $p \in [k]$ and any $\iota \in \{ 1, 2\}$:
\begin{itemize}
\item $a_p \in \facts(w_\iota)$ if and only if $b_p \in \facts(v_\iota)$,
\item if $a_p \in \facts(w_\iota)$, then $a_p  = a_{\iota,p}$, and
\item if $b_p \in \facts(v_\iota)$, then $b_p  = b_{\iota,p}$.
\end{itemize}

Likewise, for any $p \in [k]$ we know that
\begin{itemize}
\item $a_p \in A_\mathsf{other}$ if and only if $b_p \in B_\mathsf{other}$ (see~\cref{claim:Aother}),
\item if $a_p \in A_\mathsf{other}$, then $a_p = a_{1,p} \cdot a_{2,p}$, and
\item if $b_p \in B_\mathsf{other}$, then $b_p = b_{1,p} \cdot b_{2,p}$. 
\end{itemize}

\proofsubparagraph{Case 1, $a_l \in \facts(w_1)$:}
If $a_l \in \facts(w_1)$, then we know that $a_i, a_j \in \facts(w_1)$, and from the definition of Duplicator's strategy, we know that $b_l, b_i, b_j \in \facts(v_1)$.
Therefore, it follows that $a_s = a_{1,s}$ and $b_s = b_{1,s}$ for $s \in \{ i,j,l \}$.
Since $(\vec a_1, \vec b_1)$ forms a partial isomorphism, and $a_{1,l} = a_{1,i} \cdot a_{1,j}$, we know that $b_{1,l} = b_{1,i} \cdot b_{1,j}$.

\proofsubparagraph{Case 2, $a_l \in \facts(w_2)$:}
This case is symmetric to Case 1. 

\proofsubparagraph{Case 3, $a_l \in A_\mathsf{other}$:} 
Note that $a_i, a_j \in A_\mathsf{other}$ cannot hold.
This is because $a_l \in A_\mathsf{other}$, and therefore, there exists $x \in \facts(w_1)$ and $y \in \facts(w_2)$ such that $a_l = x \cdot y$.
Thus, from the fact that $a_l = a_i \cdot a_j = x \cdot y$, using a length argument we can determine that $a_i \in \facts(w_1)$ or $a_j \in \facts(w_2)$.
Thus, we split Case 3 into the following subcases:
\begin{description}
\item[Case 3.1,] $a_i \in \facts(w_1)$ and $a_j \in \facts(w_2)$.
\item[Case 3.2,] $a_i \in A_\mathsf{other}$ and $a_j \in \facts(w_2)$.
\item[Case 3.3,] $a_i \in \facts(w_1)$ and $a_j \in A_\mathsf{other}$.
\end{description}

We now prove that for each of the above mentioned subcases, our assumption that $a_i \cdot a_j = a_l$ yields $b_i \cdot b_j = b_l$.

\proofsubparagraph{Case 3.1, $a_i \in \facts(w_1)$ and $a_j \in \facts(w_2)$:}
As $a_i \in \facts(w_1)$, it follows that $a_i = a_{1,i}$ and consequently, $b_i \in \facts(v_1)$ which implies that $b_i = b_{1,i}$.
Analogously, since $a_j \in \facts(w_2)$, we know that $a_j = a_{2,j}$ and consequently, $b_j \in \facts(v_2)$ which implies that $b_j = b_{2,j}$.

As per the definition of Duplicator's strategy and the look-up games $\mathcal{G}_1$ and $\mathcal{G}_2$, we know that $a_l = a_{1,l} \cdot a_{2,l}$ and $b_l = b_{1,l} \cdot b_{2,l}$.
Since $a_l = a_i \cdot a_j$, it follows that $a_{1,l} \cdot a_{2,l} = a_{1,i} \cdot a_{2,j}$.

Due to the fact that $a_{1,l} \cdot a_{2,l} = a_{1,i} \cdot a_{2,j}$, we have that either $a_{1,l}$ is a (not necessarily strict) prefix of $a_{1,i}$, or $a_{1,i}$ is a prefix of $a_{1,l}$.
This gives us two further subcases to look at.

\subparagraph*{Case 3.1.1, $a_{1,l}$ is a prefix of $a_{1,i}$:}
For this case, we can write $a_{1,i} = a_{1,l} \cdot x$ for some $x \in \Sigma^*$. 
Therefore
$  a_{1,l} \cdot a_{2,l} = a_{1,l} \cdot x \cdot a_{2,j} $
and hence $a_{2,l} = x \cdot a_{2,j}$. 
Since $a_{1,i} \in \facts(w_1)$ and $x \sqsubseteq a_{1,i}$, we know that $x \in \facts(w_1)$.
Furthermore, $a_{2,l} \in \facts(w_2)$, and $x \sqsubseteq a_{2,l}$, and consequently $x \in \facts(w_2)$ also holds.
Thus $x \in \facts(w_1) \intersect \facts(w_2)$.

We now show that (for Case 3.1.1) if Duplicator plays their winning strategy for $\mathcal{G}_1$ and $\mathcal{G}_2$, and $a_l = a_i \cdot a_j$, then $b_l = b_i \cdot b_j$ must hold.
To show this, we look at possible choices for Spoiler in round $k+1$ up to round $k+r+2$ for $\mathcal{G}_1$ and $\mathcal{G}_2$.
Recall that $\kappa \df k+|\Sigma|+1$, and we use the last $r+2$ components of the tuples $\vec a_1$, $\vec a_2$, $\vec b_1$, and $\vec b_2$ to denote the choices made in the last $r+2$ rounds of $\mathcal{G}_1$ and $\mathcal{G}_2$ respectively.

\begin{itemize}
\item For round $k+1$ of $\mathcal{G}_1$, Spoiler can choose $\mathfrak{A}_w |_{\facts(w_1)}$ and $a_{1, \kappa+1} \df x$.
Since $x \in \facts(w_1) \intersect \facts(w_2)$, we know that $|x| \leq r$, and therefore Duplicator must respond with $\mathfrak{B}_v |_{\facts(v_1)}$ and $b_{1,\kappa+1} \df x$.
This is because Duplicator must be able to survive the next $r$ rounds, see~\cref{lemma:consistentStrats}.
Since we know that $a_{1,i} = a_{1,l} \cdot a_{1,\kappa+1}$, and Duplicator plays their winning strategy for the $k+r+2$ round game $\mathcal{G}_1$, we know that $b_{1,i} = b_{1,l} \cdot b_{1,\kappa+1}$ must also hold.
\item For round $k+1$ of $\mathcal{G}_2$, Spoiler can choose $\mathfrak{A}_w |_{\facts(w_2)}$ and $a_{2,\kappa+1} \df x$.
Duplicator must respond with $\mathfrak{B}_v |_{\facts(v_1)}$ and $b_{2,\kappa+1} \df x$.
Since $a_{2,l} = a_{2,\kappa+1} \cdot a_{2,j}$, we know that $b_{2,l} = b_{2,\kappa+1} \cdot b_{2,j}$ must hold.
\end{itemize}

Since $b_{2,\kappa+1} = b_{1,\kappa+1} = x$, we can write:
\[  b_{1,l} \cdot b_{2,\kappa+1} \cdot b_{2,j}= b_{1,l} \cdot b_{1,\kappa+1} \cdot b_{2,j} . \]
We know that $b_{2,l} = b_{2,\kappa+1} \cdot b_{2,j}$ and $b_{1,i} = b_{1,l} \cdot b_{1,\kappa+1}$, we have that
\[  b_{1,l} \cdot \underbrace{b_{2,\kappa+1} \cdot b_{2,j}}_{b_{2,l}} = \underbrace{b_{1,l} \cdot b_{1,\kappa+1}}_{b_{1,i}} \cdot b_{2,j} . \]
Furthermore, $b_l = b_{1,l} \cdot b_{2,l}$ and $b_i \cdot b_j = b_{1,i} \cdot b_{2,j}$.
Consequently, $b_l = b_i \cdot b_j$.

All of the subsequent cases have a similar proof to the proof of Case 3.1.2.
Unfortunately, they are not similar enough to hand-wave away.
If the reader is convinced that the proof of Case of 3.1.2 can be adapted to the other various cases, they are invited to skip the rest of the proof.
However, we include proofs of all the cases for completeness sake (and to avoid any pitfalls).

\subparagraph*{Case 3.1.2, $a_{1,i}$ is a prefix of $a_{1,l}$:}
We can write $a_{1,l} = a_{1,i} \cdot x$ for some $x \in \Sigma^*$.
Therefore
$ a_{1,i} \cdot x \cdot a_{2,l} = a_{1,i} \cdot a_{2,j}$,
and hence $a_{2,j} = x \cdot a_{2,l}$.
Since $x \sqsubseteq a_{2,j}$, and $a_{2,j} \in \facts(w_2)$, we have that $x \in \facts(w_2)$.
Furthermore, $x \sqsubseteq a_{1,l}$ and $a_{1,l} \in \facts(w_1)$.
Consequently, $x \in \facts(w_1) \intersect \facts(w_2)$.

Like Case 3.1.1, we shall prove that if Duplicator plays their winning strategy for $\mathcal{G}_1$ and $\mathcal{G}_2$, then $b_l = b_i \cdot b_j$.
To show this, we look at possible choices for Spoiler in rounds $k+1$ up to round $k+r+2$  for $\mathcal{G}_1$ and $\mathcal{G}_2$.

\begin{itemize}
\item For round $k+1$ of $\mathcal{G}_1$, Spoiler can choose $\mathfrak{A}_w |_{\facts(w_1)}$ and $a_{1, \kappa+1} \df x$.
Since $x \in \facts(w_1) \intersect \facts(w_2)$, we know that $|x| \leq r$, and therefore Duplicator must respond with $\mathfrak{B}_v |_{\facts(v_1)}$ and $b_{1,\kappa+1} \df x$.
Since we know that $a_{1,l} = a_{1,i} \cdot a_{1,\kappa+1}$, and Duplicator plays their winning strategy for the $k+r+2$ round game $\mathcal{G}_1$, we know that $b_{1,l} = b_{1,i} \cdot b_{1,\kappa+1}$ must also hold.
\item For round $k+1$ of $\mathcal{G}_2$, Spoiler can choose $\mathfrak{A}_w |_{\facts(w_2)}$ and $a_{2,\kappa+1} \df x$.
Duplicator must respond with $\mathfrak{B}_v |_{\facts(v_1)}$ and $b_{2,\kappa+1} \df x$.
Since $a_{2,j} = a_{2,\kappa+1} \cdot a_{2,l}$, we know that $b_{2,j} = b_{2,\kappa+1} \cdot b_{2,l}$ must hold.
\end{itemize}

Since $b_{1,\kappa+1} = b_{2,\kappa+1} = x$, we have:
\[  b_{1,i} \cdot b_{1,\kappa+1} \cdot b_{2,l} = b_{1,i} \cdot b_{2,\kappa+1} \cdot b_{2,l} . \]

Furthermore, we know $b_{1,l} = b_{1,i} \cdot b_{1,\kappa+1}$ and $b_{2,j} = b_{2,\kappa+1} \cdot b_{2,l}$.
Thus
\[  \underbrace{b_{1,i} \cdot b_{1,\kappa+1}}_{b_{1,l}} \cdot b_{2,l} = b_{1,i} \cdot \underbrace{b_{2,\kappa+1} \cdot b_{2,l}}_{b_{2,j}} . \]
Since $b_l = b_{1,l} \cdot b_{2,l}$ and $b_i \cdot b_j = b_{1,i} \cdot b_{2,j}$, we have that $b_l = b_i \cdot b_j$.

\proofsubparagraph{Case 3.2, $a_i \in A_\mathsf{other}$ and $a_j \in \facts(w_2)$:}
Since $a_i \in A_\mathsf{other}$, we have that $a_i = a_{1,i} \cdot a_{2,i}$ where $a_{1,i} \in \facts(w_1)$ and $a_{2,i} \in \facts(w_2)$.
We also have that $a_j = a_{2,j}$ as $a_j \in \facts(w_2)$.
Therefore, $a_l = a_{1,i} \cdot a_{2,i} \cdot a_{2,j}$.
Furthermore, as $a_l = a_{l,1} \cdot a_{l,2}$, where $a_{l,1} \in \facts(w_1)$ and $a_{l,2} \in \facts(w_2)$, we can write that 
\[ a_{1,l} \cdot a_{2,l} = \underbrace{a_{1,i} \cdot a_{2,i}}_{a_i} \cdot \underbrace{a_{2,j}}_{a_j}.\]
It follows that either $a_{1,l}$ is a prefix of $a_{1,i}$, or $a_{1,i}$ is a prefix of $a_{1,l}$.
This gives us two subcases to consider.

\subparagraph*{Case 3.2.1, $a_{1,l}$ is a prefix of $a_{1,i}$:}
It follows that $a_{1,i} = a_{1,l} \cdot x$ for some $x \in \Sigma^*$. 
We therefore can write that
\[  a_{1,l} \cdot a_{2,l} = a_{1,l} \cdot x \cdot a_{2,i} \cdot  a_{2,j} \]
and hence $a_{2,l} = x \cdot a_{2,i} \cdot a_{2,j}$. 
We have that $x \sqsubseteq a_{1,i}$ and $x \sqsubseteq a_{2,l}$, and thus, $x \in \facts(w_1) \intersect \facts(w_2)$.

We now look at round $k+1$ to round $k+r+2$ of $\mathcal{G}_1$ and $\mathcal{G}_2$.
\begin{itemize}
\item For round $k+1$ of $\mathcal{G}_1$, Spoiler can choose $\mathfrak{A}_w |_{\facts(w_1)}$ and $a_{1,\kappa+1} \df x$.
Since $x \in \facts(w_1) \intersect \facts(w_2)$, we know that $|x| \leq r$, and therefore Duplicator must respond with $\mathfrak{B}_v |_{\facts(v_1)}$ and $b_{1,\kappa+1} \df x$, see~\cref{lemma:consistentStrats}.
Since we know that $a_{1,i} = a_{1,l} \cdot a_{1,\kappa+1}$, and Duplicator plays their winning strategy for the $k+r+2$ round game $\mathcal{G}_1$, we know that $b_{1,i} = b_{1,l} \cdot b_{1,\kappa+1}$ must also hold.
\item For round $k+1$ of $\mathcal{G}_2$, Spoiler can choose $\mathfrak{A}_w |_{\facts(w_2)}$ and $a_{2,\kappa+1} \df x$.
Duplicator must respond with $\mathfrak{B}_v |_{\facts(v_1)}$ and $b_{2,\kappa+1} \df x$.
\item For round $k+2$ of $\mathcal{G}_2$, Spoiler can choose $\mathfrak{A}_w |_{\facts(w_2)}$ and $a_{2,\kappa+2} \df a_{2,\kappa+1} \cdot a_{2,i}$.
Note that $a_{2,\kappa+1} \cdot a_{2,i} \in \facts(w_2)$, since $x \cdot a_{2,i} \sqsubseteq a_{2,l}$, and $a_{2,l} \in\facts(w_2)$.
Duplicator therefore must respond with $\mathfrak{B}_v |_{\facts(v_1)}$ and $b_{2,\kappa+2} \df b_{2,\kappa+1} \cdot b_{2,i}$.
Since $a_{2,l} = a_{2,\kappa+2} \cdot a_{2,j}$, and Duplicator plays $\mathcal{G}_2$ with their winning strategy, we know that $b_{2,l} = b_{2,\kappa+2} \cdot b_{2,j}$ must hold.
\end{itemize}

Since $b_{2,\kappa+1} = b_{1,\kappa+1} = x$, we can write 
\[ b_{1,l} \cdot  b_{2,\kappa+1} \cdot b_{2,i} \cdot b_{2,j} = b_{1,l} \cdot b_{1,\kappa+1} \cdot b_{2,i} \cdot b_{2,j}.  \]
We also know that $b_{2,\kappa+2} = b_{2, \kappa+1} \cdot b_{2,i}$ and $b_{1,i} = b_{1,l} \cdot b_{1,\kappa+1}$, therefore:
\[ b_{1,l} \cdot b_{2,\kappa+2} \cdot b_{2,j} = b_{1,i} \cdot b_{2,i} \cdot b_{2,j}.  \]
Now, since $b_{2,l} = b_{2,\kappa+2} \cdot b_{2,j}$ we have
\[ b_{1,l} \cdot b_{2,l} = b_{1,i} \cdot b_{2,i} \cdot b_{2,j}.  \]
Since $b_l = b_{1,l} \cdot b_{2,l}$ and $b_i \cdot b_j = b_{1,i} \cdot b_{2,i} \cdot b_{2,j}$, we have that $b_l = b_i \cdot b_j$.

\subparagraph*{Case 3.2.2, $a_{1,i}$ is a prefix of $a_{1,l}$:}
It follows that $a_{1,l} = a_{1,i} \cdot x$ for some $x \in \Sigma^*$, where we know that $x \in \facts(w_1)$.
We therefore can write that 
\[  \underbrace{a_{1,i} \cdot x}_{a_{1,l}} \cdot a_{2,l} = a_{1,i} \cdot a_{2,i} \cdot  a_{2,j} \]
and hence $x \cdot a_{2,l} = a_{2,i} \cdot a_{2,j}$. 
With apologies to the reader, we split Case 3.2.2 into two further subcases:
\begin{description}
\item[Case 3.2.2.1,] $x$ is a prefix of $a_{2,i}$, and
\item[Case 3.2.2.2,] $a_{2,i}$ is a strict prefix of $x$.
\end{description}

Let us now consider the first of these two subcases.

\medskip
\emph{Case 3.2.2.1, $x$ is a prefix of $a_{2,i}$:}
Since $a_{2,i} \in \facts(w_2)$, we know that $x \in \facts(w_2)$, and therefore $x \in \facts(w_1) \intersect \facts(w_2)$.
Furthermore, as $x$ is a prefix of $a_{2,i}$, there exists some $y \in \Sigma^*$ such that $a_{2,i} = x \cdot y$.
Therefore, 
\begin{align*}
a_{1,l} \cdot a_{2,l} = & a_{1,i} \cdot a_{2,i} \cdot a_{2,j}, \\
\underbrace{a_{1,i} \cdot x}_{a_{1,l}} \cdot a_{2,l} = & a_{1,i} \cdot \underbrace{x \cdot y}_{a_{2,i}} \cdot a_{2,j}.
\end{align*}
Hence, $a_{2,l} = y \cdot a_{2,j}$, which in turn implies that:
\[ \underbrace{a_{1,i} \cdot x}_{a_{1,l}} \cdot \underbrace{y \cdot a_{2,j}}_{a_{2,l}} = a_{1,i} \cdot \underbrace{x \cdot y}_{a_{2,i}} \cdot a_{2,j}.  \]

We now look at round $k+1$ to round $k+r+2$ of $\mathcal{G}_1$ and $\mathcal{G}_2$.
\begin{itemize}
\item For round $k+1$ of $\mathcal{G}_1$, Spoiler can choose $\mathfrak{A}_w |_{\facts(w_1)}$ and $a_{1,\kappa+1} \df x$.
Since $x \in \facts(w_1) \intersect \facts(w_2)$, it follows that Duplicator must respond with $b_{1,\kappa+1} \df x$. 
Furthermore, as $a_{1,l} = a_{1,i} \cdot a_{1,\kappa+1}$, it must follow that $b_{1,l} = b_{1,i} \cdot b_{1,\kappa+1}$.
\item For round $k+1$ of $\mathcal{G}_2$, Spoiler can choose $\mathfrak{A}_w |_{\facts(w_1)}$ and $a_{2,\kappa+1} \df x$.
Duplicator must respond with $b_{2,\kappa+1} \df x$, since $x \in \facts(w_1) \intersect \facts(w_2)$.
\item For round $k+2$ of $\mathcal{G}_2$, Spoiler can choose $\mathfrak{A}_w |_{\facts(w_1)}$ and $a_{2,\kappa+2} \df y$.
Duplicator responds with some $b_{2,\kappa+2} \in \facts(v_2)$. Since $a_{2,i} = a_{2,\kappa+1} \cdot a_{2,\kappa+2}$, it follows that $b_{2,i} = b_{2,\kappa+1} \cdot b_{2,\kappa+2}$ must hold.
Furthermore, $a_{2,l} = a_{2,\kappa+2} \cdot a_{2,j}$ and therefore $b_{2,l} = b_{2,\kappa+2} \cdot b_{2,j}$.
\end{itemize}

Since $b_{1,\kappa+1} = b_{2,\kappa+1}$, we have
\[ b_{1,i} \cdot b_{1,\kappa+1} \cdot b_{2,\kappa+2} \cdot b_{2,j} = b_{1,i} \cdot b_{2,\kappa+1} \cdot b_{2,\kappa+2} \cdot b_{2,j}. \]

However, because we know that:
\begin{itemize}
\item $b_{1,l} = b_{1,i} \cdot b_{1,\kappa+1}$, 
\item $b_{2,i} = b_{2,\kappa+1} \cdot b_{2,\kappa+2}$,
\item $b_{2,l} = b_{2,\kappa+2} \cdot b_{2,j}$,
\end{itemize}
we can conclude that:
\[  b_{1,l} \cdot b_{2,l} = b_{1,i} \cdot b_{2,i} \cdot b_{2,j}. \] 
Thus, we obtain $b_i \cdot b_j = b_l$.

\medskip
\emph{Case 3.2.2.2, $a_{2,i}$ is a strict prefix of $x$:}
It follows that $x = a_{2,i} \cdot y$ for some $y \in \Sigma^+$. Due to the fact that
\[ a_{1,l} \cdot a_{2,l} = a_{1,i} \cdot a_{2,i} \cdot a_{2,j}, \]
and we know that $a_{1,l} = a_{1,i} \cdot x$ and $x = a_{2,i} \cdot y$ holds, we can write that:
\begin{equation}
\underbrace{a_{1,i} \cdot \overbrace{a_{2,i} \cdot y}^x}_{a_{1,l}} \cdot a_{2,l} = a_{1,i} \cdot a_{2,i} \cdot a_{2,j}. \label{eq:1}
\end{equation}
Thus, $y \cdot a_{2,l} = a_{2,j}$.
Substituting $y \cdot a_{2,l}$ for $a_{2,j}$ into \eqref{eq:1}, we have:
\[ \underbrace{a_{1,i} \cdot a_{2,i} \cdot y}_{a_{1,l}} \cdot a_{2,l} = a_{1,i} \cdot a_{2,i} \cdot \underbrace{y \cdot a_{2,l}}_{a_{2,j}}. \]

Note that this then implies that $a_{2,i} \in \facts(w_1) \intersect \facts(w_2)$, since $a_{2,i} \sqsubseteq a_{1,l}$.
It follows from the fact that $a_{2,i} \in \facts(w_1) \intersect \facts(w_2)$, that $b_{2,i} = a_{2,i}$.
Likewise $y \in \facts(w_1) \intersect \facts(w_2)$, since $y \sqsubseteq a_{1,l}$, and $y \sqsubseteq a_{2,j}$.

We now look at round $k+1$ to round $k+r+2$ of $\mathcal{G}_1$ and $\mathcal{G}_2$.
\begin{itemize}
\item For round $k+1$ of $\mathcal{G}_1$, Spoiler can choose $\mathfrak{A}_w |_{\facts(w_1)}$ and $a_{1,\kappa+1} \df y$.
Since $y \in \facts(w_1) \intersect \facts(w_2)$, it follows that Duplicator must respond with $b_{1,\kappa+1} \df y$. 
\item For round $k+2$ of $\mathcal{G}_1$, Spoiler can choose $\mathfrak{A}_w |_{\facts(w_1)}$ and $a_{1,\kappa+2} \df a_{2,i}$.
Since $a_{2,i} \in \facts(w_1) \intersect \facts(w_2)$, Duplicator responds with $b_{1,\kappa+2}$, where $b_{1,\kappa+2} = a_{2,i}$
\item For round $k+3$ of $\mathcal{G}_1$, Spoiler can choose $\mathfrak{A}_w |_{\facts(w_1)}$ and $a_{1,\kappa+3} \df a_{1,\kappa+2} \cdot a_{1,\kappa+1}$.	
We note that $a_{1,\kappa+3} \in \facts(w_1)$, since $a_{1,\kappa+3} = a_{2,i} \cdot y$, and $a_{2,i} \cdot y \sqsubseteq a_{1,l}$.
Duplicator then responds with some $b_{1,\kappa+3}$, where $b_{1,\kappa+3} = b_{1,\kappa+2} \cdot b_{1,\kappa+1}$ must hold.
Furthermore, as $a_{1,l} = a_{1,i} \cdot a_{1,\kappa+3}$, we have that $b_{1,l} = b_{1,i} \cdot b_{1,\kappa+3}$.
\item For round $k+1$ of $\mathcal{G}_2$, Spoiler can choose $\mathfrak{A}_w |_{\facts(w_2)}$ and $a_{2,\kappa+1} \df y$.
Therefore, Duplicator must respond with $b_{2,\kappa+1} \df y$.
We know that $a_{2,j} = a_{2,\kappa+1} \cdot a_{2,l}$, and thus $b_{2,j} = b_{2,\kappa+1} \cdot b_{2,l}$ must hold.
\end{itemize}

Since $b_{1,\kappa+1} = b_{2,\kappa+1} = y$ and $b_{1,\kappa+2} = b_{2,i} = a_{2,i}$, we can write:
\[ b_{1,i} \cdot b_{1,\kappa+2} \cdot b_{1,\kappa+1} \cdot b_{2,l} = b_{1,i} \cdot b_{2,i} \cdot b_{2,\kappa+1} \cdot b_{2,l}. \]

We know that $b_{1,\kappa+3} = b_{1,\kappa+2} \cdot b_{1,\kappa+1}$ and therefore:
\[ b_{1,i} \cdot b_{1,\kappa+3} \cdot b_{2,l} = b_{1,i} \cdot b_{2,i} \cdot b_{2,\kappa+1} \cdot b_{2,l}. \]

Furthermore, we know that
\begin{itemize}
\item $b_{1,l} = b_{1,i} \cdot b_{1,\kappa+3}$, and
\item $b_{2,j} = b_{2,\kappa+1} \cdot b_{2,l}$.
\end{itemize}
Thus,
\[ b_{1,l} \cdot b_{2,l} = b_{1,i} \cdot b_{2,i} \cdot b_{2,j}. \]
Hence, we have that $b_l = b_i \cdot b_j$.

\proofsubparagraph{Case 3.3, $a_i \in \facts(w_1)$ and $a_j \in A_\mathsf{other}$:}
This case is symmetric to Case 3.2.
We therefore omit a proof.

Thus, for all possible cases, our assumption that $a_l = a_i \cdot a_j$ yields $b_l = b_i \cdot b_j$. 
\end{claimproof}

It therefore follows that the strategy we have given for Duplicator always results in a partial isomorphism.
Thus, Duplicator has a winning strategy for the $k$-round game over $w_1 \cdot w_2$ and $v_1 \cdot v_2$. 
Consequently, $w_1 \cdot w_2 \equiv_k v_1 \cdot v_2$.
\end{proof}

\section{Appendix for Section~\ref{sec:primitivepower}}

\cref{sec:primitivepower} is dedicated to the Primitive Power Lemma.
Before giving the proof of this result, we first prove some lemmas regarding primitive words.
Most of these subsequent lemmas follow from the fact that a primitive word $w \in \Sigma^+$ cannot be a non-trivial factor of $ww$.
That is, if $ww = p w s$, then either $p = \emptyword$ or $s = \emptyword$; for more details, see Chapter 6 of~\cite{handbookOfFL}.
For our purposes, we look at a slight generalization of this idea:

\begin{lemma}\label{obs:primitive}
A word $w \in \Sigma^+$ is primitive if and only if for all $m \in \mathbb{N}$, we have that $w^m = u \cdot w \cdot v$ for any $u, v \in \Sigma^+$, implies that $u = w^n$ and $v = w^{n'}$ for some $n, n' < m$. 
\end{lemma}
\begin{proof}
For this proof, we only show that $w \in \Sigma^+$ is primitive if and only if for any $m \in \mathbb{N}$, we have that $w^m = u \cdot w \cdot v$ for $u,v \in \Sigma^+$ implies that $u = w^n$ for some $n < m$. 
The fact that $v = w^{n'}$ for some $n' < n$ also holds, follows immediately from a length argument.
Furthermore, since $u, v \in \Sigma^+$, it follows that $m \geq 2$.

\proofsubparagraph{If direction.}
We show that if $w^m = u \cdot w \cdot v$ implies that $u = w^n$ for some $n < m$, then $w \in \Sigma^+$ is primitive by proving the contraposition. 
That is, $w \in \Sigma^+$ is imprimitive implies that there exists $u,v \in \Sigma^+$ such that $w^m = u \cdot w \cdot v$ and $u \neq w^n$ for some $n < m$.
Let $w \in \Sigma^+$ and let $z \in \Sigma^+$ where $w = z^k$ for some $k > 1$.
Thus, $w^m = z^{km}$, and therefore we can write $w^m = z \cdot w \cdot z^{km-k-1}$.
Since $z \neq w$, we have proven the contraposition.

\proofsubparagraph{Only if direction.}
For this direction, we show that if $w \in \Sigma^*$ is primitive, then for all $m \in \mathbb{N}$, we have that $w^m = u \cdot w \cdot v$ for any $u,v \in \Sigma^+$ implies that $u = w^n$ for some $n < m$. 
Working towards a contradiction, let $m \in \mathbb{N}$ and assume that $w \in \Sigma^*$ is primitive, but $w^m = u \cdot w \cdot v$ where $u \neq w^n$ for any $n \in \mathbb{N}$.
Trivially, $u \neq \emptyword$ must hold since $\emptyword = w^0$.
Furthermore, since $u, v \in \Sigma^+$, it must hold that $m \geq 2$ for $w^m = u \cdot w \cdot v$ to hold.

First, we show that we can assume $u$ is a strict prefix of $w$. Note that $|u| \neq w^k$ for any $k \in \mathbb{N}$.
If $|u| > |w|$, then because $w^m = u \cdot w \cdot v$, it follows that $u = w^s \cdot \bar u$ for some $s \in \mathbb{N}_+$ and a proper prefix $\bar u \in \Sigma^+$ of $w$.
Thus, we can consider the new equality
\[ w^{m-s} = \bar u \cdot w \cdot v. \ \] 
For the rest of the proof, we simply use $u$ to denote a strict prefix of $w$ such that the equality $w^m = u \cdot w \cdot v$ holds.

Since $u$ is a strict prefix of $w$, we can write $w = u \cdot u'$ for some $u' \in \Sigma^+$. 
Furthermore, since $u \cdot w \cdot v = w^m$ and $w^m = (u \cdot u')^m$ , we have that 
\[ u \cdot w \cdot v = (u \cdot u')^m, \]
which implies
\begin{align*}
w \cdot v &= u' \cdot (u \cdot u')^{m-1}, \\
w \cdot v &= u' \cdot u \cdot u' \cdot (u \cdot u')^{m-2}.
\end{align*}
We know that $|w| = |u| + |u'|$, and therefore $w = u \cdot u'$ and $w = u' \cdot u$.
Proposition 1.3.2 in Lothaire~\cite{lothaire1997combinatorics} states that if $vu = uv$ for $u,v \in \Sigma^*$, then there is some $z \in \Sigma^*$ and $k_1,k_2 \in \mathbb{N}$ such that $u = z^{k_1}$ and $v = z^{k_2}$.
It follows that $w = z^k$ for some $k \in \mathbb{N}$, and since $\emptyword \sqsubset u \sqsubset w$ and $\emptyword \sqsubset u' \sqsubset w $, we know that $k > 1$.
Therefore $w$ is imprimitive, which is a contradiction.
\end{proof}

\subsection{Proof of Lemma~\ref{obs:factorOfRep}}
\begin{proof}
The fact that every $u \sqsubseteq w^m$ where $\expo_w(u) > 0$ can be factorized as $u = u_1 \cdot w^{\expo_w(u)} \cdot u_2$ where $u_1$ is a proper suffix of $w$ and $u_2$ is a proper prefix of $w$ is obvious.
Therefore, this proof focuses on showing that $u_1$ and $u_2$ are unique.
Working towards a contradiction assume that
\begin{align*}
u & = u_1 \cdot w^r \cdot u_2, \text{ and} \\
u & = u_3 \cdot w^r \cdot u_4,
\end{align*}
where $r = \expo_w(u)$, and where $u_1 \neq u_3$ and $u_2 \neq u_4$.
Without loss of generality, assume that $|u_1| < |u_3|$.
Since $u_1 \cdot w^r \cdot u_2 = u_3 \cdot w^r \cdot u_4$, it follows that there exists $u_1' \in \Sigma^+$ such that $u_3 =  u_1 \cdot u_1'$.
Thus,
\[ u_1 \cdot w^r \cdot u_2 = u_1 \cdot u_1' \cdot w^r \cdot u_4. \]
Which implies that
\[ w^r \cdot u_2 = u_1' \cdot w^r \cdot u_4. \] 
Since $u_2$ is a proper prefix of $w$, there exists $\bar u_2 \in \Sigma^*$ such that $u_2 \cdot \bar u_2 = w$.
Thus, 
\begin{align*}
 w^r \cdot u_2 \cdot \bar u_2 &= u_1' \cdot w^r \cdot u_4 \cdot \bar u_2, \\
 w^r \cdot w & = u_1' \cdot w^r \cdot u_4 \cdot \bar u_2, \\
 w^{r+1} & = u_1' \cdot w^r \cdot u_4 \cdot \bar u_2.
\end{align*}
Since, we know that $r > 0$, we can invoke~\cref{obs:primitive} to conclude that $u_1' = w^k$ for some $k \in \mathbb{N}$.
Therefore, $u_3  = u_1 \cdot w^k$, which is a contradiction, since either $k=0$ and then $u_3 = u_1$, or $k > 0$, and $\expo_w(u) > r$ which cannot hold since $r$ is defined to be $\expo_w(u)$.

Consequently, our assumption that $u_1 \neq u_3$ cannot hold.
Furthermore, using a simple length argument, we can conclude that $u_2 = u_4$ must also hold.
\end{proof}

\begin{lemma}\label{lemma:primitive2}
Let $w \in \Sigma^+$ be a primitive word and let $u, v \sqsubseteq w^m$ for some $m \in \mathbb{N}$, where $\expo_w(u) > 0$ and $\expo_w(v)> 0$.
If $u \cdot v \sqsubseteq w^n$ for some $n \in \mathbb{N}$, then there exists $u_1, u_2, v_1, v_2 \sqsubset w$ such that $u = u_1 \cdot w^s \cdot u_2$ and $v = v_1 \cdot w^t \cdot v_2$ where $s,t \in \mathbb{N}_+$ and either $u_2 \cdot v_1 = \emptyword$ or $u_2 \cdot v_1 = w$.
\end{lemma}
\begin{proof}
Let $w \in \Sigma^+$ be a primitive word and let $u, v \sqsubseteq w^m$ for some $m \in \mathbb{N}$.
Due to the fact that $u, v \sqsubseteq w^m$ where $\expo_w(u) > 0$ and $\expo_w(v)> 0$ it follows from~\cref{obs:factorOfRep} that we can write $u = u_1 \cdot w^s \cdot u_2$ and $u  = v_1 \cdot w^t \cdot v_2$ where $u_1, u_2, v_1, v_2 \sqsubset w$ and where $s = \expo_w(u)$ and $t = \expo_w(v)$.

To prove this lemma, we prove the contrapositive.
That is, we shall prove that if $u_2 \cdot v_1 \neq \emptyword$ and $u_2 \cdot v_1 \neq w$, then $u \cdot v \not\sqsubseteq w^n$ for any $n \in \mathbb{N}$.

To prove the contrapositive of the stated lemma, we work towards a contradiction.
Assume that $u_2 \cdot v_1 \neq \emptyword$ and $u_2 \cdot v_1 \neq w$ and $u \cdot v \sqsubseteq w^n$ for some $n \in \mathbb{N}$.
From what we have defined, we can write
\[ u \cdot v = u_1 \cdot w^s \cdot u_2 \cdot v_1 \cdot w^t \cdot v_2. \]
Since $u \cdot v \sqsubseteq w^n$, there exists $w_1, w_2 \sqsubset w$ such that $u \cdot v = w_1 \cdot w^r \cdot w_2$ for some $r \in \mathbb{N}$, see~\cref{obs:factorOfRep}.
Hence,
\[  u_1 \cdot w^s \cdot u_2 \cdot v_1 \cdot w^t \cdot v_2 = w_1 \cdot w^r \cdot w_2 . \]
Furthermore, it follows that $w_1$ is a proper suffix of $w$ and $w_2$ is a proper prefix of $w$, see~\cref{obs:factorOfRep}.
Thus, there exists $\bar w_1, \bar w_2 \in \Sigma^+$ such that $\bar w_1 \cdot w_1 = w$ and $w_2 \cdot \bar w_2 = w$.
Therefore, 
\begin{align*}
\bar w_1 \cdot u_1 \cdot w^s \cdot u_2 \cdot v_1 \cdot w^t \cdot v_2 \cdot \bar w_2 &= \bar w_1 \cdot w_1 \cdot w^r \cdot w_2 \cdot \bar w_2, \\
&= w \cdot w^r \cdot w , \\
&= w^{r+2}.
\end{align*}
Invoking~\cref{obs:primitive}, we know that 
\[ \bar w_1 \cdot u_1 \cdot w^s \cdot u_2 \cdot v_1 = w^p, \]
for some $p \in \mathbb{N}$.
Furthermore, since $s = \expo_w(u)$, we know that $p > 0$.
Thus, we can again invoke~\cref{obs:primitive}, which implies that $u_2 \cdot v_1 = w^q$ for some $q \in \mathbb{N}$.
Since $u_2, v_1 \sqsubset w$, we have that $|u_2| < |w|$ and $|v_1| < |w|$ and therefore, $|u_2 \cdot v_1| < 2|w|$.
Thus, either $u_2 \cdot v_1 = w^1$ or $u_2 \cdot v_1 = w^0$ which is a contradiction.
\end{proof}

\subsection{Proof of Lemma~\ref{lemma:primitivePower}}
Before giving the actual proof of~\cref{lemma:primitivePower}, we first give a proof for the following lemma:
\begin{lemma}\label{lemma:expoIncrease}
If $u \cdot v \sqsubseteq w^m$ for some primitive word $w \in \Sigma^*$, some $m \in \mathbb{N}$, and some $u,v \in \Sigma^*$, then either 
\begin{enumerate}
\item $\expo_w(u \cdot v) = \expo_w(u) + \expo_w(v)$, or
\item $\expo_w(u \cdot v) = \expo_w(u) + \expo_w(v) + 1$.
\end{enumerate}
\end{lemma}
\begin{proof}
We split this proof into four cases.

\proofsubparagraph{Case 1, $\expo_w(u) \geq 1$ and $\expo_w(v) \geq 1$:}
This case follows immediately follows from~\cref{lemma:primitive2}.
That is, we know $u = u_1 \cdot w^{\expo_w(u)} \cdot u_2$ and $v = v_1 \cdot w^{\expo_w(v)} \cdot v_2$ where $u_1, v_2 \sqsubset w$ and $u_2 \cdot v_1 \in \{ \emptyword, w \}$.
If $u_2 \cdot v_1 = \emptyword$, then 
\[ u \cdot v = u_1 \cdot w^{\expo_w(u)} \cdot w^{\expo_w(v)} \cdot v_2, \]
and hence $\expo_w(u \cdot v) = \expo_w(u)+ \expo_w(v)$.
If $u_2 \cdot v_1 = w$, then 
\[ u \cdot v = u_1 \cdot w^{\expo_w(u)} \cdot w \cdot w^{\expo_w(v)} \cdot v_2, \]
and hence $\expo_w(u \cdot v) = \expo_w(u)+ \expo_w(v) + 1$.

\proofsubparagraph{Case 2, $\expo_w(u) = 0$ and $\expo_w(v) = 0$:}
Working towards a contradiction, assume that $\expo_w(u\cdot v) = r$ where $r \geq 2$.
From~\cref{obs:factorOfRep}, we know that $u \cdot v = w_1 \cdot w^r \cdot w_2$ for some $r \geq 2$ where $w_1,w_2 \sqsubset w$.
Now, we have two cases: First case is if $|u| < |w_1| + |w|$. However, it then follows that $|v| > |w_2| + |w|$ and therefore $w \sqsubseteq v$.
This is a contradiction, since then $\expo_w(v) \geq 1$.
The second case is if $|u| \geq |w_1| + |w|$. However, then $w \sqsubseteq u$.
This is a contradiction, since then $\expo_w(u) \geq 1$.

\proofsubparagraph{Case 3, $\expo_w(u) = 0$ and $\expo_w(v) \geq 1$:}
Working towards a contradiction, assume that 
\begin{enumerate}
\item $\expo_w(u \cdot v) < \expo_w(u) + \expo_w(v)$ or
\item $\expo_w(u \cdot v) > \expo_w(u) + \expo_w(v) + 1$.
\end{enumerate}

Since $\expo_w(u) = 0$, it follows that $\expo_w(u) + \expo_w(v) = \expo_w(v)$.
Therefore, since $w^r \sqsubseteq u \cdot v$ where $r = \expo_w(v)$, it follows that $\expo_w(u \cdot v) < \expo_w(u) + \expo_w(v)$ cannot hold.
Thus, $\expo_w(u \cdot v) > \expo_w(u) + \expo_w(v) + 1$ must hold.

Due to the fact that $\expo_w(v) \geq 1$, we can write $v = v_1 \cdot w^r \cdot v_2$, where $v_1$ is a proper suffix of $w$ and $v_2$ is a proper prefix of $w$, and where $r = \expo_w(v)$.
Furthermore, since $u \cdot v \sqsubseteq w^m$ and $\expo_w(u \cdot v) \geq 1$, we know from~\cref{obs:factorOfRep} that $u \cdot v = w_1 \cdot w^{\expo_w(u \cdot v)} \cdot w_2$ where $w_1$ is a proper suffix of $w$ and $w_2$ is a proper prefix of $w$.
Therefore, 
\begin{align*}
u \cdot v  & = w_1 \cdot w^{\expo_w(u \cdot v)} \cdot w_2, \\
		    & = u \cdot v_1 \cdot w^r \cdot v_2,
\end{align*}
and since $w_2$ and $v_2$ are proper prefixes, it follows that $w_2 = v_2$, see~\cref{obs:factorOfRep}.
We can thus write $u \cdot v_1 \cdot w^r = w_1 \cdot w^{\expo_w(u \cdot v)}$ which implies $\expo_w(u \cdot v) = \expo_w(u \cdot v_1) + r$.
Since $u \sqsubset w$ and $v_1 \sqsubset w$, it follows that $|u \cdot v_1| < 2 \cdot |w|$.
Therefore, $\expo_w(u \cdot v_1) \leq 1$.
Thus, $\expo_w(u \cdot v) > \expo_w(u) + \expo_w(v) + 1$ cannot hold.

\proofsubparagraph{Case 4, $\expo_w(u) \geq 1$ and $\expo_w(v) = 0$:}
This follows analogously to Case 3.
Working towards a contradiction, assume that
\begin{enumerate}
\item $\expo_w(u \cdot v) < \expo_w(u) + \expo_w(v)$ or
\item $\expo_w(u \cdot v) > \expo_w(u) + \expo_w(v) + 1$.
\end{enumerate}
Since $\expo_w(v) = 0$, it follows that $\expo_w(u) + \expo_w(v) = \expo_w(u)$. 
Therefore, $\expo_w(u \cdot v) < \expo_w(u) + \expo_w(v)$ cannot hold.
Thus, $\expo_w(u \cdot v) > \expo_w(u) + \expo_w(v) + 1$ must hold.

Since $\expo_w(u) \geq 1$, we can write $u = u_1 \cdot w^r \cdot u_2$, where $u_1$ is a proper suffix of $w$, and $u_2$ is a proper prefix of $w$, and where $r = \expo_w(u)$.
Therefore, $u \cdot v = u_1 \cdot w^r \cdot u_2 \cdot v$.
We can thus write $\expo_w(u \cdot v) = r + \expo_w(u_2 \cdot v)$.
Since $|u_2| < |w|$ and $|v|< |w|$, we know that $|u_2 \cdot v| < 2 \cdot |w|$.
Therefore, $\expo_w(u_2 \cdot v) \leq 1$.
Consequently, $\expo_w(u \cdot v) > \expo_w(u) + \expo_w(v) + 1$ cannot hold.
\end{proof}

\subsubsection{Actual Proof of Lemma~\ref{lemma:primitivePower}}
\begin{proof}
Let $\mathtt{a} \in \Sigma$ and let $w \in \Sigma^+$ be a primitive word.
Let $p,q,k \in \mathbb{N}_+$ such that $\mathtt{a}^p \equiv_{k+3} \mathtt{a}^q$.
Now, let $\mathfrak{A}_q$ and $\mathfrak{B}_p$ be two $\signature_\Sigma$ structures that represent $w^q \in \Sigma^+$ and $w^p \in \Sigma^+$ respectively, where $w \in \Sigma^+$ is a primitive word. 
Let $A \df \facts(w^q) \union \{ \perp \}$ and $B \df \facts(w^p) \union \{ \perp \}$ be the universes for $\mathfrak{A}_q$ and $\mathfrak{B}_p$ respectively.
Let $\mathcal{G}$ be a $k$-round game over $\mathfrak{A}_q$ and $\mathfrak{B}_p$. 
We assume that Spoiler never chooses $\perp$ in $\mathcal{G}$ since Duplicator simply responds with $\perp$.
 
\proofsubparagraph{The look-up game.}
Let $\mathfrak{C}_q$ and $\mathfrak{D}_p$ be $\signature_{\{ \mathtt{a} \}}$ structures that represent $\mathtt{a}^q$ and $\mathtt{a}^p$ respectively.
Let $\mathcal{G}_l$ be a $k+3$-round look-up game over $\mathfrak{C}_q$ and $\mathfrak{D}_p$.
For each choice Spoiler makes in $\mathcal{G}$, we add a choice for Spoiler in $\mathcal{G}_l$ as follows:
\begin{itemize}
\item If Spoiler chooses $\mathfrak{A}_q$ and $u \in A$, then choose $\mathfrak{C}_q$ and $\mathtt{a}^{\expo_w (u)}$, for Spoiler in $\mathcal{G}_l$, and
\item if Spoiler chooses $\mathfrak{B}_p$ and $v \in B$, then choose $\mathfrak{D}_p$ and $\mathtt{a}^{\expo_w (v)}$, for Spoiler in $\mathcal{G}_l$.
\end{itemize}
After adding the choice for Spoiler in the look-up game $\mathcal{G}_l$, Duplicator then responds with their winning strategy for the $k+3$-round game.
For rounds $k+1$ to $k+3$, we assume that Spoiler chooses an arbitrary structure and factor, and we know that Duplicator can respond using their winning strategy since $\mathfrak{C}_q \equiv_{k+3} \mathfrak{D}_p$.

Before moving on to a winning strategy for Duplicator in $\mathcal{G}$, we prove two small claims about the tuples that result from the look-up game $\mathcal{G}_l$.
\begin{claim}\label{claim:onemore}
Let $\vec a' = (a_1' , a_2', \dots, a_{k+3}', \mathtt{a}, \emptyword)$ where $a_i' \sqsubseteq \mathtt{a}^q$ for all $i \in [k+3]$ and $\vec b' = (b_1', b_2', \dots, b_{k+3}', \mathtt{a}, \emptyword)$ where $b_i' \sqsubseteq \mathtt{a}^p$ for all $i \in [k+3]$ be the resulting tuples of the look-up game $\mathcal{G}_l$.
Then, for any $l,i,j \in [k]$, we have that $a_l' = a_i' \cdot a_j' \cdot \mathtt{a}$ if and only if $b_l' = b_i' \cdot b_j' \cdot \mathtt{a}$.
\end{claim}
\begin{claimproof}
To prove this claim we only prove the only if direction since the if direction follows immediately from symmetric reasoning.
To that end, working towards a contradiction let $i,j,l \leq k$, and assume that $a_l' = a_i' \cdot a_j' \cdot \mathtt{a}$ and $b_l' \neq b_i' \cdot b_j' \cdot \mathtt{a}$.

We know that the result $(\vec a', \vec b')$ of the $k+3$ round look-up game $\mathcal{G}_l$ forms a partial isomorphism.
We know show that if $a_l' = a_i' \cdot a_j' \cdot \mathtt{a}$ and $b_l' \neq b_i' \cdot b_j'  \cdot \mathtt{a}$, then Spoiler can win the $k+3$-round game $\mathcal{G}_l$ by choosing $a_{k+1}' \df a_i' \cdot a_j'$ in round $k+1$.
Since Duplicator plays with their winning strategy, we know that $b_{k+1}' = b_i' \cdot b_j'$.

It follows that that $a_l' = a_{k+1}' \cdot \mathtt{a}$, yet we have assumed $b_l' \neq b_{k+1}' \cdot \mathtt{a}$.
Consequently, Duplicator has lost which is a contradiction, since Duplicator plays with a known winning strategy.
\end{claimproof}

\begin{claim}\label{claim:almostFull}
Let $\vec a' = (a_1' , a_2', \dots, a_{k+3}', \mathtt{a}, \emptyword)$ where $a_i' \sqsubseteq \mathtt{a}^q$ for all $i \in [k+3]$ and $\vec b' = (b_1', b_2', \dots, b_{k+3}', \mathtt{a}, \emptyword)$ where $b_i' \sqsubseteq \mathtt{a}^p$ for all $i \in [k+3]$ be the resulting tuples of the look-up game $\mathcal{G}_l$.
Then, for any $r \leq k$, we have that $\mathtt{a}^q = a_r \cdot x$ if and only if $\mathtt{a}^p = b_r \cdot x$, where $x \in \{ \emptyword, \mathtt{a}, \mathtt{aa} \}$.
\end{claim}
\begin{claimproof}
Let $\vec a' = (a_1' , a_2', \dots, a_{k+3}', \mathtt{a}, \emptyword)$ and $\vec b' = (b_1', b_2', \dots, b_{k+3}', \mathtt{a}, \emptyword)$ be the resulting tuples of the look-up game $\mathcal{G}_l$.
Without loss of generality, assume that $\mathtt{a}^q = a_r \cdot x$ and $\mathtt{a}^p \neq b_r \cdot x$ for some $r \leq k$ and $x \in \{ \emptyword, \mathtt{a}, \mathtt{aa} \}$.

In round $k+1$, assume that Spoiler chooses $\mathfrak{C}_q$ and $a_{k+1} \df x$.
Duplicator must respond with $\mathfrak{D}_p$ and $b_{k+1} \df x$ since $x$ is either a constant or the concatenation of two constants. 
In round $k+1$, assume that Spoiler chooses $\mathfrak{C}_q$ and $a_{k+2} \df \mathtt{a}^q$. 
It follows that $a_{k+2} = a_r \cdot a_{k+1}$ and therefore Duplicator must choose some $b_{k+2} = b_r \cdot b_{k+1}$.
We therefore know that $b_{k+2} \neq \mathtt{a}^p$.
In round $k+3$, let Spoiler choose $\mathfrak{B}_p$ and $b_{k+3} \df b_{k+2} \cdot \mathtt{a}$.
Since Spoiler cannot choose any factor of $\mathtt{a}^q$ of the form $a_{k+2} \cdot \mathtt{a}$, Spoiler has won.
This is a contradiction, as we have assumed that Duplicator plays with their winning strategy.

So far, we have that for any $r \leq k-3$ if $\mathtt{a}^q = a_r \cdot x$ then $\mathtt{a}^p = b_r \cdot x$, where $x \in \{ \emptyword, \mathtt{a}, \mathtt{aa} \}$.
However, the other direction immediately follows by simply considering the same proof with the roles switched.
\end{claimproof}

\proofsubparagraph{Duplicator's strategy.}
We now work towards a full strategy for Duplicator in the $k$-round game over $\mathfrak{A}_q$ and $\mathfrak{B}_p$.
For each round $r \in [k]$ of $\mathcal{G}_l$, if Duplicator responds in $\mathcal{G}_l$ with $\mathtt{a}^n$ for some $n \in \mathbb{N}_+$, then we know that Spoiler's choice in $\mathcal{G}$ is of the form $u_1 \cdot w^m \cdot u_2$, where $u_1, u_2 \sqsubset w$ , and $m \in \mathbb{N}_+$.
Duplicator then responds in $\mathcal{G}$ with $u_1 \cdot w^n \cdot u_2$.
From~\cref{obs:factorOfRep}, we know that this response is uniquely defined.
See~\cref{fig:powerStrat} for an informal illustration of Duplicator's strategy.

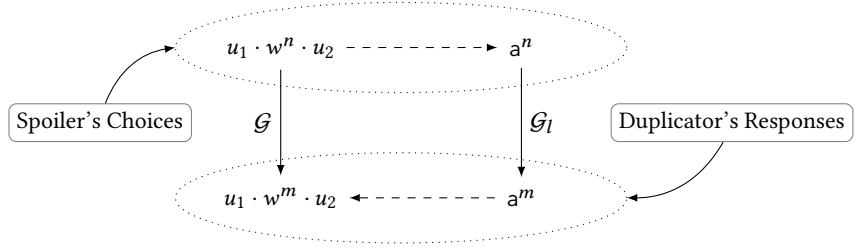
\begin{figure}
    \centering
    \begin{tikzpicture}
    \tikzstyle{vertex}=[rectangle,fill=white!35,minimum size=12pt,inner sep=3pt,outer sep=1.5pt]
   \tikzstyle{vertex2}=[rectangle,fill=white!35,draw=gray,rounded corners=.1cm]
    \node[vertex2] (7) at (-4,0) {Spoiler's Choices};
    \node[vertex2] (8) at (4.4,0) {Duplicator's Responses};

    \draw[-latex] (7) to[bend left=30] (-3,1);
    \draw[-latex] (8) to[bend left=30] (3,-1);
    
    \draw[dotted] (0,1) ellipse (3cm and 0.6cm);
    \draw[dotted] (0,-1) ellipse (3cm and 0.6cm);
    
    \node[vertex] (1) at (-1.6, 1) {$u_1 \cdot w^n \cdot u_2$};
    \node[vertex] (2) at (1.6, 1) {$\mathtt{a}^n$};
    
    \node[vertex] (3) at (1.6,-1) {$\mathtt{a}^m$};
    \node[vertex] (4) at (-1.6,-1) {$u_1 \cdot w^m \cdot u_2$};
    
    \path[-latex,dashed] (1) edge node [above] {} (2);
    \path[-latex] (2) edge node [right] {$\mathcal{G}_l$} (3);
    \path[-latex,dashed] (3) edge node [below] {} (4);
    \path[-latex] (1) edge node [left] {$\mathcal{G}$} (4);
    \end{tikzpicture}
    \caption{Duplicator's strategy when Spoiler chooses some $u_1 \cdot w^n \cdot u_2$ where $n \geq 1$.}
    \label{fig:powerStrat}
\end{figure}

In the case where Duplicator responds with $\emptyword$ in the look-up game $\mathcal{G}_l$, it follows that Spoiler's choice was some $u \in \Sigma^*$, where $\expo_w (u) = 0$.
Thus, Duplicator simply responds with the same $u \in \Sigma^*$.

Before showing that this strategy is a winning strategy for Duplicator, we first show that each choice for Duplicator (as defined above) is indeed a factor of the structure that Duplicator plays in a given round.
The case where Spoiler chooses some $u$ where $\expo_w(u) = 0$ is obvious, therefore we only look at the case where $\expo_w(u) \geq 1$.
Without loss of generality, assume that some round $r \leq k$, Spoiler chooses $\mathfrak{A}_q$ and $u_1 \cdot w^n \cdot u_2$ where $n \in \mathbb{N}_+$. 
It therefore follows that Duplicator's response is of the form $u \df u_1 \cdot w^m \cdot u_2$ where 
\begin{itemize}
\item $m \in \mathbb{N}_+$, 
\item $u_1$ is a strict suffix of $w$, and 
\item $u_2$ is a strict prefix of $w$.
\end{itemize}

Working towards a contradiction, assume that $u$ is not a factor of $w^p$.
By the definition of Duplicator's strategy (using the look-up game), we know that $m > p$ cannot hold.
Likewise, we know that $m \leq p-2$ cannot hold since then $u_1 \cdot w^m \cdot u_2$ is a factor of $w^p$.
Therefore, if $u$ is not a factor of $w^p$, then $p-2 < m \leq p$.
It follows from~\cref{claim:almostFull} that $q-n = p-m$.
However, then $u_1 \cdot w^n \cdot u_2$ is not a factor of $w^q$ if $u_1 \cdot w^m \cdot u_2$ is not a factor of $w^p$, and therefore our assumption that $u$ is not a factor of $w^p$ cannot hold.

\proofsubparagraph{Correctness.}
We now show that the given strategy for Duplicator is indeed a winning strategy.

Let $\vec a' = (a_1' , a_2', \dots, a_{k+5}')$ where $a_i' \sqsubseteq \mathtt{a}^q$ for all $i \in [k]$ and $\vec b' = (b_1', b_2', \dots, b_{k+5}')$ where $b_i' \sqsubseteq \mathtt{a}^p$ for all $i \in [k]$ be the resulting tuples of $\mathcal{G}_l$.
Since Duplicator plays $\mathcal{G}_l$ using their winning strategy, we have that $(\vec a', \vec b')$ forms a partial isomorphism.

Let $\vec a = (a_1, a_2, \dots, a_{k+|\Sigma|+1})$ where $a_i \in A$ for all $i \in [k]$, and let $\vec b = (b_1,b_2,\dots,b_{k+|\Sigma|+1})$ where $b_i \in B$ for all $i \in [k]$ be the resulting tuples from a $k$-round game over $\mathfrak{A}_q$ and $\mathfrak{B}_p$, where Duplicator plays the strategy given in this proof.

To show that $(\vec a, \vec b)$ is indeed a partial isomorphism, we shall show that the following conditions hold:
\begin{itemize}
\item For every $i \in [k+ |\Sigma|+1]$, and constant symbol $c \in \signature_\Sigma$, we have $a_i = c^{\mathfrak{A}_w}$ if and only if $b_i = c^{\mathfrak{B}_v}$,
\item for every $i,j \in [k+ |\Sigma|+1]$, we have that $a_i = a_j$ if and only if $b_i = b_j$, and
\item for every $i,j,k \in [k+ |\Sigma|+1]$, we have that $a_i = a_j \cdot a_k$ if and only if $b_i = b_j \cdot b_k$.
\end{itemize}

Due to the definition of Duplicator's strategy, we know that if Spoiler some constant $\mathtt{a} \in \Sigma$, then Duplicator responds with $\mathtt{a} \in \Sigma$. 
If Spoiler chooses some $a_j$ such that $a_j = a_i$, then by the definition of Duplicator's strategy, Duplicator chooses $b_j$ such that $b_j = b_i$. 
The case where Spoiler chooses $b_j$ follows analogously.

Therefore, it suffices to show that $a_i \cdot a_j = a_l$ if and only if $b_i \cdot b_j = b_l$ holds to prove that $(\vec a, \vec b)$ is a partial isomorphism.

Without loss of generality, assume that $a_i \cdot a_j = a_l$.
We now look at the following cases, and for each case show that $a_i \cdot a_j = a_l$ implies $b_i \cdot b_j = b_l$
\begin{itemize}
\item \textbf{Case 1:} $\expo_w(a_i) = 0$ and $\expo_w(a_j) = 0$, 
\item \textbf{Case 2:} $\expo_w(a_i) \geq 1$ and $\expo_w(a_j) \geq 1$,
\item \textbf{Case 3:} $\expo_w(a_i) = 0$ and $\expo_w(a_j) \geq 1$, and
\item \textbf{Case 4:} $\expo_w(a_i) \geq 1$ and $\expo_w(a_j) = 0$.
\end{itemize}

\textbf{Case 1, $\expo_w(a_i) = 0$ and $\expo_w(a_j) = 0$:} 
First, note that from~\cref{lemma:expoIncrease}, we know that $\expo_w(a_l) \leq 1$ must hold.
Thus, we continue Case 1 under the assumption that $\expo_w(a_l) \leq 1$.
Furthermore, we know from the definition of Duplicator's strategy, that $a_i = b_i$ and $a_j = b_j$.

\begin{itemize}
\item If $\expo_w(a_l) = 0$, then it follows directly from Duplicator's strategy that $a_l = b_l$, and $a_i = b_i$ and $a_j = b_j$.
Thus, $b_l = b_i \cdot b_j$.
\item If $\expo_w(a_l) = 1$, then we can write $a_l = w_1 \cdot w \cdot w_2$ where $w_1, w_2 \sqsubset w$.
From the look-up game, we know that $a_l' = \mathtt{a}$, and therefore $b_l' = \mathtt{a}$ must hold.
Thus, from Duplicator's strategy, we know $b_l = w_1 \cdot w \cdot w_2$. 
Hence $b_l = b_i \cdot b_j$.
\end{itemize}

\textbf{Case 2, $\expo_w(a_i) \geq 1$ and $\expo_w(a_j) \geq 1$:}
For this case, we know that we can write:
\begin{itemize}
\item $a_i = u_i  \cdot w^{m_i} \cdot v_i$, where $u_i$ is a suffix of $w$ and $v_i$ is a prefix of $w$, and $m_i = \expo_w(a_i)$, and
\item $a_j = u_j  \cdot w^{m_j} \cdot v_j$, where $u_j$ is a suffix of $w$ and $v_j$ is a prefix of $w$, and $m_j = \expo_w(a_j)$.
\end{itemize}

Due to the definition of Duplicator's strategy, we know that 
\begin{itemize}
\item $b_i = u_i  \cdot w^{n_i} \cdot v_i$, where $n_i = \expo_w(b_i)$, and
\item $b_j = u_j  \cdot w^{n_j} \cdot v_j$, where $n_j = \expo_w(b_j)$.
\end{itemize}

Since $a_l \sqsubseteq w^q$ and $b_l \sqsubseteq w^p$, we know from~\cref{lemma:primitive2} that $v_i \cdot u_j = w$ or $v_i \cdot u_j = \emptyword$.
Thus,
\[ a_l = u_i \cdot w^{m_i} \cdot w^g \cdot w^{m_j} \cdot v_j, \]
where $g \in \{0, 1 \}$.
From the definition of Duplicator's strategy,
$ b_l = u_i \cdot w^n \cdot v_j $
for some $n \in \mathbb{N}$. 
Furthermore,
\begin{align*}
b_i \cdot b_j &= u_i \cdot w^{n_i} \cdot w^g \cdot w^{n_j} \cdot v_j, \\
b_i \cdot b_j &= u_i \cdot w^{n_i + n_j + g} \cdot v_j.
\end{align*}
Therefore, $b_l = b_i \cdot b_j$ if and only if $n = n_i + n_j + g$.

Recall that the result $(\vec a', \vec b')$ of the look-up game $\mathcal{G}_l$ forms a partial isomorphism.
From the look-up game $\mathcal{G}_l$, we know that 
\begin{enumerate}
\item $a_i' = \mathtt{a}^{m_i}$ and $b_i' = \mathtt{a}^{n_i}$,
\item $a_j' = \mathtt{a}^{m_j}$ and $b_j' = \mathtt{a}^{n_j}$, and
\item $a_l' = \mathtt{a}^{m_i + m_j + g}$ and $b_l' = \mathtt{a}^{n}$.
\end{enumerate}

We consider two cases.

\emph{Case 2.1, $g = 0$:}
Since $g = 0$, it follows that $a_l' = \mathtt{a}^{m_i + m_j}$, and hence $a_l' = a_i' \cdot a_j'$.
Furthermore, because $(\vec a', \vec b')$ forms a partial isomorphism and $a_l' = a_i' \cdot a_j'$, we have that $b_l' = b_i' \cdot b_j'$. 
Hence, $b_l' = \mathtt{a}^{n_i} \cdot \mathtt{a}^{n_j}$, and consequently, $n = n_i+ n_j$.
It therefore follows that $b_l = b_i \cdot b_j$ holds.

\emph{Case 2.2, $g = 1$:}
Since $g=1$, it follows that $a_l' = \mathtt{a} \cdot \mathtt{a}^{m_i + m_j}$, and thus $a_l' = \mathtt{a} \cdot a_i' \cdot a_j'$.
It follows from~\cref{claim:onemore} that $b_l' = \mathtt{a} \cdot b_i' \cdot b_j'$, which means that $n = n_i + n_j + 1$.
It therefore follows that $b_l = b_i \cdot b_j$ holds.

\textbf{Case 3, $\expo_w(a_i) = 0$ and $\expo_w(a_j) \geq 1$:}
Since $\expo_w(a_j) \geq 1$ and therefore $\expo(a_l) \geq 1$, we can write
\begin{itemize}
\item  $a_j = u_j \cdot w^{\expo_w(a_j)} \cdot v_j$ where $u_j, v_j \sqsubset w$, and
\item $a_l = u_l \cdot w^{\expo_w(a_l)} \cdot v_l$ where $u_l, v_l \sqsubset w$.
\end{itemize}

From the definition of Duplicator's strategy, we know that 
\begin{itemize}
\item $b_i = a_i$,
\item $b_j = u_j \cdot w^{\expo_w(b_j)} \cdot v_j$, and
\item $b_l = u_l \cdot w^{\expo_w(b_l)} \cdot v_l$.
\end{itemize}

Since $a_l = a_i \cdot a_j$, we have that
\[ u_l \cdot w^{\expo_w(a_l)} \cdot v_l = a_i \cdot u_j \cdot w^{\expo_w(a_j)} \cdot v_j. \]

Notice that $\expo_w(a_i) = 0$ and $\expo_w(u_j) = 0$, and thus from~\cref{lemma:expoIncrease}, we know that $\expo_w(a_i \cdot u_j) \leq 1$.
If $\expo_w(a_i \cdot u_j) = 0$, then $a_l = a_i \cdot u_j \cdot w^{\expo_w(a_l)} \cdot v_l$, and therefore 
\[ u_l \cdot w^{\expo_w(a_l)} \cdot v_l = a_i \cdot u_j \cdot w^{\expo_w(a_l)} \cdot v_l, \]
which implies that $a_i \cdot u_j = u_l$.

If $\expo_w(a_i \cdot u_j) = 1$, then from~\cref{obs:factorOfRep}, we know that $a_i \cdot u_j = w' \cdot w \cdot w''$, where $w', w'' \sqsubset w$.
Since,
\[ (w' \cdot w \cdot w'') \cdot  w^{\expo_w(a_j)} \cdot v_l \sqsubseteq w^q , \]
it immediately follows that $w'' = \emptyword$ must hold, see~\cref{obs:primitive}.

Thus, either $a_i \cdot u_j = u_l$ or $a_i \cdot u_j = u_l \cdot w$.
Therefore, we can write:
\[ a_l = u_l \cdot w^g \cdot w^{\expo_w(a_j)} \cdot v_j, \]
where $g \in \{ 0, 1\}$.
 
We know that:
\[ b_i \cdot b_j = b_i \cdot u_j \cdot w^{\expo_w(b_j)} \cdot v_j. \]
Furthermore, $b_i = a_i$, and  $a_i \cdot u_j = u_l \cdot w^g$ where $g \in \{0,1\}$.
Therefore,
\begin{align*}
 b_i \cdot b_j = &  u_l \cdot w^g \cdot w^{\expo_w(b_j)} \cdot v_j, \\
 				= &  u_l \cdot w^{\expo_w(b_j) + g} \cdot v_j,
\end{align*}
for some $g \in \{ 0, 1\}$.
Hence $b_l = b_i \cdot b_j$ if and only if $\expo_w(b_l) = \expo_w(b_i) + g$ for $g \in \{ 0, 1\}$.
From~\cref{claim:onemore} along with the analogous reasoning given in Case 2, we know that $\expo_w(b_l) = \expo_w(b_i) + g$ must hold.
Consequently, $b_l = b_i \cdot b_j$.

\textbf{Case 4, $\expo_w(a_i) \geq 1$ and $\expo_w(a_j) = 0$:}
This case is symmetric to Case 3.
We therefore omit a proof.

\proofsubparagraph{Concluding the proof.}
We have shown that if $a_l = a_i \cdot a_j$, then $b_l = b_i \cdot b_j$ must hold.
Due to the arbitrary nature of the notation $\vec a$ and $\vec b$, it follows that if $b_l = b_i \cdot b_j$, then $a_l = a_i \cdot b_j$.
Thus, we have shown that $(\vec a, \vec b)$ is a partial isomorphism.
We have therefore given a winning strategy for Duplicator over $\mathfrak{A}_q$ and $\mathfrak{B}_p$ -- two structures that represent $w^q$ and $w^p$ respectively, whenever $\mathtt{a}^q \equiv_{k+3} \mathtt{a}^p$.
\end{proof}

\subsection{Proof of Proposition~\ref{lemma:allWords}}
\begin{proof}
The proof follows almost directly from~\cref{lemma:primitivePower}.
Let $w \in \Sigma^+$ and let $z \in \Sigma^+$ be the primitive root of $w$.
That is, $w = z^i$ for some $i \in \mathbb{N}_+$ and some primitive word $z \in \Sigma^+$.
It follows that for any $p \in \mathbb{N}_+$, we have that $w^p = z^{pi}$.
Therefore, if for any primitive word $z \in \Sigma^+$, and any $k \in \mathbb{N}$, there exists $p \in \mathbb{N}_+$ and $v \in \Sigma^*$ such that $v \neq z^{pi}$, where $z^{pi} \equiv_k v$, then we have proven the stated lemma.

To that end, we observe that the set $\{ i \cdot 2^n \mid n \in \mathbb{N}_+ \}$ is not semi-linear (recall Section~\ref{sec:EFgames}, using the same reasoning as for $L_{\mathsf{pow}}$) and, hence, not definable in $\fc$.
Consequently, from~\cref{lemma:primitivePower}, for any word $w \in \Sigma^+$ and any $k \in \mathbb{N}_+$, there exists $p \in \mathbb{N}_+$ and $v \in \Sigma^*$ such that $w^p \equiv_k v$ where $w^p \neq v$.
\end{proof}

\section{Appendix for Section~\ref{subsec:pumping}}

\subsection{Proof of Lemma~\ref{lemma:coprim}}
\begin{proof}
We prove this lemma in two parts.

\proofsubparagraph{(1) if and only if (2).}
First, we show that the only if direction holds by proving the contrapositive.
In other words, if $v$ and $w$ conjugate, then there does not exist $n_0, m_0 \in \mathbb{N}$ such that 
$ \facts(w^{n_0}) \intersect \facts(v^{m_0}) = \facts(w^{n}) \intersect \facts(v^{m}) $
for all $n > n_0$ and all $m > m_0$.
To that end, assume that there exists $x,y \in \Sigma^*$ such that $w = x \cdot y$ and $v = y \cdot x$.
Then, for any $m \in \mathbb{N}$, we have that
\begin{align*}
w^m = & \underbrace{xy \cdot xy \cdot xy \cdots xy}_{m \text{ times}}, \text{ and } \\
v^m = & \underbrace{yx \cdot yx \cdot yx \cdots yx}_{m \text{ times}}. 
\end{align*}
Thus $(xy)^{m-1} \in \facts(w^m) \intersect \facts(v^m)$ and $(xy)^{m-1} \notin \facts(w^{m-d}) \intersect \facts(v^{m-d})$ for any $d \geq 2$ where $d \leq m$.
Consequently, the only if direction holds.

Next, we work towards showing the if direction.
In order to prove this direction, we reformulate the statement.
Let
\[ \langinf \df \{ u \in \Sigma^* \mid u \sqsubseteq w^\omega \text{ and }u \sqsubseteq v^\omega \}. \]
Clearly, if $\langinf$ is infinite, then there does not exist $n_0, m_0 \in \mathbb{N}$ such that $ \facts(w^{n_0}) \intersect \facts(v^{m_0}) = \facts(w^{n}) \intersect \facts(v^{m}) $ for all $n > n_0$ and all $m > m_0$.
We can therefore reformulate the contraposition of the if direction as: If $\langinf$ is infinite then $w$ and $v$ conjugate.
To that end, assume $\langinf$ is infinite. 
Therefore, for every $n \in \mathbb{N}$, there exists $u \in \langinf$ such that $|u| \geq n$.
Invoking the periodicity lemma, we conclude that $w$ and $v$ are conjugate. 

\proofsubparagraph{(2) if and only if (3).}
Let us first consider the if direction.
Note that for any $n_1, n_2 \in \mathbb{N}$ where $n_1 < n_2$, and any word $w \in \Sigma^*$, we have that $w^{n_1} \sqsubseteq w^{n_2}$.
Therefore, $\facts(w^{n_1}) \subseteq \facts(w^{n_2})$.

Let $w, v \in \Sigma^+$ be two primitive words such that there exists $r \in \mathbb{N}$ where 
\[r \geq \mathsf{max} \{ |w| \in \mathbb{N} \mid w \in \facts(w^n) \intersect \facts(v^m) \},\]
for all $n,m \in \mathbb{N}$.
We assume that $r \in \mathbb{N}$ is the smallest such value.
Thus, for some $n', m' \in \mathbb{N}$, we have that
\[ r = \mathsf{max} \{ |w| \in \mathbb{N} \mid w \in \facts(w^{n'}) \intersect \facts(v^{m'}) \}. \]
Furthermore, due to the fact that $\Sigma$ is a fixed and finite alphabet, the number of words of in $\Sigma^*$ with a length that is less than or equal to $r$ is finite.

We now define the language
\[  L_\mathsf{intersect} \df \{ u \in \Sigma^* \mid u \in \facts(w^n) \intersect \facts(v^m) \text{ for some } n,m \in \mathbb{N} \}. \]

Since the number of words in $\Sigma^*$ with a length  that is less than or equal to $r$ is finite, it follows that $L_\mathsf{intersect}$ is also finite.
Combining this with the fact that for any $n_1, n_2 \in \mathbb{N}$ where $n_1 < n_2$, and any word $w \in \Sigma^*$, we have that $\facts(w^{n_1}) \subseteq \facts(w^{n_2})$,
there must exist some $n_0, m_0 \in \mathbb{N}$  such that for all $n \geq n_0$ and all $m \geq m_0$, we have that
\[ \facts(w^{n_0}) \intersect \facts(v^{m_0}) = \facts(w^n) \intersect \facts(v^m). \]

Now let us consider the only if direction. 
Assume that there exists $n_0, m_0 \in \mathbb{N}$ such that $\facts(w^{n_0}) \intersect \facts(v^{m_0}) = \facts(w^m) \intersect \facts(v^m)$ for all $n> n_0$ and $m > m_0$.
Let $r \df \mathsf{max} \{ \facts(w^{n_0}) \intersect \facts(v^{m_0})  \}$.
Since for any $n> n_0$ and $m>m_0$, there does not exist $u$ such that $|u| > r$ and $u \in  \facts(w^m) \intersect \facts(v^m)$, this concludes the proof.
\end{proof}

\subsection{Proof of Lemma~\ref{lemma:FCpumping}}
\begin{proof}
Let $w_1, w_2, w_3 \in \Sigma^*$, and let $u,v \in \Sigma^+$ be two distinct primitive words such that there exists $r \in \mathbb{N}$ where
$r \geq \mathsf{max} \{ |u| \in \mathbb{N} \mid u \in \facts(u^m) \intersect \facts(v^n) \} $
for all $m,n \geq 1$, see~\cref{lemma:coprim}.
Let $f \colon \mathbb{N} \rightarrow \mathbb{N}$ be an injective function, and let $\varphi \in \fc$ such that 
\[w_1 \cdot u^p \cdot w_2 \cdot v^{f(p)} \cdot w_3 \in \lang(\varphi),\]
 for all $p \in \mathbb{N}$.

\begin{claim}\label{claim:concatenatingFixedWords}
For every $k \in \mathbb{N}$, there exists $p,q \in \mathbb{N}$ such that $p \neq q$ and $w_1 \cdot u^p \cdot w_2 \equiv_k w_1 \cdot u^q \cdot w_2$.
\end{claim}
\begin{claimproof}
For every $k \in \mathbb{N}$, there exists $p, q \in \mathbb{N}$ such that $\mathtt{a}^p \equiv_{k+3} \mathtt{a}^q$ and $p \neq q$, see~\cref{lemma:pow2}.
Thus, immediately from~\cref{lemma:primitivePower}, we know that for every $k \in \mathbb{N}$, there exists $p,q \in \mathbb{N}$ such that $p \neq q$ and $ u^p  \equiv_k u^q $.
Now note that there exists some $p_0 \in \mathbb{N}$ such that for all $p, q \geq p_0$, we have that $\facts(w_1) \intersect \facts(u^p) = \facts(w_1) \intersect \facts(u^q)$.
This simply follows from the fact that $w_1$ is a fixed word.
Thus, there exists some $r \in \mathbb{N}$ such that:
\[  r \geq \mathsf{max} \{ |w| \in \mathbb{N} \mid w \in \facts(w_1) \intersect \facts(u^p) \text{ for all } p \in \mathbb{N} \}. \] 
Therefore, there exists $p_0, r_0 \in \mathbb{N}$ such that
\[  r_0 = \mathsf{max} \{ |w| \in \mathbb{N} \mid w \in \facts(w_1) \intersect \facts(u^p) \text{ for all } p \geq p_0 \}. \] 

We know that for every $k \in \mathbb{N}$, that there exists $p, q \in \mathbb{N}$ such that $u^p \equiv_{k+ r_0 +2} u^q$, where $p \neq q$.
Trivially, $w_1 \equiv_{k + r_0 +2} w_1$ holds.
Thus, invoking~\cref{lemma:congruence}, we know: For every $k \in \mathbb{N}$, there exists $p,q \in \mathbb{N}$ such that $w_1 \cdot u^p \equiv_k w_1 \cdot u^q$, where $p \neq q$.
Using the symmetric reasoning, for every $k \in \mathbb{N}$, there exists $p, q \in \mathbb{N}$ such that $w_1 \cdot u^p \cdot w_2 \equiv_k w_1 \cdot u^q \cdot w_2$, where $p \neq q$.
\end{claimproof}

Note that from~\cref{lemma:coprim}, we know that there exists $n_0, m_0 \in \mathbb{N}$ such that for all $n > n_0$ and $m > m_0$, we have that
$\facts(u^{n_0}) \intersect \facts(v^{m_0}) = \facts(u^n) \intersect \facts(v^m)$.
Furthermore, due to the fact that $w_1$, $w_2$, and $w_3$ are fixed words, there exists $n_1, m_1 \in \mathbb{N}$ such that for all $n > n_1$ and $m > m_1$, we have that
\[
\facts(w_1 \cdot u^{n_1} \cdot w_2) \intersect \facts(v^{m_1} \cdot w_3) = \facts(w_1 \cdot u^n \cdot w_2) \intersect \facts(v^m \cdot w_3). 
\]

Let $r' \df \mathsf{max} \{ |w| \in \mathbb{N} \mid w \in \facts(w_1 \cdot u^{n_1} \cdot w_2) \intersect \facts(v^{m_1} \cdot w_3) \}$.

Trivially, we know that $v^{f(p)} \cdot w_3 \equiv_k v^{f(p)} \cdot w_3$ for all $k \in \mathbb{N}_+$.
Furthermore, we know that from~\cref{claim:concatenatingFixedWords} that for every $k \in \mathbb{N}$, there exists $p,q \in \mathbb{N}$ such that $p \neq q$ and $w_1 \cdot u^p \cdot w_2 \equiv_k w_1 \cdot u^q \cdot w_2$.

Thus, for every $k \in \mathbb{N}$, we know that there exists $p,q \in \mathbb{N}$ such that $p \neq q$ and:
\begin{itemize}
\item $w_1 \cdot u^p \cdot w_2 \equiv_{k + r' + 2} \, w_1 \cdot u^q \cdot w_2$, and 
\item $v^{f(p)} \cdot w_3 \equiv_{k + r' + 2} \, v^{f(p)} \cdot w_3$.
\end{itemize}

We can then invoke~\cref{lemma:congruence} to infer that for all $k \in \mathbb{N}_+$, we have some $p,q \in \mathbb{N}$ where $p \neq q$ and 
\[ w_1 \cdot u^p \cdot w_2 \cdot v^{f(p)} \cdot w_3 \equiv_k w_1 \cdot u^q \cdot w_2 \cdot v^{f(p)} \cdot w_3. \]
Since, we have assumed that $w_1 \cdot u^p \cdot w_2 \cdot v^{f(p)} \cdot w_3 \in \lang(\varphi)$ for all $p \in \mathbb{N}$, it follows that $w_1 \cdot u^q \cdot w_2 \cdot v^{f(p)} \cdot w_3  \in \lang(\varphi)$.
Due to the fact that $f$ is injective, $f(q) \neq f(p)$ since $q \neq p$.
\end{proof}

\subsection{Proof of Lemma~\ref{lemma:nonFClangs}}
\begin{proof}
We now show that each of the following languages are not expressible in $\fc$:
\begin{itemize}
\item $L_1 \df \{ \mathtt{a}^n \cdot (\mathtt{ba})^n  \mid n \in \mathbb{N} \}$,
\item $L_2 \df \{ \mathtt{a}^i \cdot (\mathtt{ba})^j \mid 1 \leq i \leq j  \}$,
\item $L_3 \df \{ \mathtt{b}^n \cdot \mathtt{a}^m \cdot \mathtt{b}^{n+m} \mid m,n \in \mathbb{N} \}$,
\item $L_4 \df \{ \mathtt{b}^n \cdot \mathtt{a}^m \cdot \mathtt{b}^{n \cdot m} \mid m,n \in \mathbb{N} \}$, and 
\item $L_5 \df \{ (\mathtt{abaabb})^m \cdot (\mathtt{bbaaba})^m \mid m \in \mathbb{N} \}$, and
\item $L_6 \df \{ \mathtt{a}^n \cdot \mathtt{b}^n \cdot (\mathtt{ab})^n \mid n \in \mathbb{N} \}$.
\end{itemize}

\underline{$L_1$ and $L_2$:} The fact that $L_1, L_2 \notin \lang(\fc)$ follows from~\cref{prop:subword} and~\cref{obs:equivToLang}

\underline{$L_3$:} For $L_3$, consider the case where $n=0$.
We therefore get the language $\{ \mathtt{a}^m \cdot \mathtt{b}^{m} \mid m \in \mathbb{N} \}$.
From~\cref{example:anbn}, we know that for every $k \in \mathbb{N}$, there exists $p, q \in \mathbb{N}$ such that $q \neq p$ and $\mathtt{a}^q \cdot \mathtt{b}^p \equiv_k \mathtt{a}^p \cdot \mathtt{b}^p$.
Thus, invoking~\cref{obs:equivToLang}, we conclude that $L_3 \notin \lang(\fc)$.

\underline{$L_4$:}
For $L_4$, let $n=1$. We therefore get the language $\{ \mathtt{b} \cdot \mathtt{a}^m \cdot \mathtt{b}^m \mid m \in \mathbb{N} \}$.
From~\cref{lemma:pow2}, we know that for every $k \in \mathbb{N}$, there exists $p,q \in \mathbb{N}$ such that $p \neq q$ and $\mathtt{b}^p \equiv_{k+3} \mathtt{b}^q$.
It trivially holds that $\mathtt{b} \cdot \mathtt{a}^p \equiv_{k+3} \mathtt{b} \cdot \mathtt{a}^p$.
As $\facts(\mathtt{b}^m) \intersect \facts(\mathtt{b} \cdot \mathtt{a}^n) = \{ \emptyword, \mathtt{b} \}$ for any $n,m \in \mathbb{N}$, we can invoke~\cref{lemma:congruence} with $r=1$ and conclude that $\mathtt{b} \cdot \mathtt{a}^p \cdot \mathtt{b}^p \equiv_{k} \mathtt{b} \cdot \mathtt{a}^p \cdot \mathtt{b}^q$.
Thus, we know from~\cref{obs:equivToLang} that $L_4 \notin \lang(\fc)$.

\underline{$L_5$:} It is clear that $\mathtt{abaabb}$ and $\mathtt{bbaaba}$ are both co-primitive.
Thus, there exists some $r \in \mathbb{N}$ where $r \geq  \mathsf{max}\{ |u| \in \mathbb{N} \mid u \in \facts((\mathtt{abaabb})^m) \intersect \facts((\mathtt{bbaaba})^n) \}$ for all $m,n \geq 1$.
Therefore, we can invoke~\cref{lemma:pumpingConsequence} to conclude that $L_5 \notin \lang(\fc)$.

\underline{$L_6$:} We know that for every $k \in \mathbb{N}$, there exists $p, q \in \mathbb{N}$ such that $\mathtt{a}^p \cdot \mathtt{b}^p \equiv_{k+2} \mathtt{a}^q \cdot \mathtt{a}^p$ and $p \neq q$.
Furthermore, it is trivial that $(\mathtt{ab})^p \equiv_{k+2} (\mathtt{ab})^p$ for all $k \in \mathbb{N}$.

Since for any $i,j,l \in \mathbb{N}_+$, it holds that $\facts(\mathtt{a}^i \cdot \mathtt{b}^j) \intersect \facts( (\mathtt{ab})^l ) = \{ \emptyword, \mathtt{a}, \mathtt{b}, \mathtt{ab} \}$, we can invoke~\cref{lemma:congruence} with $r = 2$.
Thus, for every $k \in \mathbb{N}$, there exists $p, q \in \mathbb{N}_+$ with $p \neq q$, and
$ \mathtt{a}^q \cdot \mathtt{b}^p \cdot (\mathtt{ab})^p \equiv_k \mathtt{a}^p \cdot \mathtt{b}^p \cdot (\mathtt{ab})^p$. 
Therefore, $L_6 \notin \lang(\fc)$.
\end{proof}

\section{Appendix for Section 5}

\subsection{Proof of Lemma~\ref{lemma:bounded}}
\begin{proof}
The proof of this result closely follows the proof of Theorem 6.2 from~\cite{fre:splog}.
We prove this result for completeness sake, as Theorem 6.2 from~\cite{fre:splog} does not directly deal with $\fc$ and $\fc[\reg]$.
In this proof, we only show that if $L \subseteq \Sigma^*$ is a bounded language, then $L \in \lang(\fc)$ if and only if $L \in \lang(\fcreg)$.
Immediately from the fact that $\fc$ has conjunction, disjunction, and negation, we are able to lift this to Boolean combinations of bounded languages.

The only if direction is trivial, and therefore this proof focuses on the if direction.
That is, if $L \in \lang(\fc[\reg])$, and $L$ is bounded, then $L \in \lang(\fc)$.
Let $\varphi \in \fc[\reg]$ such that $\lang(\varphi) = L$, where $L$ is bounded.
Most of the proof of this direction follows from the following claim:
\begin{claim}\label{claim:boundedConstraints}
If $\gamma$ is a regular expression that generates a bounded regular language, then there exists $\varphi \in \fc$ such that $\fun{\varphi}(w) = \fun{x \regconst \gamma}(w)$ for all $w \in \Sigma^*$.
\end{claim}
\begin{claimproof}
From Theorem 1.1 from~\cite{ginsburg1966bounded}, we know that the class of bounded regular languages are exactly the class of languages that contain all finite languages, all languages $w^*$ where $w \in \Sigma^*$, and is closed under finite union and finite concatenation.
We show that for any $w \in \Sigma^*$, the language $w^*$ can be defined in $\fc$.

As syntactic sugar, we shall use atomic formulas that contain arbitrary concatenation.
It is clear that this can be rewritten into an $\fc$ formula that only contains binary concatenation (for example, see~\cite{freydenberger2021splitting}).
Now, consider
\[  \varphi_{w^*}(x) \df \exists x \colon \Bigl(  (x \logeq \emptyword) \lor \bigl( (x \logeq w \cdot z) \land (x \logeq z \cdot w)  \bigr) \Bigr). \]
Thus, for any interpretation $\mathcal{I} \df (\mathfrak{A}_w, \subs)$ such that $\mathcal{I} \models \varphi_{w^*}$, we have that either $\subs(x) = \emptyword$, or $\subs(x) = w \cdot u$ and $\subs(x) = u \cdot w$ for some $u \in \Sigma^*$.
Proposition 1.3.2 in Lothaire~\cite{lothaire1997combinatorics} states that if $vu = uv$ for $u,v \in \Sigma^*$, then there is some $z \in \Sigma^*$ and $k_1,k_2 \in \mathbb{N}$ such that $u = z^{k_1}$ and $v = z^{k_2}$.
Thus, if $\subs(x) = w \cdot u$ and $\subs(x) = u \cdot w$ for some $u \in \Sigma^*$, then $\subs(x) = w^k$ for some $k \in \mathbb{N}_+$.

Finite languages, finite union, and finite concatenation can easily be dealt with.
Thus, for any $\fc[\reg]$-formula of the form $x \regconst \gamma$, where $\gamma$ generates a bounded regular language, there exists $\varphi \in \fc$ such that $\fun{\varphi}(w) = \fun{x \regconst \gamma}(w)$ for all $w \in \Sigma^*$.
\end{claimproof}

We now construct $\psi \in \fc$ such that $\lang(\psi) = \lang(\varphi)$.
Note that every set of factors of a bounded language is also bounded (for example, see Lemma 5.1.1 in Ginsburg~\cite{ginsburg:CFL}).
Thus, since $\lang(\varphi)$ is bounded, it follows that $\lang(\varphi) \subseteq B$, where $B \df w_1^* \cdots w_n^*$ for some $w_1, \dots, w_n \in \Sigma^*$.
Therefore, for any $w \in \Sigma^*$ and any $x \in \var(\varphi)$, if $\subs \in \fun{\varphi}(w)$, then $\subs(x) \in L_x$, where 
\[ L_x \subseteq \{ u \in \Sigma^* \mid u \sqsubseteq v \text{ and } v \in B \} \]
is a bounded regular language.
Therefore, we can replace any atomic formulas of the form $(x \regconst \gamma)$, with $(x \regconst \gamma')$, where $\lang(\gamma') = \lang(\gamma) \intersect L_x$.
Furthermore, it follows that $\lang(\gamma')$ is a bounded regular language.
Then, we can invoke~\cref{claim:boundedConstraints} to replace $(x \regconst \gamma')$ with $\varphi_{\gamma'}(x) \in \fc$ without altering the semantics of the formula.

After carrying out this process of replacing all regular constraints, we have a formula $\psi \in \fc$ that is equivalent to $\varphi$.
\end{proof}

\subsection{Proof of Theorem~\ref{thm:relations}}
\begin{proof}
First, let $\varphi_w(\strucvar) \df \neg \exists z_1, z_2 \colon \Bigl( \neg (z_2 \logeq \emptyword) \land \bigl( (z_1 \logeq z_2 \cdot x) \lor (z_1 \logeq x \cdot z_2)  \bigr)  \Bigr)$.
It is clear that for every interpretation $(\mathfrak{A}_w, \subs)$, where $\mathfrak{A}_w$ is a $\signature_\Sigma$-structure that represents $w \in \Sigma^*$, and where $(\mathfrak{A}_w, \subs) \models \varphi_w(\strucvar)$, we have that $\subs(\strucvar) = w$.

Note that for this proof, we use arbitrary concatenation as a shorthand.
For example, $(x \logeq y_1 \cdot y_2 \cdot y_3)$ is shorthand for $(x \logeq y_1 \cdot z_1) \land (z_1 \logeq y_2 \cdot y_3)$ where $z_1$ is a new and unique variable.
Also note that all the languages given in~\cref{lemma:nonFClangs} are bounded languages, therefore $L_i \notin \lang(\fcreg)$ for $i \in [6]$.

Assume that some relation $R$ from $\mathsf{Num}_\mathtt{a}$, $\mathsf{Add}$, $\mathsf{Mult}$, $\mathsf{Scatt}$, $\mathsf{Perm}$, $\mathsf{Rev}$, $\mathsf{Shuff}$ is definable in $\fc[\reg]$, and let $\varphi_R \in \fc[\reg]$ define $R$.
Now, consider the following $\fc$-formulas:
\begin{align*}
\psi_1 & \df \exists \strucvar, x,y \colon \bigl( \varphi_w(\strucvar) \land (\strucvar \logeq x  y) \land (x \regconst \mathtt{a}^*) \land (y \regconst (\mathtt{ba})^*)  \land \varphi_{\mathsf{Num}_\mathtt{a}}(x,y)\bigr) , \\
\psi_2 & \df \exists \strucvar, x,y \colon \bigl( \varphi_w(\strucvar) \land  (\strucvar \logeq x  y) \land (x \regconst \mathtt{a}^*) \land (y \regconst (\mathtt{ba})^*) \land \varphi_\mathsf{Scatt}(x,y) \bigr), \\
\psi_3 & \df \exists \strucvar, x,y,z \colon \bigl( \varphi_w(\strucvar) \land  (\strucvar \logeq x  y  z) \land (x \regconst \mathtt{b}^*) \land (y \regconst \mathtt{a}^*) \land (z \regconst \mathtt{b}^*)  \land \varphi_\mathsf{Add}(x,y,z)  \bigr), \\
\psi_4 & \df \exists \strucvar, x,y,z \colon \bigl( \varphi_w(\strucvar) \land  (\strucvar \logeq x  y  z)  \land (x \regconst \mathtt{b}^*) \land (y \regconst \mathtt{a}^*) \land (z \regconst \mathtt{b}^*) \land \varphi_\mathsf{Mult}(x,y,z) \bigr), \\
\psi_5 &  \df \exists \strucvar, x,y \colon \bigl( \varphi_w(\strucvar) \land  (\strucvar \logeq x  y) \land (x \regconst (\mathtt{abaabb})^* ) \land (y \regconst (\mathtt{bbaaba})^*) \land \varphi_\mathsf{Perm}(x,y) \bigr), \\
\psi_5' & \df \exists \strucvar, x,y \colon \bigl( \varphi_w(\strucvar) \land  (\strucvar \logeq x  y) \land (x \regconst (\mathtt{abaabb})^* ) \land (y \regconst (\mathtt{bbaaba})^*) \land \varphi_\mathsf{Rev}(x,y) \bigr), \\
\psi_6 & \df \exists \strucvar, x, y, z \colon \bigl( \varphi_w(\strucvar) \land (\strucvar \logeq x y z) \land (x \regconst \mathtt{a}^+) \land (y \regconst \mathtt{b}^+) \land \mathsf{Shuff}(x,y,z)  \bigr).
\end{align*}
It follows that $\lang(\psi_i) = L_i$ for $i \in [6]$, where $L_i$ is defined in~\cref{lemma:nonFClangs}, and $\lang(\psi_5') = L_5$.
Thus the relations $\mathsf{Num}_\mathtt{a}$, $\mathsf{Add}$, $\mathsf{Mult}$, $\mathsf{Scatt}$, $\mathsf{Perm}$, $\mathsf{Rev}$, and $\mathsf{Shuff}$ are not definable in $\fc[\reg]$

Now, let $\Sigma \df \{ \mathtt{a}, \mathtt{b} \}$.
To conclude the proof, assume that for the morphism $h \colon \Sigma^* \rightarrow \Sigma^*$, where $h(\mathtt{a}) = \mathtt{b}$ and $h(\mathtt{b}) = \mathtt{b}$, the relation $\mathsf{Morph}_h$ is definable in $\fcreg$.
Then, let $\psi \df \exists \strucvar, x, y \colon \bigl( \varphi_w(\strucvar) \land (\strucvar \logeq x \cdot y) \land (x \regconst \mathtt{a}^*) \land \mathsf{Morph}_h(x,y) \bigr)$.
It is clear that $\lang(\psi) = \{ \mathtt{a}^n \mathtt{b}^n \mid n \in \mathbb{N} \}$, and $\lang(\psi) \notin \lang(\fcreg)$ (see~\cite{frey2019finite} or~\cref{example:anbn} along with~\cref{lemma:bounded}).
Thus, $\mathsf{Morph}_h$ is not definable in $\fcreg$.
\end{proof}

\end{document}